\documentclass[sigplan,9pt]{acmart}%\settopmatter{printfolios=true,printccs=true,printacmref=false}
\acmConference[]{}{}{}
\usepackage{hyperref}
\usepackage{float}
\usepackage{color}
\usepackage[parfill]{parskip}
\usepackage{amsmath,amssymb,latexsym,stmaryrd}
\newcommand{\ignore}[1]{}

\usepackage{graphicx}

\usepackage{framed}

\newcommand{\bool}{\mathrm{\bf Bool}}

 \newcommand{\wfcon}{\true}
  \newcommand{\pointt}{\mathbf{Point}}
\newcommand{\sett}{\mathbf{Set}}

\newcommand{\true}{\mathrm{\bf True}} \newcommand{\false}{\mathrm{\bf
    False}} \newcommand{\type}{\mathbf{Type}}
\newcommand{\class}{\mathbf{Class}}
\newcommand{\setof}{\mathbf{SetOf}}
\newcommand{\classof}{\mathbf{ClassOf}}
\newcommand{\carrier}{\mathbf{Carrier}}

 \newcommand{\mem}{\mathbf{Mem}}

\newcommand{\double}[1]{\left\llbracket #1 \right\rrbracket}
\newcommand{\semvalue}[1]{{\mathcal V}_{\Gamma}\left\llbracket #1
  \right\rrbracket}  \newcommand{\subvalue}[2]{{\mathcal
    V}_{#1}\left\llbracket #2 \right\rrbracket}
\newcommand{\deppair}{_{\intype{x\;}{\;\sigma}}\;\tau[x]}
\newcommand{\subxvalue}[1]{{\mathcal
    V}_{\Gamma;\;\intype{x\;}{\;\sigma}}\double{#1}}
\newcommand{\convalue}[1]{{\mathcal V}\left\llbracket #1 \right\rrbracket}
\newcommand{\tempvalue}[1]{\tilde{\mathcal V}_{\Gamma}\double{#1}}
\newcommand{\subtempvalue}[1]{\tilde{\mathcal V}_{\Gamma;\;\intype{x\;}{\;\sigma}}\left\llbracket #1  \right\rrbracket}

\newcommand{\intype}[2]{#1\!:\!#2}

\newcommand{\incat}[2]{#1\!\in\!#2}

\def\unnamed#1#2{
\parbox{2.0in}{
\begin{tabbing}
\hspace{1em}\= #1 \\ \> \parbox{.75in}{\noindent \hrule ~} #2
\end{tabbing}
}}

\def\ant#1{\\ \> $#1$}

 \newcommand{\leftop}{\mathbf{Left}}
\newcommand{\rightop}{\mathbf{Right}} 
\newcommand{\pointify}{\mathbf{Pnt}}

\newcommand{\subpoint}{\mathbf{SubPnt}}

\newcommand{\domop}{\mathbf{Dom}}

\newcommand{\inv}{{-1}}

\newcommand{\rdash}{\vdash}

\begin{document}
\title{Isomorphism within Set-Theoretic Type Theory}
\author{David McAllester}
\affiliation{
  \institution{TTI-Chicago}            %% \institution is required
  }
\begin{abstract}
  We provide a treatment of isomorphism within a set-theoretic formulation of dependent type theory.
  Type expressions are assigned their natural set-theoretic compositional meaning.
  Types are divided into small and large types --- sets and proper classes respectively.
  Each proper class, such as ``group'' or ``topological space'', has an associated notion of isomorphism in correspondence with standard definitions.
  Isomorphism is handled by defining a groupoid structure on the space of all definable values.
  The values are simultaneously objects (oids) and morphism --- they are ``morphoids''.
  Soundness is proved for simple and natural inference rules
  for deriving isomorphisms and for the substitution of isomorphics.

\end{abstract}

\maketitle

%\setcounter{tocdepth}{3}

%\vfill
%\tableofcontents
%\vfill

\section{Introduction}

Unlike classical set theory, a type theoretic foundation for mathematics imposes strict grammatical constraints on the well formed expressions.
These grammatical constraints are central to the concept of isomorphism.
Isomorphism is related to the notion of
an application programming interface (API) in computer software.  An
API specifies what information and behavior an object provides.  Two
different implementations can produce identical behavior when
interaction is restricted to that allowed by the API.  For example
textbooks on real analysis typically start from axioms involving
multiplication, addition, and ordering.  Addition, multiplication and
ordering define an abstract interface --- the well-formed statements
about real numbers are limited to those that can be defined in terms
of the operations of the interface.  The real numbers can be implemented as either Dedekind cuts or as Cauchy sequences.  However,
these implementations provide the same behavior as viewed through the allowed interface --- the two different implementations (or representations) are
isomorphic as ordered fields.
Grammatical well-formedness restricts access
to a particular interface.

{\bf Isomorphism and Dependent Pair Types.} The general notion of isomorphism is best illustrated by
considering dependent pair types.  A dependent pair type is typically written as
$\Sigma_{\intype{x\;}{\;\sigma}}\;\tau[x]$ where the instances of this type are the pairs $(x,y)$
where $x$ is an instance of the type $\sigma$ and $y$ is an instance of the type $\tau[x]$.  A more
transparent notation for this type might be
$\mathbf{PairOf}(\intype{x}{\sigma},\;\intype{y}{\tau[x]})$.  But we will stay with the conventional
notation $\Sigma_{\intype{x\;}{\;\sigma}}\;\tau[x]$.  The type of directed graphs can be written as
$\Sigma_{\intype{\mathcal N\;}{\;\mathbf{Set}}}\; ({\mathcal N} \times {\mathcal N}) \rightarrow \bool$.  The
instances of this type are the pairs $({\mathcal N}, P)$ where ${\mathcal N}$ is the set of nodes of the
graph and $P$ is a binary predicate on the nodes giving the edge relation.  Two directed graphs
$({\mathcal N},P)$ and $({\mathcal M},Q)$ are isomorphic if there exists a bijection from ${\mathcal N}$ to
${\mathcal M}$ that carries $P$ to $Q$.  Some bijections will carry $P$ to $Q$ while others will not.
As discussed in section~\ref{sec:semantics1}, the type ``group'' and the type ``topological space''
can each be written as a subtype of a dependent pair type.

{\bf Observational Equivalence.}  Intuitively, isomorphic objects are
``observationally equivalent'' --- they have the same observable
properties when access to the objects is restricted to those
operations allowed by the type system.  In the type theory developed here this observational
equivalence is expressed by the following inference rule for the substitution of isomorphics.

\centerline{ \unnamed {\ant{\Gamma \rdash
    \intype{\sigma,\tau}{\class}} \ant{\Gamma;\intype{x}{\sigma}
    \rdash \intype{e[x]}{\tau}} \ant{\Gamma \rdash u =_\sigma w}} {
  \ant{\Gamma \rdash e[u] =_\tau e[w]}}}

Here $=_\sigma$ is the isomorphism relation associated with the type
$\sigma$ and similarly for $=_\tau$.  As an example suppose we have
$\intype{G}{\mathbf{Graph}} \rdash \intype{\Phi[G]}{\bool}$.
Intuitively this states that $\Phi[G]$ is a graph-theoretic
property. For $G =_{\mathbf{Graph}} H$ we then have $\Phi[G] \Leftrightarrow \Phi[H]$.
We do not require that $\Phi[G]$ is a first
order formula. For example, $\Phi[G]$ might state that the spectrum of
the graph Laplacian of $G$ contains a gap of size $\lambda$.  As
another example we might have $\intype{G}{\mathbf{Topology}} \rdash
\intype{e[G]}{\mathbf{Group}}$ where $e[G]$ might be an expression
denoting the fundamental group of a topological space.  The
substitution rule then says that isomorphic topological spaces have
isomorphic fundamental groups.

{\bf Voldemort's Theorem.}  There are many situations in mathematics
where a type can be shown to be inhabited (have a member) even though
there is no natural or canonical member of that type.  A minimal
example is that there is no natural or canonical member of an abstract
two element set.  Another intuitive example is that there is no
natural or canonical point on a geometric circle.  A vector space has
no natural or canonical basis (coordinate system).  For every finite
dimensional vector space $V$ there is an isomorphism (a linear
bijection) between $V$ and its dual $V^*$.  However, there is no
natural or canonical isomorphism --- different choices of coordinates
lead to different isomorphisms.  Voldemort's theorem states that if no
natural or canonical element of a type exists then no well-typed
expression can name an element.

{\bf Cryptomorphism.}  Two types $\sigma$ and $\tau$ are cryptomorphic
in the sense of Birkoff and Rota \cite{Rota} if they ``present the
same data''.  For example a group can be defined as a four-tuple of a
set, a group operation, an identity element and an inverse operation
satisfying certain equations.  Alternatively, a group can be defined
as a pair of a set and a group operation such that an identity element
and an inverse elements exist.  These are different types with
different elements (four-tuples vs.$\!$ pairs).  However, these two
types present the same data.  Rota was fond of pointing out the large
number of different ways one can formulate the concept of a
matroid. Any type theoretic foundation for mathematics should account
formally for this phenomenon.  Here we take two types $\sigma$ and
$\tau$ to be crytomorphic if there exist natural functions
$\intype{f}{\sigma \rightarrow \tau}$ and $\intype{g}{\tau \rightarrow
  \sigma}$ such that $f \circ g$ and $g \circ f$ are the identity
functions on $\tau$ and $\sigma$ respectively and where a ``natural function''
from $\sigma$ to $\tau$ is one defined by an expression $e[x]$ where we have
$\Gamma;\;\intype{x}{\sigma} \vdash \intype{e[x]}{\tau}$.

{\bf Naive Semantics vs. MLTT and HoTT.} Martin L\"{o}f type theory (MLTT) \cite{MLTT1971} is a constructive type theory
continuing the intuitionist tradition of Brouwer \cite{Brouwer}.  As is often pointed out, one can add axioms to constructive type theory allowing
classical reasoning.  However, we argue here that an up-front commitment to classical set-theoretic foundations
greatly simplifies the semantics of type theory.

Homotopy type theory (HoTT) \cite{HOTT,SimpSets}, and the earlier groupoid model \cite{GRPD}, give models of MLTT treating isomorphism.
Both models are based on functorial rather than compositional semantics.
In functorial semantics the sequent $\Gamma \vdash \intype{e}{\sigma}$ is interpreted
as a ``dependent functor'' from the groupoid of
interpretations of $\Gamma$ to the family of groupoids that are the
different interpretations of $\sigma$ under different objects (variable interpretations) in the
groupoid denoted by $\Gamma$.
In contrast, naive semantics simply associates each context $\Gamma$ with a set of variable interpretations
satisfying the type declarations and assumptions in $\Gamma$.  Each type expression simply denotes a collection of values.
The meaning of a sequent $\Gamma \vdash \intype{e}{\sigma}$ is
simply that for any variable interpretation $\rho$ satisfying $\Gamma$ we have
$\semvalue{e}\rho \in  \semvalue{\sigma}\rho$ where $\semvalue{e}\rho$ and $\semvalue{\sigma}\rho$ are the naive values
of $e$ and $\sigma$ respectively under the naive variable interpretation $\rho$.

Another difference is that functorial semantics requires that each value be assigned to a specific category.
In contrast, naive semantics allows the same value to be in multiple classes.
An Abelian group is both an Abelian group and a group. A permutation group is both a permutation group and a group.

Naive type theory is designed to make formal mathematics as similar to naturual (naive) mathematics as possible.

\ignore{
{\bf Minimalism and Extensions.}  The version of type theory presented in
this paper is limited to a minimal set of features adequate to
explicate isomorphism in terms of compositional semantics and
value-based groupoid structure.  As briefly discussed in section~\ref{sec:Conclusions},
various important extensions are
possible.  However, explicating them here would distract from the central
ideas.
}

\section{Syntax and Semantics}
\label{sec:semantics1}

Figure~\ref{fig:constructs} lists the constructs of the type theory and figure~\ref{fig:semantics} gives their formal (but naive) semantics.
The
constructs listed in figure~\ref{fig:constructs} correspond to
those of MLTT with the exception that the naive system uses Boolean propositions rather
than propositions as types.  As in MLTT, we write $\sigma \rightarrow \tau$ for the type $\Pi_{\intype{x\;}{\;\sigma}}\;\tau$ where $x$ does not occur free in $\tau$
and we wreite $\sigma \times \tau$ for the type $\Sigma_{\intype{x\;}{\;\sigma}}\;\tau$ with $x$ not free in $\tau$.
We include a third dependent type construct $S_{\intype{x\;}{\;\sigma}}\;\Phi[x]$ which denotes the type of values $x$ in $\sigma$
such that $\Phi[x]$ is true. Absolute equalities (judgemental equalities) $x \doteq y$ have truth values but are not Boolean expressions.
Isomorphism equalities (propositional equalities) $x =_\sigma y$ are Boolean expresions.  This distinction is needed because absolute equalities
do not allow the substitution of isomorphics as is required for the observational equivalence rule stated in the introduction.  This is discussed in more detail below.

\begin{figure}
\begin{framed}
{\small
  $$
  \begin{array}{lccccccc}
    \mbox{pairs} & (e_1,e_2) & \pi_1(e) & \pi_2(e) \\ \\
    \mbox{variables,\;functions} & x & \lambda \intype{x}{\sigma}\;e[x] & f(e) \\ \\
    \mbox{dependent types} & \Sigma_{\intype{x\;}{\;\sigma}}\;\tau[x] & \Pi_{\intype{x\;}{\;\sigma}}\;\tau[x] &  S_{\intype{x\;}{\;\sigma}}\;\Phi[x] \\ \\
    \mbox{Boolean expressions} & P(e) & e_1 =_\sigma e_2 & \neg \Phi \\ \\
      &  \Phi_1 \vee \Phi_2  & \forall \intype{x}{\sigma}\; \Phi[x]  \\ \\
    \mbox{type constants} & \bool & \sett & \class \\ \\ \mbox{contexts} &
    \epsilon & \Gamma;\intype{x}{\tau} & \Gamma;\Phi \\ \\
    \mbox{sequents} & \Gamma \rdash \intype{e}{\sigma} & \Gamma \vdash e_1 \doteq e_2 & \Gamma\rdash \Phi
  \end{array}
  $$ }

\caption{{\bf Constructs of Naive Type Theory}}
  \label{fig:constructs}

\end{framed}
\end{figure}
\begin{figure}
\begin{framed}

{\small
    
\parbox{3.0 in}{(1) $\epsilon$. We define $\convalue{\epsilon}$ to be the
set containing the empty variable interpretation.}

\bigskip
\parbox{3.0in}{ (2) $\bool$. For $\convalue{\Gamma}$ defined and $\rho \in
\convalue{\Gamma}$ we define $\semvalue{\bool}\rho$ to be the set
containing the two values $\true$ and $\false$.}

\bigskip
\parbox{3.0in}{ (3) $\sett$ and $\class$. For $\convalue{\Gamma}$ defined
and $\rho \in \convalue{\Gamma}$ we define $\semvalue{\sett}\rho$ and
$\semvalue{\class}\rho$ to be the collection of all sets and
the collection of all classes respectively (see section~\ref{sec:semantics2}).}

\bigskip
\parbox{3.0in}{ (4) $\Gamma \models \intype{e}{\sigma}$.
If $\convalue{\Gamma}$, $\semvalue{e}$ and $\semvalue{\sigma}$ are defined then
we define $\Gamma \models \intype{e}{\sigma}$ to mean that for $\rho \in \convalue{\Gamma}$ we have
$\semvalue{e}\rho \in \semvalue{\sigma}\rho$.}

\bigskip
\parbox{3.0in}{ (5) $\Gamma;\intype{x}{\tau}$. For $\Gamma \models
\intype{\tau}{\class}$ and $x$ not declared in $\Gamma$ we define
$\convalue{\Gamma;\intype{x}{\tau}}$ to be the set of variable
interpretations of the form $\rho[x \leftarrow v]$ for $\rho \in
\convalue{\Gamma}$ and $v \in \semvalue{\tau}\rho$.}

\bigskip
\parbox{3.0in}{ (6) $\Gamma;\Phi$. For $\Gamma \models \intype{\Phi}{\bool}$
we define $\convalue{\Gamma;\Phi}$ to be the set of all $\rho \in
\convalue{\Gamma}$ such that $\semvalue{\Phi}\rho = \true$.}

\bigskip
\parbox{3.0in}{ (7) $\Gamma \models \Phi$.  For $\Gamma \models \intype{\Phi}{\bool}$
we define $\Gamma \models \Phi$ to mean that for $\rho \in \convalue{\Gamma}$ we have $\semvalue{\Phi}\rho = \true$.}

\bigskip
\parbox{3.0in}{ (8) $x.$ For $x$ declared in $\Gamma$ we define
$\semvalue{x}\rho$ to be $\rho(x)$.}

\bigskip
\parbox{3.0in}{ (9) $\lambda \intype{x}{\sigma}\;e[x]$.  For $\Gamma \models \intype{\sigma}{\sett}$
and $\Gamma;\intype{x}{\sigma} \models \intype{e[x]}{\sett}$ we define
$\semvalue{\lambda \intype{x}{\sigma}\;e[x]}\rho$ to be the function
(set of pairs) mapping $v \in \semvalue{\sigma}\rho$ to
$\subvalue{\Gamma;\intype{x\;}{\;\sigma}}{e[x]}\rho[x \leftarrow v]$.}

\bigskip
\parbox{3.0in}{ (10) If $\semvalue{f}$ and $\semvalue{e}$ are
defined, and for $\rho \in \convalue{\Gamma}$ we have
$\semvalue{f}\rho$ is a function (set of pairs) with $\semvalue{e}\rho$ in its
domain, then $\semvalue{f(e)}\rho$ is defined to be
$\semvalue{f}\rho$ applied to $\semvalue{e}\rho$.}

\bigskip
\parbox{3.0in}{ (11) $(u,w)$. For $\semvalue{u}$ and $\semvalue{w}$ defined,
we define $\semvalue{(u,w)}\rho$ to be
$(\semvalue{u}\rho,\;\semvalue{w}\rho).$}

\bigskip
\parbox{3.0in}{ (12) $\pi_i(e)$. If $\semvalue{e}$ is defined, and for all
$\rho \in \convalue{\Gamma}$ we have that $\semvalue{e}{\rho}$ is a
pair, then $\semvalue{\pi_i(e)}\rho$ is defined to be
$\pi_i(\semvalue{e}\rho).$}

\bigskip
\parbox{3.0in}{ (13) $\Gamma \models s \doteq w$. For $\convalue{\Gamma}$, $\semvalue{s}$ and $\semvalue{w}$ defined
we define $\Gamma \models s \doteq w$ to mean that for $\rho \in \convalue{\Gamma}$ we have 
$\semvalue{s}\rho = \semvalue{w}\rho$.}

\bigskip
\parbox{3.0in}{ (14) $\Phi \vee \Psi$. For $\Gamma \models
\intype{\Phi}{\bool}$ and $\Gamma \models \intype{\Psi}{\bool}$ we
define $\semvalue{\Phi \vee \Psi}\rho$ to be $\semvalue{\Phi}{\rho}
\vee \semvalue{\rho}{\rho}.$}

\bigskip
\parbox{3.0in}{ (15) $\neg \Phi$.  For $\Gamma \models \intype{\Phi}{\bool}$
we define $\semvalue{\neg \Phi}{\rho}$ to be
$\neg\semvalue{\Phi}{\rho}.$}

\bigskip
\parbox{3.0in}{ (16) $\forall\; \intype{x}{\sigma}\;\Phi[x]$. For $\Gamma \models \intype{\sigma}{\class}$ and
$\Gamma;\;\intype{y}{\sigma} \models \intype{\Phi[y]}{\bool}$ we
define $\semvalue{\forall\; \intype{x}{\sigma}\;\Phi[x]}\rho$ to be
$\true$ if for all $v \in \semvalue{\sigma}\rho$ we have
$\subvalue{\Gamma;\intype{y\;}{\;\sigma}}{\Phi[y]}\rho[y \leftarrow v]
= \true$.}

\bigskip\parbox{3.0in}{ (17) $\Sigma_{\intype{x\;}{\;\sigma}} \;\tau[x]$.
For $\Gamma \models \intype{\sigma}{\class}$ and $\Gamma;\intype{x}{\sigma} \models \intype{\tau[x]}{\class}$ we
define $\semvalue{\Sigma_{\intype{x\;}{\;\sigma}} \;\tau[x]}\rho$ to
be the collection of pairs $(v,w)$ with $v \in \semvalue{\sigma}\rho$
and $w \in \subvalue{\Gamma;\;\intype{x\;}{\;\sigma}}{\tau[x]}\rho[x
  \leftarrow v]$.
}

\bigskip
\parbox{3.0in}{ (18) $\Pi_{\intype{x\;}{\;\sigma}} \;\tau[x]$.
For $\Gamma \models \intype{\sigma}{\sett}$ and $\Gamma;\intype{x}{\sigma} \models
\intype{\tau[x]}{\sett}$ we define
$\semvalue{\Pi_{\intype{x\;}{\;\sigma}}\;\tau[x]}\rho$ to be the set
of all functions $f$ with domain $\semvalue{\sigma}\rho$ such that for
all $v \in \semvalue{\sigma}\rho$ we have $f(v) \in
\subvalue{\Gamma;\intype{x\;}{\;\sigma}}{\tau[x]}\rho[x \leftarrow
  v]$.
}

\bigskip
\parbox{3.0in}{ (19) $S_{\intype{x\;}{\;\sigma}}\;\Phi[x]$.
For $\Gamma \models \intype{\sigma}{\class}$ and $\Gamma;\;\intype{y}{\sigma} \models \intype{\Phi[y]}{\bool}$ we
define $\semvalue{S_{\intype{x\;}{\;\sigma}}\;\Phi[x]}\rho$ to be the
collection of all $v \in \semvalue{\sigma}\rho$ such that
$\subvalue{\Gamma;\;\intype{y\;}{\;\sigma}}{\Phi[y]}\rho[y \leftarrow
  v] = \true$.}

\bigskip
\parbox{3.0in}{(20) $s =_\sigma w$. For $\Gamma \models
\intype{\sigma}{\class}$ and $\Gamma \models \intype{s}{\sigma}$ and
$\Gamma \models \intype{w}{\sigma}$ we define $\semvalue{s =_\sigma
  w}\rho$ to be $\true$ if $\semvalue{s}\rho \;
=_{\semvalue{\sigma}\rho}\;\semvalue{w}\rho$ (see the bottom of figure~\ref{fig:classes}).}  }

\caption{{\bf Naive Semantics.}}
\label{fig:semantics}
\end{framed}
\end{figure}

Figure~\ref{fig:semantics} specifies both which expressions are well
formed (grammatical) as well as the meaning of well formed
expressions.  This is done by specifying a partial semantic value
function. A context $\Gamma$ is well formed if and only if
$\convalue{\Gamma}$ is defined.  Similarly, an expression $e$ is well
formed (is grammatical) in context $\Gamma$ if and only if
$\semvalue{e}$ is defined.  For a well formed context $\Gamma$ we have
that $\convalue{\Gamma}$ is the set of interpretations of the
variables declared in $\Gamma$ satisfying both the type declarations
and the Boolean assumptions in $\Gamma$. If $\semvalue{e}$ is defined then
for any $\rho \in \convalue{\Gamma}$ we have that $\semvalue{e}\rho$ is the value of the
expression $e$ under variable interpretation $\rho$.
As an example, clause (5) specifies when the
context $\Gamma;\;\intype{x}{\sigma}$ is well formed.  This context is well formed --- its set of variable interpretations is defined ---
if $\semvalue{\sigma}$ is defined, $x$ is not already declared in $\Gamma$, and for $\rho \in \convalue{\Gamma}$ we have that
$\semvalue{\sigma}\rho$ is a class (sets are also classes).

\ignore{
We take all values to be tagged
with one of five tags classifying every value as either a point (ur-element), a Boolean value ($\true$ or $\false$), a pair, a collection (type) or a function.
The tags ensure that one can always distinguish, say, a pair from a collection.  A pair is defined by its first and second components, a collection is defined by its set of members, and
a function is defined by its set of input-output pairs. When we write $x \in \sigma$, as in clause (4), we mean that $\sigma$ is a collection containing $x$.
}

The semantic definitions are recursive
but recursive calls ultimately involve smaller expressions --- the definitions
are well founded by eventual reduction of the size of the syntactic expressions
involved. The base-case meaning of the constants $\sett$ and $\class$ is not fully specified in clause (3).
The sets and classes are required to satisfy certain ``formation invariants'' defined in terms of the points (ur-elements) contained within them.
These formation invariants are specified in section~\ref{sec:semantics2}.  The appendix proves that the other set and class expressions
whose meaning is defined in figure~\ref{fig:semantics} also satisfy these formation invariants --- the properties are invariants of the process of forming new sets and classes.
The formation invariants are needed to define the isomorphism relation (the value-based groupoid) associated with each class.

Other than clauses (3) and (20) the semantics is completely naive
and rather obvious.  If we were not concerned with treating isomorphism in clause (20) then there would be no need for
formation invariants and sets and classes could be defined in the usual way.

The clauses in figure~\ref{fig:semantics} rely on context to
distinguish use from mention.  For example, we sometimes write $\Phi
\vee \Psi$ for a disjunctive expression --- this is a {\em mention} of
the symbol $\vee$.  Other times we write $\Phi \vee \Psi$ for the
truth value which is the disjunction of the truth values $\Phi$ and
$\Psi$ --- this is a {\em use} of the semantic disjunction operation.
In the equation
$$\semvalue{\Phi \vee \Psi}\rho = \semvalue{\Phi}\rho \vee
\semvalue{\Psi}\rho$$ the left hand side mentions the disjunction
symbol while the right hand side uses the semantic disjunction
operation.  Another example is the mention and use of pairing and projections in clauses (11) and (12).

We will write $\exists\intype{x}{\sigma}\;\Phi[x]$ as an abbreviation for
$\neg \forall\intype{x}{\sigma} \;\neg \Phi[x]$ and write $\Phi \wedge \Psi$, $\Phi \Rightarrow \Psi$, and $\Phi \Leftrightarrow \Psi$
as abbreviations for Boolean expressions built from disjunction and negation.
We can define the expression $\true$
to be $\forall \intype{P}{\bool} \;P \vee \neg P$.

The semantics given in figure~\ref{fig:semantics} specifies the semantic entailment relation $\models$.
Clauses (3) and (4) imply $\epsilon \models \intype{\sett}{\class}$ and clause (5) then implies that the
context $\epsilon;\;\intype{\alpha}{\sett}$ is well formed and that
$\convalue{\epsilon;\;\intype{\alpha}{\sett}}$
is the set of variable interpretations $\rho$ defined on the single
variable $\alpha$ such that $\rho(\alpha)$ is set.  Clause (4) and (8) then imply $\epsilon;\;\intype{\alpha}{\sett} \models \intype{\alpha}{\sett}$.
We can continue in this way to show
$$\intype{\alpha}{\sett};\;\intype{f}{\alpha \rightarrow \alpha};\;\intype{c}{\alpha} \;\;\models\;\;\intype{f(f(c))}{\alpha}$$
and
$$\intype{\alpha}{\sett};\;\intype{f}{\alpha \rightarrow \alpha};\;\intype{c}{\alpha};\;\forall \intype{x}{\alpha}\;f(x) =_\alpha x\;\; \models\;\; f(f(c)) =_\alpha c.$$

There are also semantic entailments that are consequences of subtle properties of the definition for $\sett$, $\class$ and $=_\sigma$
given in section~\ref{sec:semantics2}. For example we have
$$\intype{\alpha}{\sett};\;\intype{x}{\alpha};\;\intype{y}{\alpha};\;x =_\alpha y  \models x \doteq y.$$
However, this only holds for sets.  For example
$$\intype{\alpha}{\sett};\;\intype{\beta}{\sett};\;\alpha=_{\sett} \beta \not \models \alpha \doteq \beta$$
and
$$\intype{\alpha}{\sett};\;\intype{\beta}{\sett} \not \models \intype{(\alpha \doteq \beta)}{\bool}.$$
The second non-entailment is needed because $\alpha \doteq \beta$ and $\beta =_{\sett} \gamma$, which states that $\beta$ and $\gamma$ have the same cardinality,
does not imply $\alpha \doteq \gamma$.
Hence the equation $\alpha \doteq \beta$ does not allow the substitution of an isomorphic for $\beta$ as is required for Boolean expressions.
This subtlety is incorporated into the inference rules of
section~\ref{sec:rules}.

We have the
following examples of class expressions.
\begin{eqnarray*}
  \mathbf{Magma} & \equiv & \Sigma_{\intype{\alpha\;}{\;\sett}}\;(\alpha \times \alpha) \rightarrow \alpha \\
  \mathbf{Group} & \equiv & S_{\intype{G\;}{\;\mathbf{Magma}}}\;\Phi[G] \\
  \mathbf{HyperGraph} & \equiv &  \Sigma_{\intype{\alpha\;}{\;\sett}}\;(\alpha \rightarrow \bool) \rightarrow \bool \\
  \mathbf{Topology} & \equiv & S_{\intype{X\;}{\;\mathbf{HyperGraph}}}\;\Psi[X]
\end{eqnarray*}

\ignore{
  One can align certain clauses in figure~\ref{fig:semantics}
with ``formation rules'' in section~\ref{sec:rules}.  For example, we can align clause (17), which states the semantics of
dependent pair types, with the following formation rule (the second rule of the fourth row of figure~\ref{fig:Pairs}).

\centerline{\unnamed
{\ant{\Gamma \vdash \intype{\sigma}{\class}}
  \ant{\Gamma;\;\intype{x}{\sigma} \vdash \intype{\tau[x]}{\class}}}
{\ant{\Gamma \vdash \intype{\left(\Sigma_{\intype{x\;}{\;\sigma}}\;\tau[x]\right)}{\class}}}}
}

\begin{figure}
\begin{framed}

  {\small
~ \hfill
$\epsilon \vdash \intype{\bool}{\mathbf{Set}}$
\hfill
$\epsilon \rdash \intype{\sett}{\class}$
\hfill
\unnamed
    {\ant{\Gamma \vdash \intype{\sigma}{\mathbf{Set}}}}
    {\ant{\Gamma \vdash \intype{\sigma}{\mathbf{Class}}}}
\hfill ~

~ \hfill \unnamed {\ant{\Gamma \vdash \intype{\sigma}{\class}}
  \ant{\mbox{$x$ not declared in $\Gamma$}}}
{\ant{\Gamma;\;\intype{x}{\sigma} \vdash \wfcon}}
\hfill \unnamed
{\ant{\Gamma \vdash \intype{\Phi}{\bool}}} {\ant{\Gamma;\Phi \vdash
    \wfcon}}
\hfill \unnamed{ \ant{\Gamma;\Theta \vdash \wfcon}}
       {\ant{\Gamma;\Theta \vdash \Theta}} \hfill ~
       
~ \hfill \unnamed{\ant{\Gamma;\Theta \vdash \wfcon} \ant{\Gamma \vdash \Psi}}
{\ant{\Gamma;\Theta \vdash \Psi}}
\hfill \unnamed{\ant{\Gamma \vdash \intype{\Phi}{\bool}}}
{\ant{\Gamma \vdash \intype{\neg \Phi}{\bool}}}
 \unnamed{\ant{\Gamma \vdash \intype{\Phi}{\bool}} \ant{\Gamma
        \vdash \intype{\Psi}{\bool}}} {\ant{\Gamma \vdash
     \intype{(\Phi \vee \Psi)}{\bool}}}
\hfill ~
       
\unnamed {\ant{\Gamma \vdash \intype{\tau}{\class}}
  \ant{\Gamma \vdash \intype{w}{\tau}} \ant{\Gamma \vdash
    \intype{u}{\tau}}} {\ant{\Gamma \vdash \intype{(w
      =_\tau u)}{\bool}}}
\hspace{3ex}
\unnamed {\ant{\Gamma \vdash \intype{\Phi}{\bool}}
              \ant{\Gamma \vdash \intype{\Psi}{\bool}}
              \ant{\Gamma;\Phi \vdash \Psi} \ant{\Gamma;\neg \Phi
                \vdash \Psi}} {\ant{\Gamma \vdash \Psi}}\newline

\unnamed
  {\ant{\Gamma \vdash \intype{\tau}{\class}}
              \ant{\Gamma;\;\intype{x}{\tau} \vdash
                \intype{\Phi[x]}{\bool}}} {\ant{\Gamma \vdash
    \intype{(\forall\intype{x}{\tau}\;\Phi[x])}{\bool}}}
\hspace{3ex} \unnamed {\ant{\Gamma;\intype{x}{\sigma} \vdash
                \intype{\Phi[x]}{\bool}} \ant{\Gamma;
    \intype{x}{\sigma} \vdash \Phi[x]}}
       {\ant{\Gamma \vdash \forall \intype{x}{\sigma}\;\Phi[x]}}

\unnamed {\ant{\Gamma \vdash \forall \intype{x}{\sigma}
                \;\Phi[x]} \ant{\Gamma \vdash \intype{e}{\sigma}}}
       {\ant{\Gamma \vdash \Phi[e]}}
\hfill\unnamed {\ant{\Gamma \vdash \intype{\Phi}{\bool}}
              \ant{\Gamma \vdash \intype{\Psi}{\bool}} \ant{\Gamma
                \vdash \Phi}} {\ant{\Gamma \vdash \Phi \vee \Psi}
              \ant{\Gamma \vdash \Psi \vee \Phi} \ant{\Gamma \vdash
                \neg \neg \Phi}}
\hfill \unnamed {\ant{\Gamma \vdash    \intype{\Phi}{\bool}}
  \ant{\Gamma \vdash \intype{\Psi}{\bool}}
  \ant{\Gamma \vdash \neg \Psi}
  \ant{\Gamma \vdash \neg \Phi}}
       {\ant{\Gamma \vdash \neg(\Phi \vee \Psi)}}

\unnamed {\ant{\Gamma \vdash \intype{\tau}{\sett}}
  \ant{\Gamma;\;\intype{x}{\tau} \vdash \intype{\Phi[x]}{\bool}}}
    {\ant{\Gamma \vdash
        \intype{\left(S_{\intype{x\;}{\;\tau}}\;\Phi[x]\right)}{\sett}}}
\hspace{3ex} \unnamed {\ant{\Gamma \vdash \intype{\tau}{\class}}
      \ant{\Gamma;\;\intype{x}{\tau} \vdash \intype{\Phi[x]}{\bool}}}
           {\ant{\Gamma \vdash
               \intype{\left(S_{\intype{x\;}{\;\tau}}\;\Phi[x]\right)}{\class}}}

\unnamed {\ant{\Gamma \vdash
               \intype{\left(S_{\intype{x\;}{\;\tau}}\;\Phi[x]\right)}{\class}}
             \ant{\Gamma \vdash \intype{e}{\tau}} \ant{\Gamma \vdash
               \Phi[e]}} {\ant{\Gamma \vdash
               \intype{e}{\left(S_{\intype{x\;}{\;\tau}}\;\Phi[x]\right)}}}
           \hspace{3ex} \unnamed {\ant{\Gamma \vdash
               \intype{e}{\left(S_{\intype{x\;}{\;\tau}}\;\Phi[x]\right)}}}
                  {\ant{\Gamma \vdash \intype{e}{\tau}} \ant{\Gamma
                      \vdash \Phi[e]}}

}

\caption{{\bf Structural Rules, Boolean Rules and Subtypes}}
\label{fig:BooleanRules}
\end{framed}
\end{figure}

\begin{figure}
\begin{framed}
{\small
~ \hfill \unnamed {\ant{\Gamma \vdash \intype{e}{\tau}}} {\ant{\Gamma
      \vdash e =_\tau e}}
  \hfill \unnamed {\ant{\Gamma \vdash u=_\tau
      w}} {\ant{\Gamma \vdash w=_\tau u}}
  \hfill \unnamed {\ant{\Gamma
      \vdash u=_\tau w} \ant{\Gamma \vdash w=_\tau s}} {\ant{\Gamma
      \vdash u=_\tau s}} \hfill ~

~ \hfill \unnamed {\ant{\Gamma \vdash \intype{e}{\tau}}} {\ant{\Gamma
      \vdash e \doteq e}}
  \hfill \unnamed {\ant{\Gamma \vdash u\doteq
      w}} {\ant{\Gamma \vdash w\doteq u}}
  \hfill \unnamed {\ant{\Gamma
      \vdash u\doteq w} \ant{\Gamma \vdash w\doteq s}} {\ant{\Gamma
      \vdash u\doteq s}} \hfill ~

~ \hfill \unnamed {\ant{\Gamma\;\vdash \Theta[u]} \ant{\Gamma \vdash u
      \doteq w}} {\ant{\Gamma \vdash \Theta[w]}}
  \hfill \unnamed
  {\ant{\Gamma \vdash \intype{\sigma,\tau}{\class}}
    \ant{\Gamma;\;\intype{x}{\sigma} \vdash \intype{e[x]}{\tau}}
    \ant{\Gamma \vdash w =_\sigma u}} {\ant{\Gamma \vdash e[w] =_\tau
      e[u]}}
  \hfill \unnamed
         {\ant{\Gamma \vdash \intype{\tau}{\mathbf{Set}}}
           \ant{\Gamma \vdash u =_\tau w}}
         {\ant{\Gamma \vdash u \doteq w}}
\hfill ~

\unnamed {\ant{\Gamma \vdash \intype{\sigma}{\sett}}
  \ant{\Gamma;\;\intype{x}{\sigma} \vdash \intype{\tau[x]}{\sett}}}
         {\ant{\Gamma \vdash
      \intype{\left(\Sigma_{\intype{x\;}{\;\sigma}}\;\tau[x]\right)}{\sett}}}
  \hspace{3ex}
  \unnamed {\ant{\Gamma \vdash \intype{\sigma}{\class}}
    \ant{\Gamma;\;\intype{x}{\sigma} \vdash \intype{\tau[x]}{\class}}}
         {\ant{\Gamma \vdash
             \intype{\left(\Sigma_{\intype{x\;}{\;\sigma}}\;\tau[x]\right)}{\class}}}

\medskip

\unnamed {\ant{\Gamma \vdash
    \intype{\left(\Sigma_{\intype{x\;}{\;\sigma}}
      \tau[x]\right)}{\class}} \ant{\Gamma \vdash \intype{u}{\sigma}}
  \ant{\Gamma \vdash \intype{w}{\tau[u]}}} {\ant{\Gamma \vdash
    \intype{(u,w)}{\left(\Sigma_{\intype{x\;}{\;\sigma}}\;
      \tau[x]\right)}} \ant{\Gamma \vdash \pi_1((u,w)) \doteq u}
  \ant{\Gamma \vdash \pi_2((u,w)) \doteq w}}
\hspace{3ex}
\unnamed
{\ant{\Gamma \vdash \intype{p}{\left(\Sigma_{\intype{x\;}{\;\sigma}}
      \;\tau[x]\right)}}} {\ant{\Gamma \vdash
    \intype{\pi_1(p)}{\sigma}} \ant{\Gamma \vdash
    \intype{\pi_2(p)}{\tau[\pi_1(p)]}} \ant{\Gamma \vdash p \doteq
    (\pi_1(p),\;\pi_2(p))}}

\medskip
\unnamed
{\ant{\Gamma \vdash \intype{\sigma}{\sett}}
  \ant{\Gamma;\;\intype{x}{\sigma} \vdash \intype{\tau[x]}{\sett}}}
{\ant{\Gamma \vdash \intype{\left(\Pi_{\intype{x\;}{\;\sigma}}\;\tau[x]\right)}{\sett}}}
\hspace{3ex}
\unnamed
{\ant{\Gamma \vdash \intype{\left(\Pi_{\intype{x\;}{\;\sigma}}\;\tau[x]\right)}{\sett}}
  \ant{\Gamma;\;\intype{x}{\sigma} \vdash \intype{e[x]}{\tau[x]}}}
{\ant{\Gamma \vdash
           \intype{\left(\lambda\;\intype{x}{\sigma}\;e[x]\right)}{\left(\Pi_{\intype{x\;}{\;\sigma}}\;\tau[x]\right)}}}

\medskip
\unnamed {\ant{\Gamma \vdash
    \intype{\left(\lambda\;\intype{x}{\sigma}\;e[x]\right)}{\left(\Pi_{\intype{x\;}{\;\sigma}}\;\tau[x]\right)}}
  \ant{\Gamma \vdash \intype{u}{\sigma}}} {\ant{\Gamma \vdash
    \left(\left(\lambda\;\intype{x}{\sigma}\;e[x]\right)\;u\right)
    \doteq e[u]}}
\hspace{3ex}
\unnamed {\ant{\Gamma \vdash
    \intype{f}{\left(\Pi_{\intype{x\;}{\;\sigma}}\;\tau[x]\right)}}
  \ant{\Gamma \vdash \intype{u}{\sigma}}} {\ant{\Gamma \vdash
    \intype{f(u)}{\tau[u]}}}

\medskip
\unnamed {\ant{\Gamma \vdash
    \intype{f,g}{\Pi_{\intype{x\;}{\;\sigma}}\;\tau[x]}}
  \ant{\Gamma;\;\intype{x}{\sigma} \vdash f(x) \doteq g(x)}}
    {\ant{\Gamma \rdash f \doteq g}}
\hspace{3ex}
\unnamed {\ant{\Gamma
        \vdash
        \intype{\left(\Pi_{\intype{x\;}{\;\sigma}}\;\tau[x]\right)}{\sett}}
      \ant{\Gamma \vdash  \forall \intype{x}{\sigma}\;\exists
        \intype{y}{\tau[x]}\; \Phi[x,\;y]}}
         {\ant{\Gamma \vdash \begin{array}{l} \exists \intype{f}{\left(\Pi_{\intype{x\;}{\;\sigma}}\;\tau[x]\right)} \\\;\;\;\;\forall \intype{x}{\sigma}\;\Phi[x,f(x)] \end{array}}}
}         
  
\caption{{\bf Equality Rules, Pair Types, Function Types,
    Extensionality and Choice}}
\label{fig:Pairs}
\end{framed}
\end{figure}

\begin{figure}
  \begin{framed}
    {\small
$$\mathbf{Bijection}[u,v] \equiv S_{\intype{f\;}{\;u \rightarrow v}}\;\forall\;\intype{x,y}{u}\; f(x) =_v f(y) \Leftrightarrow x =_u y.$$
\unnamed {\ant{\Gamma \rdash \intype{u,v}{\mathbf{Set}},
    \intype{f}{\mathbf{Bijection}[u,v]}}
  \ant{\Gamma;\;\intype{\alpha}{\sett} \rdash
    \intype{\tau[\alpha]}{\sett}}} {\ant{\Gamma \vdash
    \intype{\mathbf{Carrier}(u,v,f,(\lambda\;\intype{\alpha}{\sett}\;\tau[\alpha]))}{\mathbf{Bijection}(\tau[u],\tau[v])}}
  \ant{}
  \ant{\Gamma \rdash \begin{array}{l} \forall \;\intype{h}{\tau[u]} \\
                                    \;\;(u,h) \;=_{\Sigma_{\intype{\alpha\;}{\;\sett}}\;\tau[\alpha]}\; (v,\mathbf{Carrier}(u,v,f,(\lambda\;\intype{\alpha}{\sett}\;\tau[\alpha]))(h)) \end{array}}}

\parbox{3.0in}{
\unnamed {\ant{\Gamma \rdash \intype{u,v}{\mathbf{Set}},
    \intype{f}{\mathbf{Bijection}[u,v]}}} {\ant{\Gamma \rdash
    \mathbf{Carrier}(u,v,f,(\lambda \intype{\alpha}{\sett}\;\alpha))
    \doteq f}}} \hfill ~

\parbox{3.0in}{ \unnamed {\ant{\Gamma \rdash \intype{u,v}{\mathbf{Set}},
    \intype{f}{\mathbf{Bijection}[u,v]}} \ant{\Gamma \rdash
    \intype{w}{\sett}}} {\ant{\Gamma \rdash
    \mathbf{Carrier}(u,v,f,(\lambda \intype{\alpha}{\sett}\;w)) \doteq
    (\lambda \intype{x}{w}\;x)}}}

\parbox{3.0in}{
\unnamed
{\ant{\Gamma \rdash\mathbf{Carrier}(u,v,f,(\lambda\;\intype{\alpha}{\sett}\;\tau[\alpha])) \doteq g}
  \ant{\Gamma \rdash\mathbf{Carrier}(u,v,f,(\lambda\;\intype{\alpha}{\sett}\;\sigma[\alpha])) \doteq h}}
{\ant{\Gamma \rdash  \begin{array}{l}
      \mathbf{Carrier}(u,v,f,(\lambda\intype{\alpha}{\sett}\;\tau[\alpha]
      \times \sigma[\alpha])) \\ \doteq \lambda\intype{x}{(\tau[u]
        \times \sigma[u])}\;(g(\pi_1(x)),h(\pi_2(x)))
    \end{array}
    }}}

\parbox{3.0in}{ \unnamed {\ant{\Gamma
    \rdash\mathbf{Carrier}(u,v,f,(\lambda\;\intype{\alpha}{\sett}\;\tau[\alpha]))
    \doteq g} \ant{\Gamma
    \rdash\mathbf{Carrier}(u,v,f,(\lambda\;\intype{\alpha}{\sett}\;\sigma[\alpha]))
    \doteq h} \ant{\Gamma \rdash\intype{k}{\tau[u] \rightarrow
      \sigma[u]}} \ant{\Gamma \rdash \intype{a}{\tau[u]}}}
       {\ant{\Gamma \rdash
           \mathbf{Carrier}(u,v,f,(\lambda\intype{\alpha}{\sett}\;\tau[\alpha]
           \rightarrow \sigma[\alpha]))(k)(g(a)) \doteq h(k(a))}}
      }
  
\parbox{3.0in}{ \unnamed {\ant{\Gamma \vdash
    \intype{a,b}{S_{\intype{x\;}{\;\sigma}} \Phi[x]}} \ant{\Gamma
    \vdash a =_\sigma b}} {\ant{\Gamma \vdash a
    =_{\left(S_{\intype{x\;}{\;\sigma}}\;\Phi[x]\right)}\;\;b}}}
    }
    
\caption{{\bf Isomorphism Rules}}
\label{fig:IsoRules}
\end{framed}
\end{figure}

\section{Inference Rules}
\label{sec:rules}

The rules in figures~\ref{fig:BooleanRules}, \ref{fig:Pairs} and
\ref{fig:IsoRules} define a formal proof-theoretic system.  A sequent
$\Gamma \vdash \Theta$ is called {\em valid} if we have $\Gamma
\models \Theta$ as defined in figure~\ref{fig:semantics}.  An
inference rule is {\em sound} if the validity of the antecedents
(the sequents above the line) imply the validity of the conclusion.
Soundness of the inference rules under the semantics of
figure~\ref{fig:semantics} is proved in section~\ref{sec:Soundness}.

Figure~\ref{fig:BooleanRules} gives structural rules and rules
for Boolean expressions and subtypes.  A sequent of the form $\Gamma \vdash \true$ expresses the
statements that $\Gamma$ is well-formed, i.e., that
$\convalue{\Gamma}$ is defined. A rule with multiple conclusions, such
as the second rule of the fourth row, abbreviates multiple rules each
with the same antecedents but with a separate rule for each
conclusion.  Other rules should be self explanatory and justified by
the semantics in figure~\ref{fig:semantics}.

Figure~\ref{fig:Pairs} gives equality rules, rules for pairs and
functions, and the (nonconstructive) axiom of choice.  The axiom of choice (the last rule
of the last row) is restricted to sets. \ignore{We note that choice fails for
proper classes.  To see this let $\mathbf{Two}$ be the class of all two element
sets.  We have $\forall \intype{\alpha}{\mathbf{Two}}\; \exists \intype{x}{\alpha}\;\true$
but there does not exist any functor in $\Pi_{\intype{\alpha\;}{\;\mathbf{Two}}}\;\alpha$.}

Figure~\ref{fig:IsoRules} gives the inference rules for deriving
isomorphism relationships at pair types of the form
$\Sigma_{\intype{\alpha\;}{\;\sett}}\;\tau[\alpha]$.  Intuitively, for sets $u$ and $v$ we
have $(u,h)$ is isomorphic to $(v,g)$ as a member of $\Sigma_{\intype{\alpha\;}{\;\sett}}\;\tau[\alpha]$  if there exists a bijection $f$
from $u$ to $v$ that ``carries'' $h$ to $g$.    The carring operation is a bijection from $\tau[u]$ to $\tau[v]$.
The first rule of figure~\ref{fig:IsoRules} has two conclusions the first of which
acts as a formation rule for the carring operation.
The carrying operation is detrmined by the sets $u$ and $v$, the bijection $f$ between them, and the mapping
from a set $\alpha$ to the set $\tau[\alpha]$ which we can write as
$\lambda\;\intype{\alpha}{\sett}\;\tau[\alpha]$.
The second conclusion of the first rule states that the carrying operation yields an isomorphic pair.

The remaining rules in figure~\ref{fig:IsoRules} define the
carrying operation in the case where $\tau[\alpha]$ is a simple type
--- one built from $\alpha$, type expressions not involving $\alpha$,
and simple pair and function types.  The second and third rule give
the base cases for $\tau[\alpha] = \alpha$ and for $\tau[\alpha] = w$
with $w$ not containing $\alpha$.  The fourth and fifth rules define
the carrying relation at simple pair types and simple function types
respectively.  The final rule handles subtypes. These rules can be combined with the inference rule for the substitution of isomorphics to yield a wide variety of ismorphism equations.
The general semantics of
carrying and proof of the soundness of these rules is given in section~\ref{sec:Soundness}.

\section{Morphoids}
\label{sec:semantics2}

In the semantics developed here a class $\sigma$ is just a collection of
values with no auxiliary information about isomorphisms.  But in this case how do
we define $x =_\sigma y$?  The fundamental idea is value-based
groupoid structure.  We define the values so that for every
value $x$ we also have values $\leftop(x)$, $\rightop(x)$ and
$x^\inv$. Here $x$ itself is viewed as an isomorphism between
$\leftop(x)$ and $\rightop(x)$.  For $\rightop(x) = \leftop(y)$ we
also define the composition $x \circ y$.  These operations are defined
in figure~\ref{fig:values}.  The operations are defined independent
of any type containing the values involved.  These operations satisfy
algebraic properties of a groupoid.  Any collection of values that is
closed under these operations then forms a (value-based) groupoid.
The values are simultaneously objects (oids) and morphisms --- they
are {\bf morphoids}.

{\bf A Grothendieck Universe.} To construct a space of morphoid values
we assume a Grothendieck universe $U$ --- a standard model of set theory.
All of the set-forming operations allowed in set theory can be carried out
within a single Grothendieck universe.  A Grothendiek universe is assumed to be ``full''
in the sense that if $U$ contains a set $\sigma$ then $U$ also contains all the (true Platonic) subsets of $\sigma$.

{\bf Tagged Values.} All morphoid values are tagged with one of five tags classifying each value as either a
Boolean value, a point, a pair, a collection (set or class) or a function. We will write the
Boolean values as $\true$ and $\false$, write points as
$\pointt(i,j)$, write pairs as $(x,y)$, write collections using set notation
$\{\ldots\}$ and write functions as sets of input-output pairs $\{x \mapsto y, \;\ldots\}$.
Pair values are defined by their two
components, collections are defined by their members, and
function values are defined by their input-output pairs.

{\bf Morphoid Points.} Morphoids are built from morphoid points.
Morphoid points can be thought of as structured ur-elements
of set theory. A morphoid point is written as
$\pointt(i,j)$ where $i$ and $j$ are arbitrary elements of $U$.  We
call $i$ the left index and $j$ the right index.  We define
$\leftop(\pointt(i,j)) = \pointt(i,i)$,
$\rightop(\pointt(i,j)) = \pointt(j,j)$,
$\pointt(i,j)^\inv = \pointt(j,i)$ and $\pointt(i,j) \circ \pointt(j,k) = \pointt(i,k)$.  These operations on points
satisfy the groupoid properties listed in figure~\ref{fig:GroupOps}.

It turns out that the groupoid operations on points can be extended to
all values built from points in a way that satisfies the groupoid axioms provided that we require that the members of the
constant $\sett$ satisfy the ``formation
invariant'' for sets stated in figures~\ref{fig:values}
that (hereditarily) every set is bijective --- no two elements have
the same left value or the same right value.

Figure~\ref{fig:values} defines the morphoid values and the
groupoid operations.  It starts by defining templates.  A template is
an expression specifying structure.  A template can be viewed as an
abstract type expression specifying where points occur in a value.
For example, we can define a group to be a pair of a set and a binary operation on that set
such that an identity element and an inverse operation exist satisfying the algebraic properties
of a group.  An abstract group is a group whose elements are points. For an abstract group $G$ we have
$$\intype{G}{(\mathbf{SetOf}(\pointt) \times (\pointt \times \pointt \rightarrow \pointt))}.$$
The right hand side of the above expression is a template as defined by the template grammar
at the top of figure~\ref{fig:values}.  Of course not all groups have points as group elements.
Group representations, such as permutation groups, or groups of linear operators, are also groups.
This is discussed in more detail in the discussion of figure~\ref{fig:Abstraction}.
\begin{figure}
\begin{framed}

  \parbox{3.0in}{
    \centerline{\bf Templates}
    \medskip
    {\bf A template} is an expression generated by the following grammar.}
$${\mathcal T} ::= \bool \;|\; \pointt \;|\; \mathbf{SetOf}({\mathcal T})
\;|\; {\mathcal T}_1 \times {\mathcal T}_2 \;|\; {\mathcal T}_1 \rightarrow {\mathcal
  T}_2$$
\parbox{3.0in}{For $x \in U$ and template ${\mathcal T}$ we define $\intype{x}{{\mathcal T}}$
by the following clauses.}

\medskip
\begin{itemize}
\item $\intype{x}{\pointt}$ if $x$ is a point $\pointt(i,j)$.
\item $\intype{\Phi}{\bool}$ if $\Phi$ is a Boolean value.
\item $\intype{(x,y)}{{\mathcal T}_1 \times {\mathcal T}_2}$ if
  $\intype{x}{{\mathcal T}_1}$ and $\intype{y}{{\mathcal T}_2}$.
\item $\intype{\sigma}{\mathrm{SetOf}({\mathcal T})}$ if for all $x \in
  \sigma$ we have $\intype{x}{{\mathcal T}}$.
\item $\intype{f}{{\mathcal T}_1 \rightarrow {\mathcal T}_2}$ if $f$ is any
  (possibly non-functional) set of input-output pairs such that for
  $(x \mapsto y) \in f$ we have $\intype{x}{{\mathcal T}_1}$ and
  $\intype{y}{{\mathcal T}_2}$.
\end{itemize}

\medskip
\centerline{\bf Weak Values}
\medskip
\parbox{3.0in}{{\bf A weak value} is an element $x$ of $U$ such that there exists a
template ${\mathcal T}$ with $\intype{x}{\mathcal T}$.  For a weak function
value $f$ we write $\domop(f)$ for the set of input values in the
pairs of $f$.}

\bigskip
\centerline{\bf The Groupoid Operations}
\medskip
\parbox{3.0in}{For a weak value $x$ we define $\leftop(x)$ to be the result of
replacing each point $\pointt(i,j)$ within $x$ by $\pointt(i,i)$.  This can be
defined recursively as follows.}
\begin{eqnarray*}
\leftop(\Phi) & = & \Phi \\ \leftop(\pointt(i,j)) & = & \pointt(i,i)
\\ \leftop((x,y)) & = & (\leftop(x),\leftop(y)) \\ \leftop(\sigma) & =
& \{\leftop(x),\;x \in \sigma\} \\ \leftop(f) & = & \{\leftop(x)
\mapsto \leftop(y), \; (x \mapsto y) \in f\}
\end{eqnarray*}
\parbox{3.0in}{$\rightop(x)$ similarly replaces each point $\pointt(i,j)$ in $x$ by
$\pointt(j,j)$ and $x^\inv$ replaces each point $\pointt(i,j)$ by
$\pointt(j,i)$.  For weak values $x$ and $y$ we have that $x
\circ y$ is defined when $\rightop(x) = \leftop(y)$ in which case we
define $x \circ y$ by the following rules.}
\begin{eqnarray*}
\Phi \circ \Phi & = & \Phi \\ \pointt(i,j) \circ \pointt(j,k) & = &
\pointt(i,k) \\ (x,y) \circ (z,w) & = & (x\circ z,\;y\circ w)
\\ \sigma \circ \tau & = & \{x \circ y,\;x \in \sigma,\;y \in \tau\}
\\ f \circ g & = & \left\{\begin{array}{l} x \circ x' \mapsto y \circ y': \\\; (x \mapsto y) \in f, \; (x' \mapsto y') \in g\end{array}\right\}
\end{eqnarray*}

\medskip
\centerline{\bf Values}
\medskip
\parbox{3.0in}{A weak set value $\sigma$ will be called bijective if for all $x,y \in
\sigma$ with $x \not = y$ we have $\leftop(x) \not = \leftop(y)$ and
$\rightop(x) \not = \rightop(y)$.}

\medskip
\parbox{3.0in}{A weak function value $f$ will be called functional if no two
input-output pairs of $f$ have the same input value.}

\medskip
\parbox{3.0in}{{\bf A value} is a weak value within which each set value is bijective and
each function value is functional.  More formally, a value is a weak
value that is either a Boolean value, a point, a pair of values, a
bijective set of values or a functional function value $f$ such that
$\domop(f)$ is a set value and for $(x \mapsto y) \in f$ we have that
$x$ and $y$ are values.}

\caption{\bf Values}
\label{fig:values}
\end{framed}
\end{figure}

\begin{figure}
  \begin{framed}
\medskip
\begin{itemize}
\item[(7.1)] For any value $x$ we have that $\leftop(x)$,
  $\rightop(x)$ and $x^\inv$ are also values.

\item[(7.2)] For any values $x$ and $y$ with $x \circ y$ is defined we
  have that $x \circ y$ is a value.

\item[(7.3)] $\leftop(x^{-1}) = \rightop(x)$ and $\rightop(x^{-1}) =
  \leftop(x)$

\item[(7.4)] $\leftop(x \circ y) = \leftop(x)$ and $\rightop(x \circ
  y) = \rightop(y)$.

\item[(7.5)] $(x \circ y) \circ z$ = $x \circ (y \circ z)$.

\item[(7.6)] $x^{-1} \circ x \circ y = y$ and $x \circ y \circ y^{-1}
  = x$.

\item[(7.7)] $\rightop(x) = x^\inv \circ x$ and $\leftop(x) = x \circ
  x^\inv$

\item[(7.8)] $(x^\inv)^\inv = x$.

\item[(7.9)] $(x \circ y)^\inv = y^\inv \circ x^\inv$.
\end{itemize}

\caption{{\bf Groupoid Properties}}
\label{fig:GroupOps}
\end{framed}
\end{figure}

\begin{figure}
  \begin{framed}
    {\small
      \centerline{\bf The Abstraction Operation}
      \medskip
\parbox{3.0in}{For a weak value $x$ and template ${\mathcal T}$ {\bf we define $x@{\mathcal T}$} by
  the rules below where $x@{\mathcal T}$ is undefined if no rule applies.}
\medskip
\begin{eqnarray*}
  \Phi@\bool & = & \Phi\;\mbox{for $\Phi$ a Boolean} \\
  (x,y)@({\mathcal T}_1 \times {\mathcal T}_2) & = & (x@{\mathcal T}_1,y@{\mathcal T}_2)\;\mbox{for $x@{\mathcal T}_1$, $y@{\mathcal T}_2$ defined}\\
  \sigma@\setof({\mathcal T}) & = & \{x@{\mathcal T},\;x \in \sigma\} \; \mbox{for $\sigma$ a set} \\
  & & \mbox{with $x@{\mathcal T}$ defined for all  $x \in \sigma$.} \\
  f@({\mathcal T}_1 \rightarrow {\mathcal T}_2) & = & \{x@{\mathcal T}_1 \mapsto y@{\mathcal T}_2,\;(x \mapsto y) \in f\} \\
  & & \mbox{for $f$ a function with $x@{\mathcal T}_1$ and $y@{\mathcal T}_2$} \\
  & & \mbox{defined for all $(x \mapsto y) \in f$.} \\
  \\
  \pointt(i,j)@\pointt & = & \pointt(i,j) \\
  x@\pointt & = & \pointt(\subpoint(\leftop(x)),\subpoint(\rightop(x))) \\
  & & \mbox{for $x$ not a point} \\
  \\
  \subpoint(\Phi) & = & \Phi \;\mbox{for $\Phi$ Boolean} \\
  \subpoint((x,y)) & = &(\pointify(x), \pointify(y)) \\
  \subpoint(\sigma) & = & \{\pointify(x),\; x \in \sigma\} \\
  \subpoint(f) & = & \{\pointify(x) \mapsto \pointify(y), \;(x \mapsto y) \in f\} \\
  \\
  \pointify(\pointt(i,i)) & = & \pointt(i,i) \\
  \pointify(x) & = & \pointt(\subpoint(x),\subpoint(x)) \;\mbox{$x$ not a point}
\end{eqnarray*}

\bigskip
{\bf Abstraction Properties:}

\medskip
\parbox{3.0in}{For weak values $x$ and $y$ we define $x \preceq y$ to mean that for
$y@{\mathcal T}$ defined we also have $x@{\mathcal T}$ defined and $x@{\mathcal T} = y@{\mathcal T}$.}

\medskip
\parbox{3.0in}{For weak values $x$ and $y$ we define $x \sim y$ to mean that there
exists a value $z$ with $x \circ z \circ y$ defined.}

\medskip
\begin{itemize}
\item [(8.1)] For a value $x$ with $x@{\mathcal T}$ defined we have that
  $x@{\mathcal T}$ is a value with $\intype{(x@{\mathcal T})}{{\mathcal T}}$.
\item [(8.2)] We have $x@{\mathcal T} = x$ if and only if $\intype{x}{\mathcal T}$.
\item [(8.3)] We have that $\sim$ is an equivalence relation and for
  any value $x$ we have $x \sim x^\inv \sim \leftop(x) \sim \rightop(x)$ and for values $x$ and $y$ with $x \circ y$ defined we
  have $x \sim (x \circ y) \sim y$.
\item [(8.4)] For $\intype{x}{\mathcal T}$ and $x \sim y$ we have
  $\intype{y}{\mathcal T}$.
\item [(8.5)] For $x@{\mathcal T}$ defined and $x \sim y$ we have $y@{\mathcal  T}$ defined.
\item [(8.6)] If $(x@{\mathcal T})@{\mathcal S}$ is defined then $(x@{\mathcal T})@{\mathcal S} = x@{\mathcal S}$.
\item [(8.7)] If $x@{\mathcal T}$ is defined then $x \preceq x@{\mathcal T}$. 
\item [(8.8)] $\preceq$ is a partial order on values.
\item [(8.9)] For $(x@{\mathcal T})@{\mathcal S}$ and $(x@{\mathcal S})@{\mathcal T}$ both defined we have $x@{\mathcal T} = x@{\mathcal  S}$.
\item [(8.10)] For $x@{\mathcal T}$ defined we have $(x^\inv)@{\mathcal T} = (x@{\mathcal T})^\inv$.
\item[(8.11)] For $(x \circ y)@{\mathcal T}$ defined we have $(x \circ  y)@{\mathcal T} = (x@{\mathcal T}) \circ (y@{\mathcal T})$.
\item[(8.12)] For $\intype{x}{\mathcal T}$ and $\intype{y}{\mathcal T}$ and $(x@{\mathcal S})\circ (y@{\mathcal S})$ defined
  we have that $x \circ y$ is defined.
\end{itemize}
}

\caption{{\bf Abstraction}}
\label{fig:Abstraction}
\end{framed}
\end{figure}

While we allow representations of groups, we require that for every value $x$ there exists a template ${\mathcal T}$
such that $\intype{x}{\mathcal T}$. More specifically, we define a weak value to be any element $x$ of the universe $U$ such that
there exists a template ${\mathcal T}$ such that $\intype{x}{\mathcal T}$.  For technical reasons weak functions are not required to be functional ---
they are allowed to contain two different input-output pairs with the same input value.
The values are defined to be the weak values that hereditarily
satisfy the formation invariants that sets are bijective and that functions are functional.
  
It should be noted that values need not
have unique templates --- a minimal example is
$\intype{\emptyset}{\setof({\mathcal T})}$ for any ${\mathcal T}$ where
$\emptyset$ is the empty set which can be denoted as
$S_{\intype{P\;}{\;\bool}}\;\false$.  While values need not have
unique templates, any weak value $x$ has a finite depth over points as
specified by any template ${\mathcal T}$ with $\intype{x}{\mathcal T}$. While
values have finite depth, sets (including sets of points) can have very large cardinality.

For a weak value $x$ the operation $\leftop(x)$ replaces every point
$\pointt(i,j)$ in $x$ by $\pointt(i,i)$. $\rightop(x)$ is defined
similarly and $x^\inv$ replaces every point in $x$ $\pointt(i,j)$ by
$\pointt(j,i)$.
To better understand composition we can consider sets of points.  A set of points is
bijective if no two points have the same left index or the same right index.
The composition $\sigma \circ \tau$ of two (bijective) point
sets, $\sigma$ and $\tau$, as defined in figure~\ref{fig:values}, is
the point set representing the bijection that is the composition of
the bijections represented by $\sigma$ and $\tau$.  So the class of
all point sets forms a value-based groupoid whose elements are
bijections under inverse and composition.

Figure~\ref{fig:GroupOps} states the algebraic groupoid
properties.  These properties are proved for values and classes in the appendix.

Figure~\ref{fig:Abstraction} defines the abstraction operation.  The abstraction
operation is central to defining the isomorphism relation $=_\sigma$ in a way that handles both abstract
elements and representations.
For example consider two group representations $G$ and $H$, perhaps a permutation group and a group of linear operations.
We define $G =_{\mathbf{Group}} \;H$
to mean that there exists an abstract group $F$ such that
$(G@\mathbf{Group}) \circ F \circ (H@\mathbf{Group})$ is defined.  Here $F$ is the isomorphism
between $G$ and $H$.  The expression $G@\mathbf{Group}$ is the coercion of the group representation $G$ into an abstract group --- a group whose group
elements are points.  The abstraction $G@\mathbf{Group}$ is an abbreviation for
$$G@(\mathbf{SetOf}(\pointt) \times (\pointt \times \pointt \rightarrow \pointt)).$$
For a template ${\mathcal T}$ the operation $x@{\mathcal T}$ is defined in figure~\ref{fig:Abstraction}.
The operation $x@{\mathcal T}$ converts parts of $x$ to points as
specified in ${\mathcal T}$.  The abstraction operation is
partial --- for $x@{\mathcal T}$ to be defined $x$ must have a shape
compatible with ${\mathcal T}$.  The complex definition of $x@\pointt$
achieves the property that if $x@{\mathcal T}$ is defined then $(x@{\mathcal T})@\pointt = x@\pointt$.  This supports property (8.6)
in figure~\ref{fig:Abstraction}.  Properties (8.1) through (8.12)
are proved in the appendix.

Figure~\ref{fig:classes} defines classes by stating formation
invariants that classes must satisfy.  It is useful to again consider
the class of all sets of points (point sets).  The class of all point
sets is closed under inverse and composition and hence forms a
(value-based) groupoid.  However, it is possible to form classes that
are not closed under inverse.  A minimal example is given by the
following sequent.
$$\intype{\alpha}{\sett} \vdash \intype{\left(\sett \times \alpha\right)}{\class}$$
Here it is possible to interpret $\alpha$
as a point set whose set of left indexes is disjoint from its set of
right indexes.  The point set $\alpha$ is then not closed under
inverse --- for $\pointt(i,j) \in \alpha$ we have $\pointt(j,i) \not
\in \alpha$.  The class $\sett \times \alpha$ is the set of pairs of
the form $(\sigma, \pointt(i,j))$ for $\sigma$ a set and $\pointt(i,j)
\in \alpha$.  The class $\sett \times \alpha$ is not closed under
inverse.
$$(\sigma,\;\pointt(i,j))^\inv = (\sigma^\inv,\;\pointt(j,i)) \not \in
\sett \times \alpha$$ This situation arises for the class of vector
spaces over a given field $F$.  The class of all fields forms a
value-based groupoid in which fields can be inverted and composed.  An
individual field whose set of field elements is a bijective point set
acts as an isomorphism between fields.  If $V$ is a vector space over
$F$ then $V^\inv$ is a vector space over $F^\inv$ but for $F \not =
F^\inv$ we have that $V^\inv$ is not a vector space over $F$.  In
general closed type expressions denote groupoids while open type
expressions (type expressions containing free variables) need not be
closed under inverse.

\begin{figure}
  \begin{framed}
    {\small
      \centerline{\bf Classes}
      \medskip
    \parbox{3.0in}{
 A class $\sigma \subseteq U$ is a collection of values, possibly too
 large be an element of $U$, where we require:

  \medskip
  \begin{itemize}
  \item[(9.1)] Morphoid closure --- for $x,y,z \in \sigma$ with $x
    \circ y^\inv \circ z$ defined we have $x\circ y^\inv \circ z \in
    \sigma$
    \medskip
  \item[(9.2)] Interface template --- there exists ${\mathcal T}$ such
    that for all $x \in \sigma$ we have that $x@{\mathcal T}$ is defined and
    $x@{\mathcal T} \in \sigma$.
  \end{itemize}

  \medskip
  \centerline{\bf Operations}
    \medskip
  For a class $\sigma$ and $x \in \sigma$ we define $x@\sigma$ to be
  $x@{\mathcal T}$ for ${\mathcal T}$ an interface template for $\sigma$.  By
  property (8.9) this definition is independent of the choice of
  ${\mathcal T}$.

  \medskip
  We write $\intype{\sigma}{\classof({\mathcal T})}$ if $\sigma$ is a
  class with interface template ${\mathcal T}$.

  \medskip
  For a class $\sigma$ and template ${\mathcal T}$ with $x@{\mathcal T}$
  defined for all $x \in \sigma$ we define $\sigma@\classof({\mathcal T})$
  to be the class $\{x@{\mathcal T},\;x \in \sigma\}$.

  \medskip
  The groupoid operations on classes are defined by the following
  rules.  Section~\ref{sec:ClassProofs} of the appendix proves that classes
  satisfy the groupoid properties.
  \begin{eqnarray*}
    \leftop(\sigma) & = & \{x_1 \circ x_2^\inv,\;x_1,x_2 \in \sigma\}
    \\ \rightop(\sigma) & = & \{x_1^\inv \circ x_2,\;x_1,x_2 \in
    \sigma\} \\ \sigma^\inv & = & \{x^\inv,\;x \in \sigma\} \\ \sigma
    \circ \tau & = & \{x\circ y,\;x \in \sigma,\;y\in \tau\}
  \end{eqnarray*}

  \medskip
  \centerline{\bf Abstraction Ordering and Isomorphism}
  
  \medskip
  For classes $\sigma$ and $\tau $ we define $\sigma \preceq
  \tau$ to mean that for $\intype{\tau}{\classof({\mathcal T})}$ we have
  $\sigma@\classof({\mathcal T}) = \tau@\classof({\mathcal T})$.

  \medskip
  For $x,y \in \sigma$ we define $x =_\sigma y$ to mean $(x@\sigma)
  \circ z^\inv \circ (y@\sigma)$ is defined for some $z \in \sigma$.
}}
  \caption{\bf Classes}
  \label{fig:classes}
\end{framed}
\end{figure}

\begin{figure}
  \begin{framed}
    {\small
      \parbox{3.0in}{
    A structure is mapping from a finite set of variables to values.

    \medskip
    For a structure $\rho$ we define $\leftop(\rho)$ to be the
    structure defined on the same variables as $\rho$ and satisfying
    $\leftop(\rho)(x) = \rho(\leftop(x))$.  $\rightop(\rho)$ and
    $\rho^\inv$ are defined similarly.  If $\rightop(\rho_1) =
    \leftop(\rho_2)$ then $\rho_1 \circ \rho_2$ is defined by $(\rho_1
    \circ \rho_2)(x) = \rho_1(x) \circ \rho_2(x)$.

    \medskip
    A structure template is a mapping from variables to templates.
    For a structure $\rho$ and structure template $\eta$ defined on
    the same variables we write $\intype{\rho}{\eta}$ if
    $\intype{\rho(x)}{\eta(x)}$ for each $x$ and $\rho@\eta$ is
    defined if $\rho(x)@\eta(x)$ is defined for each $x$ in which case $(\rho@\eta)(x) = \rho(x)@\eta(x)$.
    We define $\rho_1 \preceq \rho_2$ to mean that $\rho_1(x) \preceq \rho_2(x)$ for each $x$.

    \bigskip
    For $\semvalue{e}\rho$ defined and $\intype{\rho}{\eta}$ we define
    $\tempvalue{e}{\eta}$ to be a template using the following rules.

    \bigskip
\begin{eqnarray*}
  \tempvalue{x}{\eta} & = & \eta(x) \\ \tempvalue{\bool}{\eta} & = & \setof(\bool) \\
  \tempvalue{\sett}{\eta} & = & \classof(\setof(\pointt)) \\
  \tempvalue{\lambda\;\intype{x}{\sigma}\;e[x]}\eta & = & {\mathcal M}(\tempvalue{\sigma}\eta) \rightarrow \tempvalue{e[x]}\eta' \\
  \tempvalue{\Pi_{\intype{x\;}{\;\sigma}}\;\tau[x]}\eta & = & \setof({\mathcal M}(\tempvalue{\sigma}\eta) \rightarrow {\mathcal M}(\tempvalue{\tau[x]}\eta')) \\
  \tempvalue{\Sigma_{\intype{x\;}{\;\sigma}} \;\tau[x]}{\eta} & = & \classof({\mathcal M}(\tempvalue{\sigma}\eta) \times {\mathcal M}(\tempvalue{\tau[x]}\eta')) \\
  \eta' & = & \eta[x \leftarrow  {\mathcal M}(\tempvalue{\sigma}\eta)] \\
  \tempvalue{S_{\intype{x\;}{\;\sigma}}\;\Phi[x]}{\eta} & = & \tempvalue{\sigma}{\eta} \\
  \tempvalue{f(e)}{\eta} & = & \mathbf{Range}(\tempvalue{f}{\eta}) \\
  \tempvalue{(u,\;w)}{\eta} & = & \tempvalue{u}{\eta} \times \tempvalue{w}{\eta} \\
  \tempvalue{\pi_i(e)}{\eta} & = & \pi_i(\tempvalue{e}{\eta}) \\
  \tempvalue{\Phi}{\eta} & = & \bool \;\;\;\mbox{for $\Phi$ an equality, disjunction,} \\
  & & \;\;\;\;\;\;\;\;\;\;\mbox{negation or quantified formula}
\end{eqnarray*}

\bigskip
\vspace{-4ex}
\begin{eqnarray*}
  {\mathcal M}(\setof({\mathcal T})) & = & {\mathcal M}(\classof({\mathcal
    T})) = {\mathcal T} \\ \mathbf{Range}({\mathcal T}_1 \rightarrow {\mathcal
    T}_2) & = & {\mathcal T}_2 \\ \pi_i({\mathcal T}_1 \times {\mathcal T}_2) & =
  & {\mathcal T}_i
\end{eqnarray*}
    }}
    
\caption{{\bf Structures and Template Evaluation.}}
\label{fig:AbsEval}
\end{framed}
\end{figure}

\begin{figure}
  \begin{framed}
    {\small
      \parbox{3.0in}{
    \medskip
    For $\convalue{\Gamma}$ defined we have

    \medskip
    \begin{itemize}
      \item[(11.1)] For $\rho \in \convalue{\Gamma}$ we have that
        $\rho$ is a structure (all variables are mapped to values).

        \medskip
      \item[(11.2)] For $\rho \in \convalue{\Gamma}$ we have
        $\rho^\inv \in \convalue{\Gamma}$ and for $\rho_1,\rho_2 \in \convalue{\Gamma}$
        with $\rho_1 \circ \rho_2$ defined we have
        $\rho_1 \circ \rho_2 \in \convalue{\Gamma}$.
    \end{itemize}

    \bigskip
    We define a proper class to be a class that is not an element of $U$.

    \medskip
    We define a denotable value to be either a value, a proper class,
    or a pair of denotable values.
    
    \medskip
    For $\semvalue{e}$ defined with $e \not = \class$ and for $\rho,\rho_1,\rho_2 \in \semvalue{\Gamma}$ we have the following.

    \medskip
    \begin{itemize}
    \item[(11.3)] $\semvalue{e}\rho$ is a denotable value.

      \medskip
    \item[(11.4)] $\semvalue{e}(\rho^\inv) = (\semvalue{e}\rho)^\inv$ and for
      $\rho_1 \circ \rho_2$ defined we have $\semvalue{e}(\rho_1 \circ \rho_2) = (\semvalue{e}\rho_1) \circ (\semvalue{e}\rho_2)$.
      
      \medskip
    \item[(11.5)] For
      $\rho_1 \preceq \rho_2$ we have $\semvalue{e}\rho_1 \preceq
      \semvalue{e}\rho_2$.
      
      \medskip
    \item[(11.6)] For structure template $\eta$ with
      $\intype{\rho}{\eta}$ we have
      $\intype{\semvalue{e}\rho}{\tempvalue{e}\eta}$.
    \end{itemize}
  }}
    \caption{{\bf Evaluation Properties.}}
    \label{fig:EvalProps}
  \end{framed}
\end{figure}

While classes are not in general closed under inverse we require that
all classes satisfy the ``morphoid closure condition'' --- (9.1) in
figure~\ref{fig:classes}.  This condition states that for any class
$\sigma$ (including sets), and for $x,y,z \in \sigma$ with $x \circ
y^\inv \circ z$ defined, we have $x \circ y^\inv \circ z \in \sigma$.
The definition of set values imply (9.1) and (9.2) and we have that all
sets are classes.

To better understand the morphoid closure condition (9.1) again
consider the class $\sett \times \alpha$ discussed above.  Consider
$$(\sigma_1,\;\pointt(i_1,j_1)) \circ (\sigma_2,\;\pointt(i_2,j_2))^\inv \circ (\sigma_3,\;\pointt(i_3,j_3))$$
where $\sigma_k$ is a set and
$\pointt(i_k,j_k) \in \alpha$.  Because $\alpha$ must be a bijection
this composition can only be defined if $j_2 = j_1$, implying $i_2 = i_1$, and $i_3 = i_2$, implying $i_3 = i_1$ and $j_3 = j_1$.  So any
such composition has the form
\begin{eqnarray*}
  & & (\sigma_1,\;\pointt(i,j)) \circ (\sigma_2,\;\pointt(i,j))^\inv
  \circ (\sigma_3,\;\pointt(i,j)) \\ & = & (\sigma_1 \circ
  \sigma_2^\inv \circ \sigma_3,\;\pointt(i,j)) \\ & \in & \sett \times
  \alpha
\end{eqnarray*}
Hence the class $\sett \times \alpha$ satisfies the morphoid closure
condition (9.1).

Condition (9.2) specifies that every class must have an interface
template --- for any class $\sigma$ there must exist a template ${\mathcal
  T}$ such that for $x \in \sigma$ we have that $x@{\mathcal T}$ is
defined and $x@{\mathcal T} \in \sigma$.  Figure~\ref{fig:classes} defines
$x@\sigma$ to be $x@{\mathcal T}$ for any interface template ${\mathcal T}$
for $\sigma$. Property (8.9) implies that this definition of
$x@\sigma$ is independent of the choice of the interface template.
For any group $G$ we have
$$G@\mathbf{Group} = G@(\mathbf{SetOf}(\pointt) \times (\pointt
\times \pointt \rightarrow \pointt)).$$ For any group $G$, the group
$G@\mathbf{Group}$ is an ``abstract'' group where the group elements
are points.

Figure~\ref{fig:classes} also defines the isomorphism relation
$=_\sigma$ associated with the class $\sigma$. We have $x =_\sigma y$ if
there exists $z \in \sigma$ with $(x@\sigma) \circ z^\inv \circ (y@\sigma)$ defined.
The inverse operation is needed to handle the case where the class is not closed under inverse.

Figure~\ref{fig:AbsEval} defines structures (variable
interpretations) and structure templates and extends the operations
and relations of the previous figures to structures.  The figure also
defines template evaluation --- a form of abstract interpretaion which computes a ``template value''
for an expression when provided a template for each free variable.  Property (11.6) in figure~\ref{fig:EvalProps}
states that for $\semvalue{e}\rho$ defined and $\intype{\rho}{\eta}$, where $\eta$ is a
structure template, we have that $\tempvalue{e}\eta$ is a template such
that $\intype{\semvalue{e}{\rho}}{\tempvalue{e}\eta}$.

Figure~\ref{fig:EvalProps} States general properties of well-formed
contexts and well-formed expressions. The properties of figure~\ref{fig:EvalProps}
are central to the soundness of the inference rules, especially those involving isomorphism.
These properties are proved by a case analysis
over the constructs listed in figure~\ref{fig:constructs} under the semantics listed in figure
~\ref{fig:semantics}.  This case analysis is done in section~\ref{sec:EvalProps} in the appendix.

\section{Soundness}
\label{sec:Soundness}

We now assume the properties in figures~\ref{fig:GroupOps}, \ref{fig:Abstraction} and~\ref{fig:EvalProps}, plus the lemma that every set value is a class,
all of which are proved in the appendix.  Given this we consider the soundness of the inference rules in figures~\ref{fig:BooleanRules}, \ref{fig:Pairs} and~\ref{fig:IsoRules}.

Most rules follow from the semantic clauses of figure~\ref{fig:semantics} other than the semantically subtle clauses (3) and (20). For such rules, for example the rules for Boolean reasoning,
soundness is immediate. Rules whose soundness rests on clauses (3) and (20) can be divided into type formation rules and isomorphism rules.
For example consider the following formation rule for pair types.

\centerline{ \unnamed {\ant{\Gamma \vdash \intype{\sigma}{\class}} \ant{\Gamma;\;\intype{x}{\sigma}
    \vdash \intype{\tau[x]}{\class}}} {\ant{\Gamma \vdash
    \intype{\left(\Sigma_{\intype{x\;}{\;\sigma}}\;\tau[x]\right)}{\class}}}}

The soundness of such type formation rules follows from property (11.3) which states that for $e \not = \class$ we have that if $\semvalue{e}$ is defined
then for $\semvalue{e}\rho \in U$ we have that $\semvalue{e}\rho$ is a denotable value.  If the pair type in the above rule is a denotable value then it must be a class.
We now turn to proving the soundness of the rules explicitly involving isomorphism.

\begin{theorem}
  \label{lem:EquivalenceRelation}
  For any class $\sigma$ we have that $=_\sigma$ is an equivalence
  relation on the elements of $\sigma$. \label{lem:isoequiv}
\end{theorem}

\begin{proof}
  For $x,y \in \sigma$ we have $x =_\sigma y$ if there exists $z \in \sigma$
  such that $(x@\sigma) \circ z^\inv \circ (y@\sigma)$ is
  defined.  For any $\incat{x}{\sigma}$ we have that $(x@\sigma) \circ (x@\sigma)^{-1} \circ (x@\sigma)$
  is defined and hence $x =_\sigma x$.
  To show symmetry suppose $x =_\sigma y$ with $(x@\sigma) \circ z^{-1} \circ (y@\sigma)$ defined.
  In this case we have that $(y@\sigma) \circ ((x@\sigma) \circ z^{-1} \circ (y@\sigma))^{-1} \circ (x@\sigma)$ is defined.
  By morphoid closure we have $((x@\sigma) \circ z^\inv \circ (y@\sigma)) \in \sigma$ and hence $y =_\sigma x$.
  For transitivity suppose $x =_\sigma y =_\sigma z$.  In this case
  there exist $s$ and $t$ in $\sigma$ that $(x@\sigma) \circ s^\inv \circ (y@\sigma) \circ t^\inv \circ (z@\sigma)$ is defined.  But in
  this case we have $(x@\sigma) \circ (t \circ (y@\sigma)^\inv \circ s)^\inv \circ (z@\sigma)$ is defined.
  Morphoid closure implies $(s \circ (y@\sigma)^\inv \circ t) \in \sigma$ and the result follows.
\end{proof}

\begin{theorem}
  The isomorphism substitution rule

\centerline{
\unnamed
    {\ant{\Gamma \vdash \intype{\sigma,\tau}{\class}}
      \ant{\Gamma;\;\intype{x}{\sigma} \vdash \intype{e[x]}{\tau}} \ant{\Gamma \vdash w =_\sigma u}}
    {\ant{\Gamma \vdash e[w] =_\tau e[u]}.}
    }

is sound
    \label{lem:isocongruence}
\end{theorem}

\begin{proof}
Consider $\rho \in \convalue{\Gamma}$.  Let $\sigma^*$ be
$\semvalue{\sigma}\rho$ and similarly for $\tau^*$, $w^*$ and $u^*$.
For $s \in \sigma^*$ let $e^*[s]$ be
$\subvalue{\Gamma;\intype{x\;}{\;\sigma}}{e[x]}\rho[x \leftarrow s]$.
We must show that the validity of the antecedents of the rule implies
$e^*[w^*] =_{\tau^*} e^*[u^*]$.

$\Gamma \models w =_\sigma u$ and clause (20) of figure~\ref{fig:semantics} imply $w^* \in \sigma^*$ and $u^*\in \sigma^*$
and $w^*=_{\sigma^*} u^*$.  This implies that there exists $z \in \sigma^*$
such that $(w^*@\sigma^*) \circ z^\inv \circ (u^*@\sigma^*)$ is
defined.  By the definitions of figure~\ref{fig:classes} we have
$w^*@\sigma^* \in \sigma^*$ and $u^*@\sigma^* \in \sigma^*$ and by the morphoid closure of $\sigma^*$
we have  $(w^*@\sigma^*) \circ z^\inv \circ (u^*@\sigma^*) \in \sigma^*$. Since
$\Gamma;\intype{x}{\sigma} \models \intype{e[x]}{\tau}$ we have that
$\subxvalue{e[x]}$ is defined and by (11.2) for $\Gamma$ and (11.4) for $e[x]$ we have
\begin{eqnarray*}
  & & \semvalue{e[x]}\rho[x\leftarrow (w^*@\sigma^*) \circ z^\inv \circ (u^*@\sigma^*)] \\
  & = & \semvalue{e[x]}(\rho \circ \rho^\inv \circ \rho)[x\leftarrow (w^*@\sigma^*) \circ z^\inv \circ (u^*@\sigma^*)] \\
    & = & \semvalue{e[x]}(\rho[x\leftarrow w^*@\sigma^*] \circ \rho^\inv[x \leftarrow z^\inv] \circ \rho[x \leftarrow u^*@\sigma^*]) \\
  & = & \semvalue{e[x]}(\rho[x\leftarrow w^*@\sigma^*] \circ \rho[x \leftarrow z]^\inv \circ \rho[x \leftarrow u^*@\sigma^*]) \\
  & = &  e^*[w^*@\sigma^*] \circ e^*[z]^\inv \circ e^*[u^*@\sigma^*]
\end{eqnarray*}
By the morphoid closure property (9.1) for $\tau^*$ this composition
is a member of $\tau^*$.  By (8.11) we then have
\begin{eqnarray*}
  & & (e^*[w^*@\sigma^*] \circ e^*[z]^\inv \circ e^*[u^*@\sigma^*])@\tau^* \\
  & = & e^*[w^*@\sigma^*]@\tau^* \circ (e^*[z]@\tau^*)^\inv \circ e^*[u^*@\sigma^*]@\tau^*
\end{eqnarray*}
By (8.7) we have $w^* \preceq w^*@\sigma^*$ and by (11.5) for $e[x]$ we then have $e^*[w^*] \preceq e^*[w^*@\sigma^*]$. By the definition of $\preceq$ we then have $e^*[w^*@\sigma^*]@\tau^*
= e^*[w^*]@\tau^*$ and similarly for $u^*$. We now have that
$$e^*[w^*]@\tau^* \circ (e^*[z]@\tau^*)^\inv \circ e^*[u^*]@\tau^*$$
is defined which implies the result.
\end{proof}

We now turn to figure~\ref{fig:IsoRules}.  In this section we will define the semantics of the
the carrier expression.  The proof of the soundness of the rules in figure~\ref{fig:IsoRules}
is given in the last section of the appendix.

For
$\semvalue{\carrier(\sigma,\gamma,f,(\lambda\;\intype{\alpha}{\sett}\;\tau[\alpha]))}$ to be
defined we require $\Gamma \models \intype{\sigma,\gamma}{\sett}$ and $\Gamma \models
\intype{f}{\mathrm{Bijection}[\sigma,\gamma]}$ and $\Gamma; \intype{\alpha}{\sett} \models
\intype{\tau[\alpha]}{\sett}$.  To state the definition consider $\rho \in \convalue{\Gamma}$.
Let
$\sigma^*$ abbreviate $\semvalue{\sigma}\rho$ and similarly for $\tau^*$ and $f^*$.  For a set $s$
let $\tau^*[s]$ abbreviate $\subvalue{\Gamma;\;\intype{\alpha\;}{\;\sett}}{\tau[\alpha]}\rho[\alpha \leftarrow s]$.
We start with the following definition.

\begin{definition} For function value $f$ we define $Y(f)$ by
  {\small $$Y(f) = \{\pointt(\;\mathrm{Lindex}(f(u)@\pointt),\;\mathrm{Rindex}(u@\pointt)\;),\;u \in   \domop(f)\}$$}
  where
  $$\mathrm{Lindex}(\pointt(i,j)) = i\;\;\;\mbox{and}\;\;\;\;\mathrm{Rindex}(\pointt(i,j)) = j.$$
\end{definition}

For a bijection $f$ from $\sigma$ to $\tau$ we then have that $Y(f)$ is a set value with
$$\sigma@\setof(\pointt) \circ Y(f)^\inv \circ \tau@\setof(\pointt)$$ defined and hence $\sigma
=_{\sett} \tau$.  By (11.4) for $\tau[\alpha]$ we have that
$$\tau^*[\sigma^*@\setof(\pointt)] \circ \tau^*[Y(f^*)]^\inv \circ
\tau^*[\gamma^*@\setof(\pointt)]$$ is defined.  The need for the inverses can be seen in the
following more explicit derivation.
\begin{eqnarray*}
  & & \tau^*[X \circ Y^\inv \circ Z] \\
  & = & \subvalue{\Gamma;\;\intype{\alpha\;}{\;\sett}}{\tau[\alpha]}\rho[\alpha \leftarrow (X \circ Y^\inv
    \circ Z)] \\ & = & \subvalue{\Gamma;\;\intype{\alpha\;}{\;\sett}}{\tau[\alpha]}(\rho \circ
  \rho^\inv \circ \rho)[\alpha \leftarrow (X \circ Y^\inv \circ Z)] \\ & = &
  \subvalue{\Gamma;\;\intype{\alpha\;}{\;\sett}}{\tau[\alpha]}(\rho[\alpha \leftarrow X] \circ
  (\rho[\alpha \leftarrow Y])^\inv \circ \rho[\alpha \leftarrow Z]) \\ & = & \tau^*[X] \circ
  \tau^*[Y]^\inv \circ \tau^*[Z]
\end{eqnarray*}
Now let ${\mathcal T}$ be a template such that $\intype{\tau^*[Y(f^*)]}{\setof({\mathcal T})}$. Abstracting
the above equation to $\setof({\mathcal T})$ gives that
$$\tau^*[\sigma^*]@\setof({\mathcal T})\circ \tau^*[Y(f^*)]^\inv \circ \tau^*[\gamma^*]@\setof({\mathcal
  T})$$ is defined. We now make the following definition.

\begin{definition} Given sets $X$, $Y$ and $Z$ with
  $$X@\setof({\mathcal T}) \circ Y \circ Z@\setof({\mathcal T})$$ defined, we define $C(X,Y,Z)$ to be the
  bijection $g$ from $X$ to $Z$ such that for $x \in X$ there exists $y \in Y$ such that
  $$x@{\mathcal T} \circ y \circ g(x)@{\mathcal T}$$ is defined.
\end{definition}

To show that this definition is well formed we use the following lemma.

\begin{lemma}
  For set values $\sigma$ and $\tau$ with $\sigma \preceq \tau$ and with $\intype{\tau}{\setof({\mathcal T})}$
  the mapping $\{x \mapsto x@{\mathcal T},\;x\in \sigma\}$ is a bijection.
  \label{lem:setdown}
\end{lemma}

\begin{proof}
  The definition of $\preceq$ requires that $\sigma@\setof({\mathcal T}) = \tau$.  This implies that the
  mapping $\{x \mapsto x@{\mathcal T},\;x\in \sigma\}$ is onto $\tau$.  We must show that no two elements of
  $\sigma$ map to the same element of $\tau$.  We have $\intype{\sigma}{\setof({\mathcal S})}$ for some
  template ${\mathcal S}$.  Consider $x_1$ and $x_2$ with $x_1@{\mathcal T} = x_2@{\mathcal T}$.  We then have
  that $(x_1@{\mathcal T}) \circ (x_2@{\mathcal T})^\inv$ is defined where we have $\intype{x_1}{\mathcal S}$
  and $\intype{x_2}{\mathcal S}$.  By property (8.12) we then have that $x_1 \circ x_2^\inv$ is
  defined and by the bijectivity of $\sigma$ we then have $x_1 = x_2$.
\end{proof}

The well-formedness of the definition of $C(X,Y,Z)$ then follows from the above lemma
and the fact that the set $Y$ is bijective.

\begin{definition}
  \begin{eqnarray*} & & \semvalue{\carrier(\sigma,\gamma,f,(\lambda\;\intype{\alpha}{\sett}\;\tau[\alpha]))}\rho \\
    & = &  C(\tau^*[\sigma^*],\tau^*[Y(f^*)]^\inv,\tau^*[\gamma^*])
  \end{eqnarray*}
\end{definition}

We then have that the carrier function is a bijection from $\tau^*[\sigma^*]$ to $\tau^*[\gamma^*]$
as required.

\section{Extensions and Conclusions}
\label{sec:Conclusions}

This paper makes a case that the semantics of type theory, including the treatment of isomorphism, can be greatly simplified under a commitment to a classical
set-theoretic foundation for mathematics. The main result of this paper is a naive composition semantics supporting the proof system defined in figures~\ref{fig:BooleanRules},~\ref{fig:Pairs} and ~\ref{fig:IsoRules}.
The system presented here is intentionally minimal for presenting the basic ideas of naive semantics and value-based groupoid structure.  Various extensions are possible.  We will mention two such extensions here.

{\bf Higher Functions.} The most obvious extension is the introduction of higher lambda expressions and function types.  We should be able to extend the system
with the following rules.

\centerline{
$\vdash \intype{\class}{\type}$
\hspace{3ex}
\unnamed
{\ant{\Gamma \vdash \intype{\sigma}{\class}}}
{\ant{\Gamma \vdash \intype{\sigma}{\type}}}
\hspace{3ex}
\unnamed
{\ant{\Gamma \vdash \intype{\sigma}{\type}}
  \ant{\mbox{$x$ is not declared in $\Gamma$}}}
{\ant{\Gamma;\;\intype{x}{\sigma} \vdash \true}}
}
    
\centerline{
\unnamed
{\ant{\Gamma \vdash \intype{\sigma}{\type}}
  \ant{\Gamma;\;\intype{x}{\sigma} \vdash \intype{\tau[x]}{\type}}}
{\ant{\Gamma \vdash \intype{\left(\Pi_{\intype{x\;}{\;\sigma}}\;\tau[x]\right)}{\type}}}
\hspace{3ex}
\unnamed
{\ant{\Gamma \vdash \intype{\left(\Pi_{\intype{x\;}{\;\sigma}}\;\tau[x]\right)}{\type}}
  \ant{\Gamma;\;\intype{x}{\sigma} \vdash \intype{e[x]}{\tau[x]}}}
{\ant{\Gamma \vdash
           \intype{\left(\lambda\;\intype{x}{\sigma}\;e[x]\right)}{\left(\Pi_{\intype{x\;}{\;\sigma}}\;\tau[x]\right)}}}
}

This would allow one to write a lambda expression for the mapping of a topological space to its fundamental group.
It would also allow one to write a function type for the space of all ``natural maps'' from topological spaces to groups.
The semantics of higher function classes would have to be restricted to only include ``natural'' functions --- functions satisfying certain commutativity
conditions with respect to value composition.
The main technical issue, however, is formulating an appropriate generalization of evaluation property (11.6) characterizing template evaluation.  One approach is that the system of higher functions
has an associate system of higher functions on templates, functions of function of templates, and so on.  While this seems possible it also seems rather cumbersome.
Another approach is to take higher functions to be closures ---  pairs of a lambda expression and a variable interpretation.  An expression denoting a class or value
would then be guaranteed
to beta-reduce to an expression not involving higher lambda expressions and (11.6) would then be required only for class or value expressions.

{\bf ``Up to Isomorphism'' Definite Descriptions.} Mathematical objects are often only defined up to isomorphism.  For example,
{\bf the} real numbers can be defined by axioms, or as Dedekind cuts, or as Cauchy sequence and these definitions are considered equivalent
--- up to isomorphism they yield the same ordered field.  We would like to introduce the following inference rules.

~ \hfill
\unnamed
{\ant{\Gamma \vdash u =_\sigma w}}
{\ant{\Gamma \vdash \mathbf{The}(\sigma,\;u) \doteq \mathbf{The}(\sigma,\;w)}}
\hfill
\unnamed
{\ant{\Gamma \vdash \intype{u}{\sigma}}}
{\ant{\Gamma \vdash \intype{\mathbf{The}(\sigma,\;u)}{\sigma}}
  \ant{\Gamma \vdash \mathbf{The}(\sigma,\;u) =_\sigma u}}
\hfill ~

We could then talk about {\bf the} complete graph $K_5$ or {\bf the} complete bipartite graph $K_{5,5}$ as well as things like
{\bf the vector space} $\mathbb{R}^3$ or {\bf the topological space} $\mathbb{R}^3$.  We could also use such definite
descriptions to build structures.  For example, the natural numbers can be taken to be the values of the form $\mathbf{The}(\sett,\;s)$
where $s$ is a finite set.  ``Two'' is then {\bf the set} with two elements.

Of course we would want to specify a semantics for such definite descriptions.  This could involve a choice oracle for selecting a value.
But care must be taken to ensure the evaluation properties in figure~\ref{fig:EvalProps} for $\mathbf{The}(\sigma,\;u)$ in the case where $\sigma$ and $u$ contain free variables.

But independent of extensions, the main contribution of this paper is a naive set-theoretic compositional semantics of type theory supporting a treatment of isomorphism.

%\bibliographystyle{plain}
%\bibliography{found}

\vfill
\eject
\appendix

\section{Weak value groupoid properties}
\label{sec:AbsProps1}

We now consider the groupoid properties in figure~\ref{fig:GroupOps}. Recall that
figure~\ref{fig:values} defines a weak value to be an element $x$ of $U$ such that $\intype{x}{\mathcal
  T}$ for some template ${\mathcal T}$.  In this section we prove groupoid properties that hold over all
weak values.

\begin{definition} We define the size $S({\mathcal T})$ of a template ${\mathcal T}$
to be the number of nodes in the syntax tree for ${\mathcal T}$ (including the root node).  We define
the size $S(x)$ of a weak value $x$ to the minimum of $S({\mathcal T})$ over all template expressions
${\mathcal T}$ such that $\intype{x}{\mathcal T}$.
\end{definition}

Note that for a weak value pair $(x,y)$ we have that $S(x)$ and $S(y)$ are both strictly less than
$S((x,y))$.  Also for a weak value set $\sigma$ and for $x \in \sigma$ we have $S(x) < S(\sigma)$.
For a weak value function $f$ and $(x \mapsto y) \in f$ we have that $S(x)$ and $S(y)$ are both less
than $S(f)$ and also $S(\domop(f)) < S(f)$.  Most proofs will be either by structural induction on
templates or by induction on value size.

We can think of weak function values as sets of pairs.  For weak values we have that the template
${\mathcal T}_1 \rightarrow {\mathcal T}_2$ is essentially the same as the template $\setof({\mathcal T}_1
\times {\mathcal T}_2)$. For (strong) values we require that functions are functional.

\begin{lemma}[Partner Lemma.] For weak set values $\sigma$ and $\tau$ such that $\sigma \circ \tau$ is defined we have that for any $x \in \sigma$
  there exists $y \in \tau$ such that $x \circ y$ is defined and for any $y \in \tau$ there exists
  $x \in \sigma$ such that $x \circ \sigma$ is defined.  For weak functions $f$ and $g$ with $f
  \circ g$ defined and for $(x \mapsto y) \in f$ there exists $(x' \mapsto y') \in g$ with $x \circ
  x'$ and $y \circ y'$ defined and for $(x' \mapsto y') \in g$ there exists $(x \mapsto y) \in f$
  with $x \circ x'$ and $y \circ y'$ defined. \label{lem:partner}
\end{lemma}

\begin{proof}
  Since $\sigma \circ \tau$ is defined we have $\rightop(\sigma) = \leftop(\tau)$.  So for $x \in
  \sigma$ we have $\rightop(x) = \leftop(y)$ for some $y \in \tau$.  The proof for weak functions is
  similar where we think of a weak function as a set of pairs.
\end{proof}

\begin{lemma}  For a weak value $x$ with $\intype{x}{\mathcal T}$ we have $\intype{\leftop(x)}{\mathcal T}$, $\intype{\rightop(x)}{\mathcal T}$
    and $\intype{x^\inv}{\mathcal T}$ and for weak values $x$ and $y$ with $x \circ y$ defined and
    $\intype{x}{\mathcal T}$ we also have $\intype{(x \circ y)}{\mathcal T}$ and $\intype{y}{\mathcal T}$.
    \label{lem:closed1}
\end{lemma}

\begin{proof} The proof can be done by structural induction on the template ${\mathcal T}$.  For the operations of left, right and inverse the result follows directly from
  the induction hypothesis and the definition of the operations given in figure~\ref{fig:values}.
  For composition the result is immediate for Booleans and points.  For pairs we have
  $(x,y)\circ(x',y') = (x\circ x',\;y\circ y').$ If $\intype{(x,y)}{({\mathcal T}_1 \times {\mathcal T}_2)}$
  then $\intype{x}{{\mathcal T}_1}$ and $\intype{y}{{\mathcal T}_2}$.  By the induction hypothesis we then
  have $\intype{(x \circ x')}{{\mathcal T}_1}$ and $\intype{x'}{{\mathcal T}_1}$ and $\intype{(y \circ
    y')}{{\mathcal T}_2}$ and $\intype{y'}{{\mathcal T}_2}$.  This gives $\intype{(x\circ x',\;y\circ
    y')}{({\mathcal T}_1 \times {\mathcal T}_2)}$ and $\intype{(x',y')}{({\mathcal T}_1 \times {\mathcal T}_2)}$ as
  desired.

  Now consider two weak sets $\sigma$ and $\tau$ such that $\sigma \circ \tau$ is defined and with
  $\intype{\sigma}{\setof({\mathcal T})}$.  An element of $\sigma \circ \tau$ has the form $x \circ y$
  for $x \in \sigma$ and $y \in \tau$ and where we have $\intype{x}{{\mathcal T}}$.  By the induction
  hypothesis this gives that $(x \circ y)@{\mathcal T}$ and $y@{\mathcal T}$ are defined. This gives
  $\intype{\sigma \circ \tau}{\setof({\mathcal T})}$.  To show $\intype{\tau}{\setof({\mathcal T})}$ we note
  that the partner lemma \ref{lem:partner} gives us that for all $y \in \tau$ there exists $x \in
  \sigma$ with $x \circ y$ defined. The proof for a composition of weak functions is similar where
  we think of weak functions as sets of pairs.
\end{proof}

\begin{corollary} The weak values are closed under the groupoid operations.
\end{corollary}

\medskip
We define an identity value to be a weak value in which every point has the form $\pointt(i,i)$.
More formally we have the following definition.

\begin{definition}
  An identity value is either a Boolean value, a point of the form $\pointt(i,i)$, a pair of
  identity values, a weak set value whose elements are all identity values, or a weak function value
  such that for $(x \mapsto y) \in f$ we have that $x$ and $y$ are identity values.
  \label{def:idvalue}
\end{definition}

The following lemma is straightforward.

\begin{lemma}
  For any weak value $x$ we have that $\leftop(x)$ and $\rightop(x)$ are identity values and if $x$
  is an identity value then $\leftop(x) = x$ and $\rightop(x) = x$.
  \label{lem:identity}
\end{lemma}

\begin{lemma}[Domain Lemma.] For any weak function $f$ we have $\domop(f^\inv) = \domop(f)^\inv$, $\domop(\leftop(f)) = \leftop(\domop(f))$  and\newline
  $\domop(\rightop(f)) = \rightop(\domop(f))$.
  Also, for any weak function values $f$ and $g$ with $f \circ g$ defined we have $\domop(f \circ g)
  = \domop(f) \circ \domop(g)$.
\end{lemma}

\begin{proof}
  The case of inverse follows from the duality of left and right.  The cases of $\leftop(f)$ and
  $\rightop(f)$ are immediate from the definition.  For the case of composition we note that
  $\rightop(f) = \leftop(g)$ implies that for every pair $(x \mapsto y) \in f$ we have that the pair
  $\leftop(x) \mapsto \leftop(y))$ is equal to some pair $\rightop(x') \mapsto \rightop(y'))$ for
  $(x' \mapsto y') \in g$.  This implies that for every $x \in \domop(f)$ we have that $f \circ g$
  contains a pair of the form $x \circ x' \mapsto y \circ y'$ and hence $\domop(f \circ g) =
  \domop(f) \circ \domop(g)$.
\end{proof}

The following lemma follows from the duality of left and right.

\begin{lemma}[Property (7.3)] For any weak value $x$ we have\newline
  $\leftop(x^\inv) = \rightop(x)$ and $\rightop(x^\inv) = \leftop(x)$
  \end{lemma}

\begin{lemma}[Property (7.4)] For weak values $x$ and $y$ we have
  $\leftop(x \circ y) = \leftop(x)$ and $\rightop(x \circ y) = \rightop(y)$.
\end{lemma}

\begin{proof}
  The proof is by induction on value size. The case of Booleans and points is immediate.  The case
  of pairs follows straightforwardly from the induction hypothesis.

  We first show that for sets $\sigma$ and $\tau$ we have that $\leftop(\sigma \circ \tau) =
  \leftop(\sigma)$.  By the partner lemma \ref{lem:partner} for every $x \in \sigma$ that exists $y
  \in \tau$ with $x \circ y$ defined.  We then have that that $\leftop(\sigma \circ \tau)$ equals
  the set of values of the form $\leftop(x \circ y)$ for $x \in \sigma$ and $y \in \tau$ which by
  the induction hypothesis equals $\leftop(x)$ for $x \in \sigma$ which equals $\leftop(\sigma)$.
  The proof for functions is similar.
\end{proof}

\begin{lemma}[Property (7.5)] For weak values $x$, $y$ and $z$ we have $(x \circ y) \circ z = x \circ (y \circ z)$.
\end{lemma}

\begin{proof}
  The proof is by induction on value size.  We will consider sets; the proof functions is similar.

  Consider sets $\sigma$, $\tau$ and $\gamma$ with with $(\sigma \circ \tau) \circ \gamma$ defined.
  By property (7.4) proved above we have that $\sigma \circ (\tau \circ \gamma)$ is also defined.
  The induction hypothesis implies that for $x \in \sigma$, $y \in \tau$ and $z \in \gamma$ we have
  $(x \circ y) \circ z = x \circ (y \circ z)$ which implies the result.
\end{proof}

{\bf Properties (7.8) and (7.9)} follow for all weak values from the duality of left and right.

\section{Value groupoid properties}

\begin{lemma}[Property (7.1)]
For a value $x$ we have that $x^\inv$, $\leftop(x)$ and $\rightop(x)$ are also values.
\end{lemma}

\begin{proof}  For $x^\inv$ the result follows from the duality of left and right.  We consider the case of $\leftop(x)$.
  The proof is by induction on the size of $x$.  For Booleans and points the result is immediate.
  For pairs the result follows directly from the induction hypothesis.  For a set value $\sigma$ we
  have that every member of $\leftop(\sigma)$ has the form $\leftop(x)$ for $x \in \sigma$ and by
  the induction hypothesis this is a value.  We must also check that $\leftop(\sigma)$ is bijective.
  But this follows from the fact that every element $z \in \leftop(\sigma)$ is an identity value
  (definition~\ref{def:idvalue}) and for any identity value $z$ we have $\leftop(z) = \rightop(z) =
  z$.  For a function value $f$ we must show that $\leftop(f)$ is functional --- no two pairs of $f$
  have the same input value.  But since $\domop(f)$ is bijective, for each value $z \in \domop(f)$
  there is a unique $x \in \domop(f)$ with $z = \leftop(x)$.  Since $f$ is functional there is a
  unique pair $(x \mapsto y)$ in $f$ with $z = \leftop(x)$.  This implies that there is a unique
  pair in $\leftop(f)$ with input value $z$.
\end{proof}

\begin{lemma}[Property (7.2)]
For two values $x$ and $y$ with $x \circ y$ defined we have that $x \circ y$ is a value.
\end{lemma}

\begin{proof}
The proof is by induction on value size.  The result is immediate for Boolean values and points and
follows straightforwardly from the induction hypothesis for pairs.

Consider two set values $\sigma$ and $\tau$ such that $\sigma \circ \tau$ is defined.  By the
induction hypothesis for $x \circ y \in \sigma \circ \tau$ we have that $x \circ y$ is a value.  It
remains to show that $\sigma \circ \tau$ is bijective.  We first note that for $z \in \sigma \circ
\tau$ there exists unique $x \in \sigma$ and $y \in \tau$ with $z = x \circ y$ --- there is a unique
$x \in \sigma$ with $\leftop(x) = \leftop(z)$ and a unique $y \in \tau$ with $\rightop(y) =
\rightop(z)$.  Hence for $z,z' \in \sigma \circ \tau$ with $z \not = z'$ we have $z = x\circ y$ and
$z' = x' \circ y'$ with $x \not = x'$ and $y \not = y'$ which implies that $\leftop(z) \not =
\leftop(z')$ and $\rightop(z) \not = \rightop(z')$.  Hence $\sigma \circ \tau$ is bijective.

Now consider two function values with $f \circ g$ defined. The induction hypothesis implies that for
every input-output pair $x \circ y \mapsto f(x) \circ f(y)$ we have that $x \circ y$ and $f(x) \circ
f(y)$ are values.  By the domain lemma we also have that $\domop(f \circ g)$ equals $\domop(f) \circ
\domop(g)$ which by the induction hypothesis is a set value.  Finally we must show that $f \circ g$
is functional.  But for $z \in \domop(f\circ g) = \domop(f) \circ \domop(g)$ we have $(f \circ g)(z)
= f(x) \circ g(y)$ where $x$ and $y$ are the unique values in $\domop(f)$ and $\domop(g)$ with $z =
x\circ y$.
\end{proof}

\begin{lemma}[Property (7.6)] For values $x$ and $y$ we have $x^\inv \circ x \circ y = y$ and $x \circ y \circ y^\inv = x$.
\end{lemma}

\begin{proof}
  The proof is by induction on value size.  We consider sets; the proof for functions is similar.
  Consider sets $\sigma$ and $\tau$ with $\sigma \circ \tau$ defined.  We have that $\sigma^\inv
  \circ \sigma \circ \tau$ is the set of values of the form $x_1^\inv \circ x_2 \circ y$ with
  $x_1,x_2 \in \sigma$ and $y \in \tau$.  But since no two values in $\sigma$ have the same right
  value we must have that $x_2 = x_1$.  By the induction hypothesis we then have that $x_1^\inv
  \circ x_1 \circ y = y$.  By the partner lemma \ref{lem:partner} for every $y \in \tau$ there
  exists $x \in \sigma$ with $x \circ y$ defined.  These facts together imply that the set of
  instances of $\sigma^\inv \circ \sigma \circ \tau$ are exactly the members of $\tau$ which proves
  the result.
\end{proof}

\begin{lemma}[Property (7.7)] For a value $x$ we have $\rightop(x) = x^\inv \circ x$ and $\leftop(x) = x \circ x^\inv$
\end{lemma}

\begin{proof}
  The proof is by induction on value size.  We consider sets.  Consider a set $\sigma$.  We have
  that $\leftop(\sigma)$ equals the set of values of the form $\leftop(x)$ for $x \in \sigma$.  By
  the induction hypothesis this is the same as the set of values of the form $x \circ x^\inv$.
  Since no two members of $\sigma$ have the same left value, this is the same as the set of values
  of the form $x_1 \circ x_2^\inv$ for $x_1,x_2 \in \sigma$ but this is the same as the set of
  values of $\sigma \circ \sigma^\inv$. Other cases are similar.
\end{proof}

\section{The abstraction properties}
\label{sec:AbsProps2}

All of the abstraction properties hold for weak values with the exception of (8.1) which states that
for any (strong) value $x$ with $x@{\mathcal T}$ defined we have that $x@{\mathcal T}$ is a (strong) value.
We prove (8.1) at the end of this section after showing the other properties for weak values.

The following lemma can be proved by a straightforward structural induction on ${\mathcal T}$.

\begin{lemma} For a weak value $x$ with $x@{\mathcal T}$ defined we have
  $\intype{(x@{\mathcal T})}{\mathcal T}$.
   \label{lem:closed2}
\end{lemma}

We note that lemma~\ref{lem:closed2} implies that the weak values are closed under abstraction ---
if $x$ is a weak value, and $x@{\mathcal T}$ is defined, then $x@{\mathcal T}$ is a weak value.

The following lemma can also be proved by a straightforward structural induction on ${\mathcal T}$.

\begin{lemma}[Property (8.2)] For any weak value $x$ we have that $x@{\mathcal T} = x$ if and only if $\intype{x}{\mathcal T}$.
\end{lemma}

\begin{lemma} For a weak value $x$ with $x@{\mathcal T}$ defined we have that
    $\leftop(x)@{\mathcal T}$, $\rightop(x)@{\mathcal T}$
    and $x^\inv@{\mathcal T}$ are all defined and for weak values $x$ and $y$ with $x \circ y$ defined
    and with $x@{\mathcal T}$ defined we also have $(x \circ y)@{\mathcal T}$ and $y@{\mathcal T}$ are defined.
    \label{lem:closed1a} 
\end{lemma}

\begin{proof} The proof is by structural induction on ${\mathcal T}$.  For the operations of left, right and inverse the result follows directly from
  the induction hypothesis and the definition of the operations given in figure~\ref{fig:values}.
  For composition the result is immediate for Booleans and points and follows directly from the
  induction hypothesis for pairs.  Now consider two weak sets $\sigma$ and $\tau$ such that $\sigma
  \circ \tau$ is defined and with $\sigma@\setof({\mathcal T})$ defined.  An element of $\sigma \circ
  \tau$ has the form $x \circ y$ for $x \in \sigma$ and $y \in \tau$ and where we have $x@{\mathcal T}$
  defined.  By the induction hypothesis this gives that $(x \circ y)@{\mathcal T}$ is defined and that
  $y@{\mathcal T}$ is defined. This gives that $(\sigma \circ \tau)@\setof({\mathcal T})$ is defined.  To
  show that $\tau@\setof({\mathcal T})$ is defined we note that the partner lemma \ref{lem:partner}
  gives us that for all $y \in \tau$ there exists $x \in \sigma$ with $x \circ y$ defined. The proof
  for a composition of weak functions (which are not required to be functional) is similar.
\end{proof}

Recall that for weak values $x$ and $y$ figure~\ref{fig:Abstraction} define $x \sim y$ to mean that
there exists a value $z$ with $x \circ z \circ y$ defined.  The following two properties are
corollaries of lemmas~\ref{lem:closed1} and~\ref{lem:closed1a} respectively.

\begin{lemma}[Property (8.4)] For a weak value $x$ with $\intype{x}{\mathcal T}$ and $x \sim y$ we have $\intype{y}{\mathcal T}$.
\end{lemma}

\begin{lemma}[Property (8.5)] For a weak value $x$ with $x@{\mathcal T}$ defined and $x \sim y$ we have $y@{\mathcal T}$ defined.
\end{lemma}

\begin{lemma}[Property (8.3)] $\sim$ is an equivalence relation on weak values and for any weak value $x$ we have $x \sim x^\inv \sim \leftop(x) \sim \rightop(x)$ and for weak values $x$ and $y$ with $x \circ y$
  defined we have $x \sim (x \circ y) \sim y$.
\end{lemma}

\begin{proof}
  First we show that $\sim$ is an equivalence relation. Since $x \circ x^\inv \circ x$ is defined we
  have $x \sim x$.  Given $x \sim y$ we have $x \circ z \circ y$ is defined for some $z$.  We then
  have that $y \circ (y^\inv \circ z^\inv \circ x^\inv) \circ x$ is defined giving $y \sim x$.
  Finally assume $x \sim y$ and $y \sim z$.  We then have that $x \circ w \circ y$ is defined and $y
  \circ s \circ z$ is defined.  We then have that $x \circ (w \circ y \circ s) \circ z$ is defined
  giving $x \sim z$.

  Next we note that if $x \circ y$ is defined then $x \circ (x^\inv \circ x) \circ y$ is defined
  giving $x \sim y$.  Since $x\circ x^\inv$ is defined we have $x \sim x^\inv$ and similarly $x \sim
  \leftop(x)$, $x \sim \rightop(x)$, $(x \circ y) \sim y^\inv \sim y$ and $x^\inv \sim (x \circ y) \sim x$.
\end{proof}

We now prove some lemmas supporting (8.6).

\begin{lemma} For any identity value $x$ such that $x@{\mathcal T}$ is defined we have that $\pointify(x@{\mathcal T}) = \pointify(x)$.
\end{lemma}

\begin{proof} The proof is a straightforward structural induction on ${\mathcal T}$. The result is immediate for ${\mathcal T} = \bool$ or
  ${\mathcal T} = \pointt$. We explicitly consider the case of ${\mathcal T} = \setof({\mathcal T}')$.
  \begin{eqnarray*}
    \pointify(\sigma@\setof({\mathcal T}')) & = & \pointt\left(\begin{array}{l} \subpoint(\sigma@\setof({\mathcal T}')), \\ \subpoint(\sigma@\setof({\mathcal T}') \end{array}\right) \\
    \subpoint(\sigma@\setof({\mathcal T}')) & = & \{\pointify(x@{\mathcal T}'),\;x\in \sigma\} \\
    & = & \{\pointify(x),\;x\in \sigma\} \\
    & = & \subpoint(\sigma) \\
    \pointify(\sigma@\setof({\mathcal T}')) & = & \pointt(\subpoint(\sigma),\;\subpoint(\sigma)) \\
    & = & \pointify(\sigma)
    \end{eqnarray*}
\end{proof}

\begin{corollary} For a weak value $x$ with $x@{\mathcal T}$ defined we have $(x@{\mathcal T})@\pointt = x@\pointt$.  \label{cor:atpoint}
\end{corollary}

\begin{lemma}[Property (8.6)]  For a weak value $x$ with $(x@{\mathcal T}_1)@{\mathcal T}_2$ defined we have that $(x@{\mathcal T}_1)@{\mathcal T}_2 = x@{\mathcal T}_2$.
\end{lemma}

\begin{proof}
  The proof is by structural induction on ${\mathcal T}_2$.  The case of ${\mathcal T}_2 = \pointt$ is
  handled by corollary~\ref{cor:atpoint}.  The other cases are straightforward.
\end{proof}

\begin{lemma}[Property (8.7)] For $x@{\mathcal T}$ defined we have $x \preceq x@{\mathcal T}$.
\end{lemma}

\begin{proof}
  Suppose that $x@{\mathcal T}@{\mathcal S}$ is defined. By (8.7) proved above we have $x@{cal T}@{\mathcal S} = x@{\mathcal S}$. The result then follows from the definition of $x \prec x@{\mathcal T}$.
\end{proof}

\begin{corollary}[Property (8.8)] $\preceq$ is a partial order on weak values.
\end{corollary}

\begin{lemma}[Property (8.9)] For any weak value $x$ with $(x@{\mathcal T}_1)@{\mathcal T}_2$ and $(x@{\mathcal T}_2)@{\mathcal T}_1$ both defined we have
  $x@{\mathcal T}_1 = x@{\mathcal T}_2$.
\end{lemma}

\begin{proof}
  The proof is by induction on the size of $x$.  The result is immediate if $x$ is a point or
  Boolean value and follows straightforwardly form the induction hypothesis if $x$ is a pair.  Now
  consider a weak set $\sigma$ and suppose that $(\sigma@\setof{\mathcal T}_1))@\setof({\mathcal T}_2)$
  and $(\sigma@\setof{\mathcal T}_12))@\setof({\mathcal T}_1)$ are both defined.  For each $x \in \sigma$ we then
  have that both $(x@{\mathcal T}_1)@{\mathcal T}_2$ and $(x@{\mathcal T}_2)@{\mathcal T}_1$ are defined.  By the
  induction hypothesis we then have that $x@{\mathcal T})_1 = x@{\mathcal T}_2$ But this implies that
  $\sigma@\setof({\mathcal T}_1) = \sigma@\setof({\mathcal T}_2)$. The case of weak function values is
  similar.
\end{proof}

The following lemma is immediate form the duality of left and right.

\begin{lemma}[Property (8.10)] For any weak value $x$ with $x@{\mathcal T}$ defined we have $x^\inv@{\mathcal T} = (x@{\mathcal T})^\inv$.
\end{lemma}

We will now prove a series of lemmas supporting property (8.11).

\begin{lemma}[$@\pointt$ commutes with left and right] For any weak value $x$ we have that $\leftop(x)@\pointt = \leftop(x@\pointt)$ and $\rightop(x)@\pointt = \rightop(x@\pointt)$  \label{lem:atpoint1}
\end{lemma}

\begin{proof}  The result is immediate in the case that $x$ is a point.  If $x$ is not a point then we have
  \begin{eqnarray*}
    \leftop(x)@\pointt & = & \pointt\left(\begin{array}{l}\subpoint(\leftop(\leftop(x))), \\ \subpoint(\rightop(\leftop(x)))\end{array}\right) \\
    & = & \pointt(\subpoint(\leftop(x)),\subpoint(\leftop(x))) \\
    & = & \leftop(\pointt(\subpoint(\leftop(x)),\subpoint(\rightop(x)))) \\
    & = & \leftop(x@\pointt)
  \end{eqnarray*}
\end{proof}

\begin{definition}[Abstraction Inverse]  We define the partial operation $x \downarrow {\mathcal T}$ for an identity value $x$ and template ${\mathcal T}$ by the following rules
  where the operation is undefined if no rule applies or the right hand side of the applicable rule
  is undefined.
  \begin{eqnarray*}
    \pointt(\Phi,\Phi)\downarrow\bool & = & \Phi \\
    \pointt((x,y),(x,y))\downarrow({\mathcal T}_1 \times {\mathcal T}_2) & = & (x\downarrow {\mathcal T}_1,\;y\downarrow{\mathcal T}_2) \\
    \pointt(\sigma,\sigma)\downarrow\setof({\mathcal T}) & = & \{x\downarrow {\mathcal T},\;x \in \sigma\} \\
    \pointt(f,f)\downarrow({\mathcal T}_1\rightarrow {\mathcal T}_2) & = & \{x\downarrow {\mathcal T}_1 \mapsto y\downarrow {\mathcal T}_2,\;(x \mapsto y) \in f\} \\
    \\
    \Phi \downarrow \bool & = & \Phi \\ \pointt(i,i)\downarrow\pointt & = & \pointt(i,i) \\
    (x,y) \downarrow ({\mathcal T}_1 \times {\mathcal T}_2) & = & (x\downarrow {\mathcal T}_1,\;y\downarrow {\mathcal T}_2) \\
    \sigma \downarrow \setof({\mathcal T}) & = & \{x \downarrow {\mathcal T},\;x \in \sigma\} \\
    f \downarrow {\mathcal T}_1 \rightarrow {\mathcal T}_2 & = & \{x\downarrow {\mathcal T}_1 \mapsto y \downarrow {\mathcal T}_2,\;(x \mapsto y) \in f\}
  \end{eqnarray*}
\end{definition}

\begin{lemma} Abstraction of identity values is invertible.  More specifically, for any identity value $x$ with $\intype{x}{\mathcal T}$ and with $x@{\mathcal T}'$ defined we have $(x@{\mathcal T}')\downarrow {\mathcal T} = x$.
\end{lemma}

\begin{proof}
  The proof is by induction on the template ${\mathcal T}$.  Note in particular that for
  $\intype{x}{\mathcal T}$ we must show that $(x@\pointt)\downarrow {\mathcal T} = x$.  We omit the details.
\end{proof}

\ignore{
\begin{lemma}[Abstraction Distributes Out] If $(x@ \pointt)\circ(y@\pointt)$ is defined, and $\intype{x}{\mathcal T}$ and $\intype{y}{\mathcal T}$, then $x \circ y$ is defined
  and hence $(x@\pointt)\circ (y@\pointt) = (x \circ y)@\pointt$.
  \label{lem:distout} 
\end{lemma}

\begin{proof}
  By lemma~\ref{lem:distin} it suffices to show that $x \circ y$ is defined.  We are given that
  $(x@\pointt) \circ (y@\pointt)$ is defined which implies $\rightop(x@\pointt) =\leftop(y@\pointt)$
  which by lemma~\ref{lem:atpoint1} implies $\rightop(x)@\pointt = \leftop(y)@\pointt$.  But we then
  have
  $$\rightop(x) = (\rightop(x)@\pointt)\downarrow {\mathcal T} = (\leftop(y)@\pointt)\downarrow {\mathcal T}
  = \leftop(y)$$ which proves the result.
\end{proof}
}

\begin{lemma}[Property (8.12)] For any weak values $x$ and $y$ with $\intype{x}{\mathcal T}$ and $\intype{y}{\mathcal T}$ and with $(x@{\mathcal T}') \circ (y@{\mathcal T}')$ defined
  we have that $x \circ y$ is defined. \label{lem:distout}
\end{lemma}

\begin{proof}
  \begin{eqnarray*}
    \rightop(x) & = & (\rightop(x)@{\mathcal T}') \downarrow {\mathcal T} \\ & = & \rightop(x@{\mathcal T}')
    \downarrow {\mathcal T} \\ & = & \leftop(y@{\mathcal T}') \downarrow {\mathcal T} \\ & = & (\leftop(y)@{\mathcal
      T}') \downarrow {\mathcal T} \\ & = & \leftop(y)
  \end{eqnarray*}
\end{proof}

\begin{lemma}[Property (8.11)] For any weak values $x$ and $y$ with $(x \circ y)@{\mathcal T}$ defined we have
  $(x \circ y)@{\mathcal T} = (x@{\mathcal T}) \circ (y@{\mathcal T})$.
\end{lemma}

\begin{proof}  The proof is by induction on the template ${\mathcal T}$. The result is immediate for Booleans.
  For points we have the following calculation.
  \begin{eqnarray*}
    (x \circ y)@\pointt & = & \pointt(\subpoint(\leftop(x \circ y)), \subpoint(\rightop(x \circ y)))
    \\ & = & \pointt(\subpoint(\leftop(x)), \subpoint(\rightop(y))) \\ & = &
    \pointt(\subpoint(\leftop(x)), \subpoint(\rightop(x))) \\ & & \circ
    \pointt(\subpoint(\leftop(y)), \subpoint(\rightop(y))) \\ & = & (x@\pointt) \circ (y@\pointt)
  \end{eqnarray*}
  For pairs the result follows straightforwardly from the induction hypothesis.
  
  Now consider weak sets $\sigma$ and $\tau$ with $(\sigma \circ \tau)@\setof({\mathcal T})$ defined.
  We have
  \begin{eqnarray*}
    (\sigma \circ \tau)@\setof({\mathcal T}) & = & \{(x \circ y)@{\mathcal T},\;x \in \sigma,\;y \in \tau\}
    \\ & = & \{(x@{\mathcal T}) \circ (y@{\mathcal T}),\;\;x \in \sigma,\;y \in \tau\} \\ & = &
    (\sigma@\setof({\mathcal T})) \circ (\tau@\setof({\mathcal T}))
  \end{eqnarray*}
  The validity of the second line is subtle.  The induction hypothesis implies that every value of
  the form $(x \circ y)@{\mathcal T}$ can be written as $(x@{\mathcal T}) \circ (y@{\mathcal T})$.  But for the
  equality of the sets we must also have that every value of the form $(x@{\mathcal T}) \circ (y@{\mathcal
    T})$ can be written as $(x \circ y)@{\mathcal T}$.  Since $\sigma$ is a weak value we have
  $\intype{\sigma}{\setof({\mathcal T}')}$ for some ${\mathcal T}'$ which implies
  $\intype{\tau}{\setof({\mathcal T}')}$ which implies $\intype{x,y}{{\mathcal T}'}$.  By
  lemma~\ref{lem:distout} we then have that if $(x@{\mathcal T}) \circ (y@{\mathcal T})$ is defined then $x
  \circ y$ is defined and by the induction hypothesis $(x@{\mathcal T}) \circ (y@{\mathcal T}) = (x\circ
  y)@{\mathcal T}$.

  The case of functions is similar.
\end{proof}

Finally we prove the only abstraction property specific to (strong) values.

\begin{lemma}[Property (8.1)] Values are closed under abstraction. More specifically, for a value $x$ with $x@{\mathcal T}$ defined we have
  that $x@{\mathcal T}$ is a value.
\end{lemma}

\begin{proof}
  The proof is by structural induction on the template ${\mathcal T}$. The result is immediate for
  Booleans and points and follows straightforwardly from the induction hypothesis for pairs.

  Consider a set value $\sigma$ and a template ${\mathcal T}$ such that $\sigma@\setof({\mathcal T})$ is
  defined.  We have already proved abstractions of weak values are weak values and by the induction
  hypothesis we have that $x@{\mathcal T}$ is a value for each $x \in \sigma$.  It remains only to show
  that $\sigma@\setof({\mathcal T})$ is bijective.  It suffices to show that for $x,y \in \sigma$ we
  have that $\leftop(x@{\mathcal T}) = \leftop(y@{\mathcal T})$ implies that $x = y$ and similarly for
  $\rightop$. But $\leftop(x@{\mathcal T}) = \leftop(y@{\mathcal T})$ implies that $(x@{\mathcal T}) \circ
  (y^\inv@{\mathcal T})$ is defined and lemma~\ref{lem:distout} then implies that $x \circ y^\inv$ is
  defined which implies $\leftop(x) = \leftop(y)$ which implies $x=y$. The case of $\rightop$ is
  similar.

  Now consider a function value $f$ with $f@({\mathcal T}_1 \rightarrow {\mathcal T}_2)$ defined.  Again we
  have already shown that an abstraction of weak value is a weak value and by the induction
  hypothesis we have that for $(x \mapsto y) \in f$ we have that $x@{\mathcal T}_1$ and $y@{\mathcal T}_2$
  are values.  The induction hypothesis also gives us that $\domop(f@({\mathcal T}_1 \rightarrow {\mathcal
    T}_2)) = \domop(f)@\setof({\mathcal T}_1)$ is a set value and hence is bijective.  Consider
  $(x@{\mathcal T}_1 \mapsto y@{\mathcal T}_2)$ and $(x'@{\mathcal T}_1 \mapsto y'@{\mathcal T}_2)$ in $f@({\mathcal
    T}_1 \rightarrow {\mathcal T}_2)$.  To show that $f@({\mathcal T}_1 \rightarrow {\mathcal T}_2)$ is
  functional it suffices to show that $x@{\mathcal T}_1 = x'@{\mathcal T}_1$ implies $x = x'$.  Since
  $\domop(f)$ is bijective it suffices to show that $\leftop(x) = \leftop(x')$.  Since $\domop(f)$
  is a weak set value we have $\intype{\domop(f)}{\setof({\mathcal T}')}$ for some template ${\mathcal T}'$
  which implies $\intype{x,x'}{{\mathcal T}'}$.  But we have $x@{\mathcal T}_1 = x'@{\mathcal T}_1$ implies that
  $((x')^\inv@{\mathcal T}_1) \circ (x@{\mathcal T}_1)$ is defined and by lemma~\ref{lem:distout} we have
  that $(x')^\inv \circ x$ is defined which implies $\leftop(x') = \leftop(x)$.
\end{proof}

\section{The groupoid properties for classes}
\label{sec:ClassProofs}

\begin{lemma} Every set value is a class.  Furthermore, for a set value $\sigma$ with $\intype{\sigma}{\setof({\mathcal T})}$ we have that ${\mathcal T}$ is an interface template for $\sigma$
  and for $x,y \in \sigma$ we have $x =_\sigma y$ if and only if $x=y$. \label{lem:setclass}
\end{lemma}

\begin{proof}
  We note that $\intype{\sigma}{\setof({\mathcal T})}$ implies that for $x\in \sigma$ we have
  $\intype{x}{\mathcal T}$ which implies that $x@{\mathcal T} = x$.  So for a set value $\sigma$ and $x \in
  \sigma$ we have $x@\sigma = x$.  The morphoid closure condition for $\sigma$ follows from the
  bijectivity of $\sigma$ which implies that for $x,y,z \in \sigma$ with $x \circ z^\inv \circ y$
  defined we must have $x = z = y$.  This also implies that $x =_\sigma y$ if and only if $x = y$.
\end{proof}

\begin{lemma} The definitions of the groupooid operations and the abstraction ordering on set values agree with the more general definitions for classes.
\end{lemma}

\begin{proof}
  Inverse and composition are defined the same way sets and classes.  To show that the two
  definitions of $\leftop$ agree it suffices to show that the for a set value $\sigma$ we have that the
  set of values of the form $x_1 \circ x_2^\inv$ for $x_1,x_2 \in \sigma$ is the same as the set of
  values of the form $\leftop(x)$ for $x \in \sigma$.  But by property (7.7) we have that
  $\leftop(x) = x \circ x^\inv$ and by the bijectivity of $\sigma$ for $x_1 \circ x_2^\inv$ defined
  we must have $x_1 = x_2$.  The proof for $\rightop$ is similar.

  The equivalence of the abstraction order definition follows from the fact that for set values
  $\sigma$ and $\tau$ we have that $\{x@\tau,\;x \in \tau\} = \tau$ and $\{x@\tau,\;x \in \sigma\} =
  \sigma@\setof({\mathcal T})$ where we have $\intype{\tau}{\setof({\mathcal T})}$.
\end{proof}

\begin{lemma}[(7.1) for classes] The Morphopid Classes are closed under $(\cdot)^\inv$, $\leftop$ and $\rightop$
  and any interface template for $\sigma$ is also an interface template for $\sigma^\inv$,
  $\leftop(\sigma)$ and $\rightop(\tau)$. \label{lem:InterfaceTransfer1}
\end{lemma}

\begin{proof}
The case of inverse follows from the duality of left and right.  We will show that $\leftop(\sigma)$
is a class.  We must show that $\leftop(\sigma)$ satisfies the morphoid closure condition (9.1) and
has an interface template as required by (9.2).  We let $x$ range over members of $\sigma$.  The
elements of $\leftop(\sigma)$ are (strong) the values of the form $x_1 \circ x_2^\inv$.

We first consider morphoid closure. Suppose that $(x_1 \circ x_2^\inv) \circ (x_3 \circ
x_4^\inv)^\inv \circ (x_5 \circ x_6^\inv)$ is defined.  By the groupoid properties of values we have
$$(x_1 \circ x_2^\inv) \circ (x_3 \circ x_4^\inv)^\inv \circ (x_5 \circ x_6^\inv) = (x_1 \circ
x_2^\inv \circ x_4) \circ (x_6 \circ x_5^\inv \circ x_3)^\inv$$ which proves morphoid closure.

Next we consider (9.2) --- the existence of an interface template.  Let ${\mathcal T}$ be an interface
template for $\sigma$ and let $x_1 \circ x_2^\inv$ be an element of $\leftop(\sigma)$.  We have
$x_1@{\mathcal T}$ is defined and $x_1 \sim (x_1 \circ x_2^\inv)$ so we have that $(x_1 \circ
x_2^\inv)@{\mathcal T}$ is defined and equal to $(x@{\mathcal T}) \circ (x_2@{\mathcal T})^\inv \in
\leftop(\sigma)$.
\end{proof}

\begin{lemma}[(7.2) for classes]
  For morphoid classes $\sigma$ and $\tau$ with $\sigma \circ \tau$ defined we have that $\sigma
  \circ \tau$ is a class and any interface template for $\sigma$ or $\tau$ is an interface
  template for $\sigma$, $\tau$ and $\sigma \circ \tau$. \label{lem:InterfaceTransfer2}
\end{lemma}

\begin{proof}
Again we must show (9.1) and (9.2).  We let $x$ range over elements of $\sigma$ and $y$ range over
elements of $\tau$.

The elements of $\sigma \circ \tau$ are the values of the form $x \circ y$.  We must show that for
$(x_1\circ y_1)\circ (x_2 \circ y_2)^{-1} \circ (x_3 \circ y_3)$ defined we have that this
composition is in $\sigma \circ \tau$.  Since $\rightop(\sigma) = \leftop(\tau)$, every value of the
form $y_1 \circ y_2^\inv$ can be written as $x_1^\inv \circ x_2$.  We then have the following.
\begin{eqnarray*}
  & & (x_1\circ y_1)\circ (x_2 \circ y_2)^{-1} \circ (x_3 \circ y_3) \\
  & = & x_1\circ (y_1 \circ y_2^{-1})  \circ x_2^{-1} \circ x_3 \circ y_3 \\
  & = & x_1\circ (x_4^{-1} \circ x_5) \circ x_2^{-1} \circ x_3  \circ y_3 \\
  & = & ((x_1\circ x_4^{-1} \circ x_5) \circ x_2^{-1} \circ x_3) \circ y_3 \\
  & = & x_7  \circ y_3
\end{eqnarray*}
We must also show (9.2).  Let ${\mathcal T}_1$ be an interface template for $\sigma$ and let ${\mathcal
  T}_2$ be an interface template for $\tau$.  For any $x \in \sigma$ we have $\rightop(x@{\mathcal T}_1)
= (x@{\mathcal T}_1)^\inv \circ (x@{\mathcal T}_1) \in \rightop(\sigma) = \leftop(\tau)$.  By the preceding
lemma we have that ${\mathcal T}_2$ is an interface template for $\leftop(\tau)$ which implies that
$((x@{\mathcal T}_1)^\inv \circ (x@{\mathcal T}_1))@{\mathcal T}_2$ is defined which implies that $(x@{\mathcal
  T}_1)@{\mathcal T}_2$ is defined.  But we also have $\rightop(x) \in \leftop(\tau)$ which implies that
$\rightop(x)@{\mathcal T}_2 \in \leftop(\tau) = \rightop(\sigma)$ which gives that $(\rightop(x)@{\mathcal
  T}_2)@{\mathcal T}_1$ is defined which implies that $(x@{\mathcal T}_2)@{\mathcal T}_1$ is defined.  We then
have that $x@{\mathcal T}_1 = x@{\mathcal T}_2$.  Similarly, for $y \in \tau$ we have $y@{\mathcal T}_2 =
y@{\mathcal T}_1$.  This implies that both ${\mathcal T}_1$ and ${\mathcal T}_2$ are interface templates for
both $\sigma$ and $\tau$ and we have $(x \circ y)@{\mathcal T}_1 = (x@{\mathcal T}_1) \circ (y@{\mathcal T}_1)
\in \sigma \circ \tau$ and similarly for ${\mathcal T}_2$.  So both ${\mathcal T}_1$ and ${\mathcal T}_2$ are
also interface templates for $\sigma \circ \tau$.

\end{proof}

\begin{lemma}[(7.3) for classes] $\leftop(x^\inv) = \rightop(x)$ and $\rightop(x^\inv) = \leftop(x)$.
\end{lemma}

\begin{proof}
This is a consequence of the duality of left and right.
\end{proof}

\begin{lemma}[(7.4) for classes]  $\leftop(\sigma \circ \tau) = \leftop(\sigma)$ and $\rightop(\sigma \circ \tau) = \rightop(\tau)$.
\end{lemma}

\begin{proof}
We will show $\leftop(\sigma \circ \tau) = \leftop(\sigma)$.  We will use $x$ to range over elements
of $\sigma$ and $y$ range over elements of $\tau$.  We first show that every member of
$\leftop(\sigma \circ \tau)$ is an member of $\leftop(\sigma)$.  A member of $\leftop(\sigma \circ
\tau)$ has the form $(x_1 \circ y_1) \circ (x_2 \circ y_2)^{-1}$.  Since $\rightop(\sigma) =
\leftop(\tau)$ we have that every value of the form $y_1^\inv \circ y_2$ can be written as $x_1
\circ x_2^\inv$.  We then have:
\begin{eqnarray*}
(x_1 \circ y_1) \circ (x_2 \circ y_2)^{-1} & = & x_1 \circ (y_1 \circ y_2^{-1}) \circ x_2^{-1} \\ &
  = & x_1 \circ (x_3^{-1} \circ x_4) \circ x_2^{-1} \\ & = & (x_1 \circ x_3^{-1} \circ x_4) \circ
  x_2^{-1} \in \leftop(\sigma)
\end{eqnarray*}
For the converse we consider a value $x_1 \circ x_2^{-1}$ in $\leftop(\sigma)$. For this we have the
following.
\begin{eqnarray*}
x_1 \circ x_2^{-1} & = & x_1 \circ x_2^{-1} \circ x_2 \circ x_2^{-1} \\ & =& x_1 \circ (x_2^{-1}
\circ x_2) \circ x_2^{-1} \\ & =& x_1 \circ (y_1 \circ y_2^{-1}) \circ x_2^{-1} \\ & =& (x_1 \circ
y_1) \circ (y_2^{-1} \circ x_2^{-1}) \\ & =& (x_1 \circ y_1) \circ (x_2 \circ y_2)^{-1} \in
\leftop(\sigma \circ \tau)
\end{eqnarray*}
\end{proof}

\begin{lemma}[(7.5) for classes]  $(\sigma \circ \tau) \circ \gamma$ = $\sigma \circ (\tau \circ \gamma)$.
\end{lemma}

\begin{proof} Property (7.4) implies that $(\sigma \circ \tau)\circ \gamma$ is defined if and only if $\sigma \circ (\tau \circ \gamma)$ is defined.
The values in $(\sigma \circ \tau) \circ \gamma$ are the values of the form $(x \circ y) \circ z$
for $\incat{x}{\sigma}$, $\incat{y}{\tau}$ and $\incat{z}{\gamma}$. But these are the same as the
members of $\sigma \circ (\tau \circ \gamma)$.
\end{proof}

\begin{lemma}[(7.6) for classes] $\sigma^{-1} \circ \sigma \circ \tau = \tau$ and $\sigma \circ \tau \circ \tau^{-1} = \sigma$.
\end{lemma}

\begin{proof} We will show that if $\sigma \circ \tau$ is defined then $\sigma^{-1} \circ \sigma \circ \tau = \tau$.
We will let $x$ range over elements of $\sigma$ and $y$ range over elements of $\tau$.  We first
show that every value $y$ in $\tau$ is in $\sigma^{-1} \circ \sigma \circ \tau$.  For this we note
$$ y \; = \; (y \circ y^\inv) \circ y \; = \; (x_1^\inv \circ x_2) \circ y \;\in \; \sigma^\inv
\circ \sigma \circ \tau$$ Conversely, consider $x_1^\inv \circ x_2 \circ y \in \sigma^\inv \circ
\sigma \circ \tau$.  For this case we have $(x_1^{-1} \circ x_2) \circ y_1 = (y_2 \circ y_3^{-1})
\circ y_1 \in \tau$.
\end{proof}

\begin{lemma}[(7.7) for classes]  $\rightop(\sigma) = \sigma^\inv \circ \sigma$ and $\leftop(\sigma) = \sigma \circ \sigma^\inv$
\end{lemma}

\begin{proof}
We will show that $\leftop(\sigma) = \sigma \circ \sigma^\inv$.  Property (7.3) implies that $\sigma
\circ \sigma^\inv$ is defined.  The result is then immediate from the definitions of
$\leftop(\sigma)$ and $\sigma \circ \sigma^\inv$.
\end{proof}

{\bf Properties (7.8) and (7.9)} follow from the duality of left and right.

\begin{lemma}[Class Partner Lemma]  For classes $\sigma$ and $\tau$ such that $\sigma \circ \tau$ is defined we have that for any $x \in \sigma$
  there exists $y \in \tau$ such that $x \circ y$ is defined and for any $y \in \tau$ there exists
  $x \in \sigma$ such that $x \circ \sigma$ is defined. \label{lem:classpartner}
\end{lemma}

\begin{proof}
  Consider classes $\sigma$ and $\tau$ with $\sigma \circ \tau$ defined and consider $x \in \sigma$.
  We have $x^\inv \circ x \in \rightop(\sigma)$.  Since $\rightop(\sigma) = \leftop(\tau)$ we have
  $x^\inv \circ x = y_1 \circ y_2^\inv$ for some $y_1,y_2 \in \tau$.  Since $y_1 \circ y_2^\inv
  \circ y_2$ is defined we have that $x^\inv \circ x \circ y_2$ is defined which implies that $x
  \circ y_2$ is defined.  The reverse partner relationship is similar.
\end{proof}

\section{The Evaluations Properties}
\label{sec:EvalProps}

The evaluation properties in figure~\ref{fig:EvalProps} are proved by simultaneous induction on the
size of the expressions involved.  In this simultaneous induction proof we assume all properties for
smaller expressions while proving any given property on any given expression.

{\bf Proof of (11.1) and (11.2) for defined contexts.} For convenience we repeat the conditions
here.

\medskip
For $\convalue{\Gamma}$ defined we have
\begin{itemize}
\item[(11.1)] For $\rho \in \convalue{\Gamma}$ we have that $\rho$ is a structure (all variables are
  mapped to values).
  
  \medskip
\item[(11.2)] For $\rho \in \convalue{\Gamma}$ we have $\rho^\inv \in \convalue{\Gamma}$ and for
  $\rho_1,\rho_2 \in \convalue{\Gamma}$ with $\rho_1 \circ \rho_2$ defined we have $\rho_1 \circ
  \rho_2 \in \convalue{\Gamma}$.
\end{itemize}

Whether $\convalue{\Gamma}$ is defined, and its meaning when it is defined, is specified by clauses
(5) and (6) if figure~\ref{fig:semantics}.  We have that $\convalue{\Gamma}$ is defined if one of
the following two conditions hold.
\begin{itemize}
\item[(a)] $\Gamma = \Delta;\;\intype{x}{\tau}$ where $\convalue{\Delta}$ is defined and $\Delta
  \models \intype{\tau}{\class}$.
  
\item[(b)] $\Gamma = \Delta;\Phi$ where $\convalue{\Delta}$ is defined and $\Delta \models
  \intype{\Phi}{\bool}$.
\end{itemize}
By the induction hypotheses we have that (11.1) and (11.2) hold for $\Delta$.  Since all members of
a class are values, (11.1) for $\Gamma$ follows immediately from (11.1) for $\Delta$.  For (11.2) we
consider the case of composition and consider each of cases (a) and (b) above.  For case (a) we must
consider $\rho_1,\rho_2 \in \convalue{\Delta}$ with $\rho_1 \circ \rho_2$ defined.  Let $\tau_1^*$
and $\tau_2^*$ abbreviate $\subvalue{\Delta}{\tau}\rho_1$ and $\subvalue{\Delta}{\tau}\rho_2$
respectively.  By the definition of $\Delta \models \intype{\tau}{\class}$ we have that
$\subvalue{\Delta}{\tau}$ is defined. By the induction hypothesis for (11.4) we have that
$\subvalue{\Delta}{\tau}(\rho_1 \circ \rho_2) = \tau_1^* \circ \tau_2^*$.  Now consider $v_1 \in
\tau_1^*$ and $v_2 \in \tau_2^*$ with $v_1 \circ v_2$ defined.  This is the general case in which
$\rho_1[x \leftarrow v_1] \circ \rho_2[x \leftarrow v_2]$ is defined.  Note that
$$\rho_1[x \leftarrow v_1] \circ \rho_2[x \leftarrow v_2] = (\rho_1 \circ \rho_2)[x \leftarrow v_1
  \circ v_2].$$ We must show that
$$(\rho_1 \circ \rho_2)[x \leftarrow (v_1 \circ v_2)] \in \convalue{\Delta;\;\intype{x}{\tau}}.$$
But this now follows from (11.2) for $\Delta$ and (11.4) for $\tau$ which implies that $\subvalue{\Delta}{\tau}(\rho_1 \circ \rho_2) = \tau_1^* \circ \tau_t^*$
and hence $v_1 \circ v_2 \in \subvalue{\Delta}{\tau}(\rho_1 \circ \rho_2)$.

Now we consider composition for case (b) above.  Again consider $\rho_1,\rho_2 \in \convalue{\Delta}$.
To show the composition case of (11.2) for $\Delta;\;\Phi$ we must show that if that $\subvalue{\Delta}{\Phi}\rho_1 = \true$
and that $\subvalue{\Delta}{\Phi}\rho_2 = \true$
then $\subvalue{\Delta}{\Phi}(\rho_1 \circ \rho_2) = \true$.  But $\Delta \models \intype{\Phi}{\bool}$ implies that $\subvalue{\Delta}{\Phi}$ is defined and the result follows from
(11.4) for $\Phi$.

{\bf Proof of (11.3) through (11.6) for defined expressions.}  For convenience we repeat the
properties here.

\medskip
For $\semvalue{e}$ defined with $e \not = \class$ we have
\begin{itemize}
\item[(11.3)] For $\rho \in \convalue{\Gamma}$ we have that $\semvalue{e}\rho$ is a denotable value (eigher a value, a proper class
  or a pair of denotable values).
\item[(11.4)] For $\rho_1,\rho_2 \in \convalue{\Gamma}$ with $\rho_1 \circ \rho_2$ defined we have
  $\semvalue{e}(\rho_1 \circ \rho_2) = (\semvalue{e}\rho_1) \circ (\semvalue{w}\rho_2)$.
\item[(11.5)] For $\rho_1,\rho_2 \in \convalue{\Gamma}$ with $\rho_1 \preceq \rho_2$ we have
  $\semvalue{e}\rho_1 \preceq \semvalue{e}\rho_2$.
\item[(11.6)] For $\rho \in \convalue{\Gamma}$ with $\intype{\rho}{\eta}$ we have
  $\intype{\semvalue{e}\rho}{\tempvalue{e}{\eta}}$.
\end{itemize}
The first part of (11.4) follows from the duality of left and right and we only consider the
composition part.

We must prove these properties for each of the following kinds of defined expressions from
figure~\ref{fig:constructs}.
$$\begin{array}{cccccc} \Sigma_{\intype{x\;}{\;\sigma}}\;\tau[x] & S_{\intype{x\;}{\;\sigma}}\;\Phi[x] & \Pi_{\intype{x\;}{\;\sigma}}\;\tau[x] & \lambda \intype{x}{\sigma}\;e[x] & \forall \intype{x}{\sigma}\; \Phi[x] \\
  \\
  f(e) & e_1 =_\sigma e_2 &  \neg \Phi & \Phi_1 \vee \Phi_2 & (e_1,e_2) \\
  \\
  \pi_i(e) & x & \bool & \sett
\end{array}$$

$\bullet$ $x$, $\bool$, $\sett$.  For variables we have that (11.3) follow from (11.1) for $\Gamma$.  Properties (11.4)
through (11.6) are immediate for variables.  (11.3) through (11.6) are immediate for the constant
$\bool$.  For the constant $\sett$ we first show (11.3) by showing that the collection of all
set values is a proper class.  The morphoid closure condition for sets follows from
properties (7.1) and (7.2) which implies that the sets are closed under both inverse and
composition.  We can also show that the template $\setof(\pointt)$ is an interface template for
$\sett$.  For this we must show that for any set $\sigma$ we have that $\sigma@\setof(\pointt)$ is
defined and is also a set.  But this follows from property (8.1) which states that values are closed
under abstraction.  For (11.4) we must show that $\sett = \sett \circ \sett$.  But it is
straightforward to show that a composition of bijective sets is a set giving
$\sett \circ \sett \subseteq \sett$.  Conversely, for any set $\sigma$ we have $\sigma \circ \sigma^\inv$ is a set giving
$\sett \subseteq \sett \circ \sett$.  Property (11.5) follows from property (8.8) that $\preceq$ is a partial order and hence
$\sett \preceq \sett$. (11.6) follows
from the previously noted fact that $\setof(\pointt)$ is an interface template for $\sett$.
Finally we note that the collection of all singleton sets is too large to be in $U$ and hence $\sett \not \in U$.

$\bullet$ $f(e)$. For application expressions clause (10) of figure~\ref{fig:semantics} requires that $\semvalue{f}$ is defined
and for $\rho \in \convalue{\Gamma}$ we have that $\semvalue{f}\rho$ is a function value (a set of pairs).
This implies that $\semvalue{f(e)}\rho$ is a value and we have (11.3).
Property (11.6) follows from the induction hypothesis for $f$.  For (11.4) consider
$\rho_1,\rho_2 \in \convalue{\Gamma}$ with $\rho_1 \circ \rho_2$ defined.  Let $f^*_1$ abbreviate
$\semvalue{f}\rho_1$ and similarly for $u_1^*$ and let $f_2^*$ and $u_2^*$ be defined similarly in
terms of $\rho_2$.  By the (11.4) for $f$ we have $f_1^* \circ f_2^*$ is defined and by (11.4) for
$e$ we have that $u^*_1 \circ u^*_2$ is defined.  Property (11.4) for $f(e)$ is then equivalent to
$$(f_1^* \circ f_2^*)(u_1^* \circ u_2^*) = f_1^*(u_1^*) \circ f_2^*(u_2^*).$$ But this follows from
the fact that all input-output pairs of $f_1^* \circ f_2^*$ are of the form
$u_1 \circ u_2 \mapsto f_1^*(u_1) \circ f_2^*(u_2)$.

Finally we consider (11.5).  Consider $\rho_1,\rho_2 \in \convalue{\Gamma}$ with
$\rho_1 \preceq \rho_2$.  Let $f^*_1$ and $u_1^*$ be defined in terms of $\rho_1$, and $f_2^*$ and $u_2^*$ be defined
similarly in terms of $\rho_2$.  We must show $f_1^*(u_1^*) \preceq f_2^*(u_2^*)$. We have
$\intype{f_2^*}{{\mathcal T} \rightarrow {\mathcal S}}$ for some templates ${\mathcal T}$ and ${\mathcal S}$.  By (8.7)
it suffices to show that $f_1^*(e_1^*)@{\mathcal S} = f_2^*(e_2^*)$. Since
$\semvalue{f(u)}$ is defined we have that $u_2^*$ must be in the domain of $f_2^*$ which implies
$\intype{u_2}{\mathcal T}$. By (11.5) for $f$ we have $f_1^* \preceq f_2^*$ which implies that
$f_1^*@({\mathcal T} \rightarrow {\mathcal S}) = f_2^*$.  This implies that the pairs of $f_2^*$ are all of
the form $u_1@{\mathcal T} \mapsto f_1^*(u_1)@{\mathcal S}$.  By (11.5) for $e$ we have $u_1^*@{\mathcal T} =
u_2^*$.  We now have $f_2^*(u_2) = f_1^*(u_1)@{\mathcal S}$ which implies the result.

{\bf $\bullet$ $(e_1,e_2)$ and $\pi_i(e)$.}  Properties (11.3) through (11.6) immediately follow
from the corresponding induction hypothesis for pairing and projection expressions.

{\bf Boolean Expressions.}  For Boolean expressions conditions (11.3) and (11.6) are immediate and
we need only consider (11.4) and (11.5).  Furthermore, for Boolean expressions conditions (11.4) and
(11.5) can be simplified to the following.
\begin{itemize}
\item[(11.4b)] For $\rho \in \convalue{\Gamma}$ we have
  $$\semvalue{\Phi}\rho =  \semvalue{\Phi}\leftop(\rho) = \semvalue{\Phi}\rightop(\rho).$$
\item[(11.5b)] For $\rho_1,\rho_2 \in \convalue{\Gamma}$ with $\rho_1 \preceq \rho_2$ we have
  $$\semvalue{\Phi}\rho_1 = \semvalue{\Phi}\rho_2$$
\end{itemize}

{\bf $\bullet$ $\neg \Phi$ and $\Phi \vee \Psi$.} For the Boolean expressions $\neg \Phi$ and $\Phi
\vee \Psi$ properties (11.4b) and (11.5b) follow from the induction hypotheses applied to the
arguments of the operation and the fact that the truth value of the operation is determined by the
truth value of its arguments.

{\bf $\bullet$ $\forall\; \intype{x}{\sigma}\;\Phi[x]$}. For $\forall\; \intype{x}{\sigma}\;\Phi[x]$
we will verify (11.4) (as opposed to (11.4b)).  Consider $\rho_1,\rho_2 \in
\convalue{\Gamma}$ with $\rho_1 \circ \rho_2$ defined.  Let $\tau_1^*$ abbreviate
$\semvalue{\tau}\rho_1$ and let $\Phi_1^*[u_1]$ be $\subxvalue{\Phi[x]}\rho_1[x \leftarrow
  u_1]$.  Let $\tau_2^*$ and $\Phi^*_2[u_2]$ be defined similarly in terms of $\rho_2$.  We must
show that $\Phi_1^*[u_1]$ holds for all $u_1 \in \tau_1^*$ if and only if $\Phi_2^*[u_2]$ holds for
all $u_2 \in \tau_2^*$.  Suppose $\Phi_1^*[u_1]$ holds for all $u_1 \in \tau_1^*$ and consider $u_2
\in \tau_2^*$.  By (11.4) for $\tau$ we have that $\tau_1^* \circ \tau_2^*$ is defined.  By the
class partner lemma~\ref{lem:classpartner} there exists $u_1 \in \tau_1^*$ with $u_1 \circ u_2$
defined.  By (11.4) for $\Phi[x]$ we then have $\Phi_2^*[u_2] = \Phi_1^*[u_1] = \true$ which proves
the result.  The converse is similar.

Next we consider (11.5b).  Consider $\rho_1,\rho_2 \in \convalue{\Gamma}$ with
$\rho_1 \preceq \rho_2$.
Let $\tau_1^*$ and $\Phi^*_1[u_1]$ be defined in terms of $\rho_1$ as before and let
$\tau^*_2$ and $\Phi^*_2[v]$ be defined similarly in terms of $\rho_2$.  As before we must show that
$\Phi^*_1[u_1]$ holds for all $u_1 \in \tau^*_1$ if and only if $\Phi^*_2[u_2]$ holds for all
$u_2 \in \tau^*_2$.  Suppose $\Phi_1^*[u_1]$ holds for all $u_1 \in \tau_1^*$ and consider $u_2 \in \tau_2^*$.
By (11.5) for $\tau$ we have $\tau^*_1 \preceq \tau^*_2$ which implies that for
$\intype{\tau_2^*}{\classof({\mathcal T})}$ we have $\tau_1^*@\classof({\mathcal T}) =
\tau_2^*@\classof({\mathcal T})$ .  This implies that there exists $u_1 \in \tau_1^*$ such that
$u_1@{\mathcal T} = u_2@{\mathcal T}$.  This implies that $u_1@{\mathcal T} \in \tau_2^*$ and by (11.5) for
$\Phi[x]$ we have $\Phi_1^*[u_1] = \Phi_2^*[u_1@{\mathcal T}]$.  We also have that both $u_2$ and
$u_2@{\mathcal T}$ are in $\tau_2^*$ and again by (11.5) for $\Phi[x]$ we have $\Phi^*_2[u_2] =
\Phi^*_2[u_2@{\mathcal T}]$.  Together this gives $\Phi_2^*[u_2] = \true$ as desired.  For the converse
suppose that $\Phi^*_2[u_2]$ is true for all $u_2 \in \tau^*_2$ and consider $u_1 \in \tau_1^*$.
Again noting that $\tau_1^*@\classof({\mathcal T}) = \tau_2^*@\classof({\mathcal T})$ we have that $u_1@{\mathcal
  T} \in \tau_2^*$ and by (11.5) for $\Phi[x]$ we have $\Phi_1^*[u_1] = \Phi_2^*[u_1@{\mathcal T}] = \true$.

$\bullet \;w =_\tau u$.  We will first show (11.4b).  Consider $\rho \in \convalue{\Gamma}$ and let $w^*$
abbreviate $\semvalue{w}\rho$ and similarly for $\tau^*$ and $u^*$.  By (11.4b) for $w$ we have that
$\semvalue{w}\leftop(\rho) = \leftop(w^*)$ and similarly for $\tau^*$ and $u^*$ We will show that
$\leftop(w^*) =_{\leftop(\tau^*)}\;\leftop(u^*)$ if and only if $w^* =_{\tau^*}\;u^*$.  First
suppose $w^* =_{\tau^*} u^*$.  In this case there exists $z \in \tau^*$ with $(w^*@\tau^*) \circ
z^\inv \circ (u^*@\tau^*)$ defined.  This gives that
$$(w^*@\tau^*) \circ (w^*@\tau^*)^\inv \circ (z \circ z^\inv)^\inv \circ (u^*@\tau^*) \circ
(u^*@\tau^*)^\inv$$ is defined. Since $w^* \sim (w^*)^\inv$ and similarly for $u^*$ we have that
$$((w^* \circ (w^*)^\inv)@\tau^*) \circ (z \circ z^\inv)^\inv \circ ((u^* \circ (u^*)\inv)@\tau^*)$$
is defined.  By lemma~\ref{lem:InterfaceTransfer1} we have that any interface template for $\tau$ is
also an interface template for $\leftop(\tau)$. So we now have that
$$\leftop(w^*)@\leftop(\tau^*) \circ (z \circ z^\inv)^\inv \circ \leftop(u^*)@\leftop(\tau^*)$$ is
defined which established $\leftop(w^*) =_{\leftop(\tau^*)} \leftop(u^*)$.

For the converse suppose that $\leftop(w^*) =_{\leftop(\tau^*)} \leftop(u^*)$.  In this case there
exists $z_1,z_2 \in \tau^*$ such that
$$\leftop(w^*)@\leftop(\tau^*) \circ (z_1 \circ z_2^\inv)^\inv \circ \leftop(u^*)@\leftop(\tau^*)$$
is defined. By lemma~\ref{lem:InterfaceTransfer1} we can select the interface template for
$\leftop(\tau^*)$ to be an interface template for $\tau^*$ and we then have that
$$\leftop(w^*)@\tau^* \circ z_2 \circ z_1^\inv \circ \leftop(u^*)@\tau^*$$ is defined.  We then have
that
$$w^*@\tau^* \circ (z_1 \circ z_2^\inv \circ w^*@\tau^*)^\inv \circ u^*@\tau^*$$ is defined which
establishes $w^* =_{\tau^*} u^*$.

We must also show (11.5b).  Consider $\rho_1,\rho_2 \in \convalue{\Gamma}$ with $\rho_1 \preceq
\rho_2$.  Let $w_1^*$ abbreviate $\semvalue{w}\rho_1$ and similarly for $u_1^*$ and $\tau_1^*$.  Let
$w_2^*$, $u_2^*$ and $\tau_2^*$ be similarly defined in terms of $\rho_2$.  We must show that $w_1^*
=_{\tau_1^*} u_1^*$ if and only if $w_2^* =_{\tau_2^*} u_2^*$.  Let ${\mathcal T}_1$ be an interface
template for $\tau_1^*$ and let ${\mathcal T}_2$ be an interface template for $\tau_2^*$.  By (11.5b)
for $\tau$, $u$ and $w$ and the definition of $\preceq$ we have
$$\tau_1^*@\classof({\mathcal T}_2) = \tau_2^*@\classof({\mathcal T}_2)$$
$$u_1^*@{\mathcal T}_2 = u_2^*@{\mathcal T}_2$$
$$w_1^*@{\mathcal T}_2 = w_2^*@{\mathcal T}_2.$$ Suppose $w_1^* =_{\tau_1^*} u_1^*$.  In this case there
exists $z_1 \in \tau_1^*$ such that $(w_1^*@{\mathcal T}_1) \circ z_1^\inv \circ (u_1^*@{\mathcal T}_1)$ is
defined.  Abstracting this to ${\mathcal T}_2$ gives that $(w_1^*@{\mathcal T}_2) \circ (z_1@{\mathcal
  T}_2)^\inv \circ (u_1^*@{\mathcal T}_2)$ is defined which now gives that $(w_2^*@{\mathcal T}_2) \circ
(z_1@{\mathcal T}_2)^\inv \circ (u_2^*@{\mathcal T}_2)$ is defined which give $w_2^* =_{\tau_2^*} u_2^*$ as
desired.

Conversely suppose that $w_2^* =_{\tau_2^*} u_2^*$.  In this case there exists $z_2 \in \tau_2^*$
with $(w_2^*@{\mathcal T}_2) \circ z_2^\inv \circ (u_2^*@{\mathcal T}_2)$ defined.  By propert (8.4) we have
that $\intype{z_2}{{\mathcal T}_2}$ which implies $z_2 = z_2@{\mathcal T}_2$.  By (11.5b) for $\tau$ we then
have $z_2 = z_1@{\mathcal T}_2$ for some $z_1 \in \tau_1^*$ and we then have that
$(w_1^*@{\mathcal T}_2) \circ (z_1@{\mathcal T}_2)^\inv \circ (u_1^*@{\mathcal T}_2)$ is defined.  By property (8.12) we
then have that $(w_1^*@{\mathcal T}_1) \circ (z_1@{\mathcal T}_1)^\inv \circ (u_1^*@{\mathcal T}_1)$ is defined
which gives $w_1^* =_{\tau_1^*}u_1^*$ as desired.

$\bullet\;S\deppair$.  We will first show (11.6). Consider $\rho \in \convalue{\Gamma}$ with
$\intype{\rho}{\eta}$.  We must show $\intype{(\semvalue{S\deppair}\rho)}{\tempvalue{\sigma}\eta}$.
Let $\sigma^*$ abbreviate $\semvalue{\sigma}\rho$ and for $u \in \sigma^*$ let $\Phi^*[u]$
be $\subxvalue{\Phi[x]}\rho[x \leftarrow u]$.  Let ${\mathcal T}$ be the template such that
$\tempvalue{\sigma}\eta$ either has the form $\setof({\mathcal T})$ or the form $\classof({\mathcal T})$.
By (11.6) for $\sigma$ we have that ${\mathcal T}$ is an interface template for $\sigma^*$.  We must
show that ${\mathcal T}$ is also an interface template for $\semvalue{S\deppair}\rho$.  Consider $u \in \semvalue{S\deppair}\rho$.
We must show $u@{\mathcal T} \in \semvalue{S\deppair}\rho$.
By (11.5) for $u$ we have $u \preceq u@{\mathcal T}$ and then by (11.5b) for $\Phi[x]$
we have $\Phi_1^*[u] = \Phi^*_2[u@{\mathcal T}] = \true$ as desired.

Next we show (11.4). Consider $\rho_1,\rho_2 \in \convalue{\Gamma}$ with $\rho_1 \circ \rho_2$
defined.  Let $\tau_1^*$ and $\Phi_1^*[u_1]$ be defined in terms of $\rho_1$ as usual and let
$\tau_2^*$ and $\Phi_2^*[u_2]$ be similarly defined in terms of $\rho_2$.  We must show
$$\semvalue{S\deppair}(\rho_1 \circ \rho_2) = (\semvalue{S\deppair}\rho_1) \circ
(\semvalue{S\deppair}\rho_2)$$ By (11.4) for $\Phi[x]$ we have that for $u_1 \in \tau_1^*$ and $u_2
\in \tau_2^*$ with $u_1 \circ u_2$ defined we have
$$\subxvalue{\Phi[x]}(\rho_1 \circ \rho_2)[x \leftarrow (u_1 \circ u_2)] = \Phi_1^*[u_1] =
\Phi_2^*[u_2]$$ which implies the result.

Next we show (11.5b). Consider $\rho_1,\rho_2 \in \convalue{\Gamma}$ with $\rho_1 \preceq \rho_2$
and let $\tau_1^*$, $\Phi_1^*[u_1]$, $\tau_2^*$ and $\Phi_2^*[u_2]$ be defined as usual in terms of
$\rho_1$ and $\rho_2$.  Let ${\mathcal T}$ be an interface template for $\tau_2^*$.  By (11.5b) for $\tau$
we have $\tau_1^*@\classof({\mathcal T}) = \tau_2^*@\classof({\mathcal T})$.  By (11.5b) for $\Phi[x]$,
for any $u_1 \in \tau_1$ we have $\Phi_1^*[u_1]$ if
and only if $\Phi_2^*[u_1@{\mathcal T}]$.  This implies
\begin{eqnarray*}
  & & (\semvalue{S\deppair}\rho_1)@\classof({\mathcal T}) \\
  & = & (\semvalue{S\deppair}\rho_2)@\classof({\mathcal  T})
\end{eqnarray*}
as desired.

$\bullet$ $\Sigma\deppair$.  We will show (11.3) and (11.6) together.  For (11.3) we must show that the pair type denotes either a set value or a proper class.
We first show that in any case it denotes a class.  We will then show that if the class is in $U$ then it is a set value.

To show that the value is a class consider $\rho \in \convalue {\Gamma}$. Let $\sigma^*$ be $\semvalue{\sigma}\rho$ and for
$u \in \sigma^*$ let $\tau^*[u]$ be $\subvalue{\Gamma\;\intype{x\;}{\;\sigma}}{\tau[x]}\rho[x \leftarrow u]$.
We have that $\semvalue{\sigma_{\intype{x\;}{\;\sigma}}\;\tau[x]}\rho$ is the
collection of pairs $\{(u,v),\;u \in \sigma^*,\;v \in \tau^*[u]\}$.  For any pair $(u,v)$ in this
set we have that property (11.3) for $\sigma$ and $\tau[x]$ imply that $u$ and $v$ are values and hence the pair
$(u,v)$ is a value.  But we must show that this collection of pairs is morphoid closed and has an
interface template.  To show morphoid closure consider $u_1,u_2,u_3 \in \sigma^*$ with $u_1 \circ u_2^\inv \circ u_3$
defined and consider $w_1 \in \tau^*[u_1]$, $w_2 \in \tau^*[u_2]$ and $w_3 \in \tau^*[u_3]$
with $w_1 \circ w_2^\inv \circ w_3$ defined.  We must show that
$$(u_1\circ u_2^\inv \circ u_3,\;w_1 \circ w_2^\inv \circ w_3)$$ is in the pair type.  By morphoid
closure of $\sigma^*$ we have $u_1 \circ u_2^\inv \circ u_3 \in \sigma^*$.  We must show that
$w_1 \circ w_2^\inv \circ w_3 \in \tau^*[u_1 \circ u_2^\inv \circ u_3]$.
By (11.4) for $\tau[x]$ we have
$\tau^*[u_1 \circ u_2^\inv \circ u_3] = \tau^*[u_1] \circ \tau^*[u_2]^\inv \circ \tau^*[u_3]$.
Since $w_1 \circ w_2^\inv \circ w_3 \in \tau^*[u_1] \circ \tau^*[u_2]^\inv \circ \tau^*[u_3]$ this
proves the result.

To prove the existence of an interface template let ${\mathcal T}$ abbreviate
${\mathcal M}(\tempvalue{\sigma}\eta)$ and let ${\mathcal S}$ abbreviate
${\mathcal M}(\subtempvalue{\tau[x]}\eta[x \leftarrow {\mathcal T}])$.  We will show that ${\mathcal T} \times {\mathcal S}$
is an interface template for the pair type.  By (11.6) for $\sigma$ we have that
${\mathcal T}$ is an interface template for $\sigma^*$ and by (11.6) for $\tau[x]$ we have that ${\mathcal S}$
is an interface template for $\tau^*[u@{\mathcal T}]$ (independent of the choice of $u$).  For
$(u,v)$ in the pair type we must show that $(u@{\mathcal T},v@{\mathcal S})$ is in the pair type.  By (11.5)
for $\tau[x]$ we have $\tau^*[u] \preceq \tau^*[u@{\mathcal T}]$.  By the definition of $\preceq$ on
classes we have
$$\{w@{\mathcal S},\;w \in \tau^*[u]\} = \{s@{\mathcal S},\;s \in \tau^*[u@{\mathcal T}]\} \subseteq \tau^*[u@{\mathcal T}].$$
This implies $v@{\mathcal S} \in \tau^*[u@{\mathcal T}]$ which proves the result.

We note that (11.6) for the pair type now follows from the fact that for ${\mathcal T}$ and ${\mathcal S}$
as defined above and for $\rho \in \convalue{\Gamma}$ and $\intype{\rho}{\eta}$ we have
$$\;\tilde{\mathcal V}_{\Gamma;\intype{x\;}{\;\sigma}}\double{\Sigma_{\intype{x\;}{\;\sigma}}\;\tau[x]}\eta = \classof({\mathcal T} \times {\mathcal S})$$

To complete the proof of (11.3) for the pair type we must show that if the pair type is not a proper class (if it is an element of $U$) then it is a set value ---
is bijective and has a template.
Again consider $\rho \in \convalue{\Gamma}$ and let $\sigma^*$ and $\tau^*[u]$ be defined as before.  If $\sigma^*$ a proper class or if $\tau^*[u]$ is a proper class for
any $u \in \sigma^*$ then the pair type is a proper class. So if the pair type is not a proper class then by (11.3) for $\sigma$ we have that $\sigma^*$ is a set value and by (11.3) for $\tau[x]$
we have that $\tau^*[u]$ is a set value for all $u \in \sigma^*$.
Since $\sigma^*$ is a set value we have $\intype{\sigma^*}{\setof({\mathcal T})}$ for
some template ${\mathcal T}$. By (11.1) for $\rho$ have $\intype{\rho}{\eta}$ for some structure template $\eta$.
Let ${\mathcal S}$ be the template $\mem(\tilde{{\mathcal V}}_{\Gamma;\;\intype{x\;}{\;\sigma}}{\tau[x]}\eta[x \leftarrow {\mathcal T}])$.  By (11.6) for
$\tau[x]$ we have $\intype{\tau^*[u]}{\setof({\mathcal S})}$ for any $u \in \sigma^*$.  This implies
that the pair type has template ${\mathcal T} \times {\mathcal S}$.  Finally we must show that the pair type
is bijective.  We will show that if two pairs $(u,v)$ and $(u',v')$ have the same left value then
they are the same.  Assume $\leftop(u) = \leftop(u')$ and $\leftop(v) = \leftop(v')$.
The bijectivety of $\sigma^*$ implies that $u = u'$.  We then have that both $v$ and $v'$ are in
$\tau^*[u]$ and the bijectivity of $\tau^*[u]$ implies that $v = v'$.

We now consider (11.4) for the pair type. Consider $\rho_1,\rho_2 \in \convalue{\Gamma}$.  We must
show
\begin{eqnarray}
  & & \semvalue{\Sigma_{\intype{x\;}{\;\sigma}}\;\tau[x]}(\rho \circ \rho_2) \nonumber \\
  & = & \left(\semvalue{\Sigma_{\intype{x\;}{\;\sigma}}\;\tau[x]}\rho_1\right) \circ  \left(\semvalue{\Sigma_{\intype{x\;}{\;\sigma}}\;\tau[x]}\rho_2\right)
  \label{eq:pairtype}
\end{eqnarray}
Let $\sigma_1^*$ abbreviate $\semvalue{\sigma}\rho_1$ and for $u \in \sigma_1^*$ let $\tau_1^*[u]$
abbreviate $\subvalue{\Gamma;\intype{x\;}{\;\sigma}}{\tau(x)}\rho_1[x \leftarrow u]$.  Define
$\sigma_2^*$ and $\tau_2^*[u]$ similarly in terms of $\rho_2$.

We will first show that the composition on the right hand side of (\ref{eq:pairtype}) is defined ---
that the set of pairs of the form $(\rightop(u_1),\rightop(v_1))$ for $u_1 \in \sigma_1^*$
and $v_1 \in \tau_1^*[u_1]$ is the same as the set of pairs of the form
$(\leftop(u_2),\leftop(v_2))$ for $u_2 \in \sigma_2^*$ and $v_2 \in \tau_2^*[u_2]$.  We will show
that every pair of the first form is also of the second form.  The converse is similar.  Consider
$u_1 \in \sigma^*_1$ and $v_1 \in \tau^*_1[u_1]$.  It now suffices to show that there exists $u_2 \in\sigma^*_2$
and $v_2$ in $\tau^*_2[u_2]$ with $u_1 \circ u_2$ and $v_1 \circ v_2$ defined.  By (11.4) applies to $\sigma$
we have that $\sigma_1^* \circ \sigma_2^*$ is defined. By the class partner lemma~\ref{lem:classpartner} we then have
that there exists $u_2 \in \sigma_2^*$ with $u_1 \circ u_2$ defined.  By (11.4) for $\tau[x]$ we then have that
$\tau_1^*[u_1] \circ \tau_2^*[u_2]$ is defined and again by the class partner lemma there exists $v_2 \in \tau_2^*[u_2]$ such
that $v_1 \circ v_2$ is defined as desired.

To show that (\ref{eq:pairtype}) holds we show containment in both directions.  We first show that
every member of the left hand side is a member of right hand side. Consider $u_i \in \sigma^*_i$ and
$v_i \in \tau_i^*(u_i)$ with $(u_1,v_1) \circ (u_2,v_2)$ defined.  We must show
$$(u_1 \circ u_2,\;v_1\circ v_2) \in \semvalue{\sigma_{\intype{x\;}{\;\sigma}},\;\tau[x]}(\rho_1 \circ \rho_2).$$
By (11.4) for $\sigma$ we have $(u_1 \circ u_2) \in \semvalue{\sigma}(\rho_1 \circ \rho_2)$.  By
(11.4) for $\tau[x]$ we have
$$(v_1 \circ v_2) \in \subvalue{\Gamma;\;\intype{x\;}{\;\sigma}}{\tau[x]}(\rho_1 \circ \rho_2)[x
  \leftarrow (u_1 \circ u_2)]$$ which proves the result.

For the converse consider $(u,v) \in \semvalue{\Sigma_{\intype{x\;}{\;\sigma}}\;\tau[x]}(\rho_1 \circ \rho_2).$
By (11.4) for $\sigma$ we have $u \in \semvalue{\sigma}(\rho_1 \circ \rho_2) = \sigma_1^* \circ \sigma_2^*$
and hence there exist $u_1 \in \sigma_1^*$ and $u_2 \in \sigma_2^*$ such that $u = u_1 \circ u_2$.
By (11.4) for $\tau[x]$ we have $v \in \tau_1^*[u_1] \circ \tau_2^*[u_2]$ which gives
$v = v_1 \circ v_2 \in \tau^*_1[u_1] \circ \tau^*_2[u_2]$ as desired.

Finally we show (11.5) for the pair type.  Consider $\rho_1,\rho_2 \in \convalue{\Gamma}$ with
$\rho_1 \preceq \rho_2$.  We must show
$$\semvalue{\Sigma_{\intype{x\;}{\;\sigma}}\;\tau[x]}\rho_1 \preceq \semvalue{\Sigma_{\intype{x\;}{\;\sigma}}\;\tau[x]}\rho_2.$$
Define $\sigma^*_1$, $\tau_1^*[u]$, $\sigma_2^*$ and $\tau_2^*[u]$ as before.
We have $\intype{\rho_2}{\eta}$ for some structure templates $\eta$.
Let ${\mathcal T}$ denote $\tempvalue{\sigma}\eta$ and let ${\mathcal S}$ denote
$\subtempvalue{\tau[x]}\eta[x \leftarrow {\mathcal T}]$.
By the definition of $\preceq$ on classes (figure~\ref{fig:classes}) we must show
\begin{eqnarray*}
  & & \{(u_1,v_1)@({\mathcal T} \times {\mathcal S}),\;u_1 \in \sigma^*_1, \;v_1 \in \tau^*_1[u_1]\} \\
  & = &  \{(u_2,v_2)@({\mathcal T} \times {\mathcal S}),\;u_2 \in \sigma^*_2, \;v_2 \in \tau^*_2[u_2]\}.
\end{eqnarray*}
We note that in the case that the pair type is a set value this equality also suffices by virtue of (8.7) and the fact that in that
case we have $(u_2,v_2)@({\mathcal T} \times {\mathcal S}) = (u_2,v_2)$.

We will show containment in both directions.  First consider $u_1 \in \sigma_1^*$ and $v_1 \in
\tau^*_1[u_1]$.  We must show that $(u_1,v_1)@({\mathcal T}\times {\mathcal S})$ is defined and is in
$\semvalue{\Sigma_{\intype{x\;}{\;\sigma}}\;\tau[x]}\rho_2$.  By (11.5) for $\sigma$ we have
$\sigma_1^* \preceq \sigma_2^*$ and therefore that $u_1@{\mathcal T} \in \sigma_2^*$.  By (11.5) for
$\tau[x]$ we have $\tau_1[u_1] \preceq \tau_2[u_1@{\mathcal T}_1]$.  By (11.6) for $\tau[x]$ we also
have that that ${\mathcal S}$ is an interface template for $\tau_2[u_1@{\mathcal T}]$.  These two facts
together give that $v_1@{\mathcal S} \in \tau_2^*[u_1@{\mathcal T}_1]$ which gives
$(u_1,v_1)@({\mathcal T} \times {\mathcal S}) \in \semvalue{\Sigma_{\intype{x\;}{\;\sigma}}\;\tau[x]}\rho_2$ as desired.

Finally consider $u_2 \in \sigma^*_2$ and $v_2 \in \tau^*_2[u_2]$.  We must show that there exists
$u_1 \in \sigma_1^*$ and $v_1 \in \tau^*_1[u_1]$ with $u_1@{\mathcal T} = u_2@{\mathcal T}$ and
$v_1@{\mathcal S} = v_2@{\mathcal S}$.  But, as previously noted, we have $\sigma^*_1 \preceq
\sigma_2^*$ which implies that there exists $u_1 \in \sigma_1^*$ such that $u_1@{\mathcal T} =
u_2@{\mathcal T}$.  We also have
$$\tau_1^*[u_1] \preceq \tau_2^*[u_1@{\mathcal T}] = \tau^*_2[u_2@{\mathcal T}]$$
$$v_2 \in \tau_2^*[u_2] \preceq \tau_2^*[u_2@{\mathcal T}]$$
$$v_2@{\mathcal S} \in \tau_2^*[u_2@{\mathcal T}]$$
$$\tau_1[u_1] \preceq \tau_2^*[u_2@{\mathcal T}]$$
The last two conditions imply that there exists
$v_1 \in \tau^*_1[u_1]$ with $v_1@{\mathcal S} = v_2@{\mathcal S}$ as desired.

$\bullet$ $\Pi\deppair$. In the system presented here we only allow set-level dependent function
types.  For $\semvalue{\Pi\deppair}$ to be defined we must have $\Gamma \models
\intype{\sigma}{\sett}$ and $\Gamma;\;\intype{x}{\sigma} \models \intype{\tau[x]}{\sett}$.  To show
(11.3) we show that $\Pi_{\intype{x\;}{\;\sigma}}\;\tau[x]$ is a set value.  Consider $\rho \in
\convalue{\Gamma}$ and let $\eta$ be a structure template such that $\intype{\rho}{\eta}$.  Define
$\sigma^*$ and $\tau^*[u]$ in terms of $\rho$ as usual.  Let ${\mathcal T}$ be
$\mem(\tempvalue{\sigma}\eta)$ and ${\mathcal S}$ be $\mem(\subtempvalue{\tau[x]}\eta[x \leftarrow {\mathcal
    T}])$.  By (11.6) for $\sigma$ we have $\intype{u}{\mathcal T}$ for all $u \in \sigma^*$. By (11.6)
for $\tau[x]$ we then have $\intype{v}{\mathcal S}$ for all $v \in \tau^*[u]$.  This implies that for
every function $f \in \semvalue{\Pi\deppair}\rho$ we have $\intype{f}{{\mathcal T} \rightarrow {\mathcal
    S}}$.  This implies $\intype{\semvalue{\Pi\deppair}\rho}{\setof({\mathcal T} \rightarrow {\mathcal
    S})}$.  We have now established that that $\semvalue{\Pi\deppair}\rho$ is a weak value.
By definition every member of $\semvalue{\Pi_{\intype{x\;}{\;\sigma}}\;\tau[x]}\rho$ is a function with domain
$\sigma$ and hence every member of
$\semvalue{\Pi_{\intype{x\;}{\;\sigma}}\;\tau[x]}\rho$ is a value.

We must also show that $\semvalue{\Pi\deppair}\rho$ is bijective.
Consider function $f$ and $g$ in this type with $\leftop(f) = \leftop(g)$.
We will show that in this case $f=g$.  Consider a pair $(u \mapsto v) \in f$. Note that we have
$v \in \tau^*[u]$.  Since $\leftop(g) = \leftop(f)$ we must have $(\leftop(u) \mapsto \leftop(v)) \in \leftop(g)$.
Because of the bijectivity of $\sigma^*$ this implies that $g$ must contain a pair of
the form $u \mapsto v'$ with $v' \in \tau^*[u]$ and $\leftop(v') = \leftop(v)$.  But by the
bijectivity of $\tau^*[u]$ we then have $v' = v$ and hence $g$ also contains the pair $u \mapsto v$.
But this implies $f=g$.

Property (11.6) for the function type follows from
$\intype{\semvalue{\Pi\deppair}\rho}{\setof({\mathcal T} \rightarrow {\mathcal S})}$ proved above.

We now show (11.4) for the function type. Consider $\rho_1,\rho_2 \in \convalue{\Gamma}$ with
$\rho_1 \circ \rho_2$ defined.  Define $\sigma^*_1$ and $\tau^*_1[u]$ in terms of $\rho_1$ and
$\sigma^*_2$ and $\tau^*_2[u]$ in terms of $\rho_2$ as usual.  We must show
\begin{eqnarray*}
  & & \semvalue{\Pi_{\intype{x\;}{\;\sigma}}\;\tau[x]}(\rho_1 \circ \rho_2) \\
  & = & (\semvalue{\Pi_{\intype{x\;}{\;\sigma}}\;\tau[x]}\rho_1) \circ (\semvalue{\Pi_{\intype{x\;}{\;\sigma}}\;\tau[x]}\rho_2).
\end{eqnarray*}
We will show containment in both
directions. Consider a function $f \in \semvalue{\Pi_{\intype{x\;}{\;\sigma}}\;\tau[x]}(\rho_1 \circ \rho_2)$.
By (11.4) for $\sigma$ we have $\domop(f) = \sigma_1^* \circ \sigma_2^*$.
For $(u \mapsto v) \in f$ we then have $u = u_1 \circ u_2$ with $u_1 \in \sigma_1^*$ and $u_2 \in \sigma_2^*$.  By (11.4) for $\tau[x]$ we
have $v \in \tau^*_1[u_1] \circ \tau^*_2[u_2]$.  This implies that $v = v_1 \circ v_2$ for some $v_1
\in \tau^*_1[u_1]$ and $v_2 \in \tau^*_2[u_2]$.  We then have that every pair $(u \mapsto v) \in f$
can be written as $(u_1 \circ u_2) \mapsto (v_1 \circ v_2)$.  The bijectivity of $\sigma_1^*$ and
$\sigma_2^*$ and of $\tau_1^*[u_1]$ and $\tau_2^*[u_2]$ imply that this decomposition is unique.
The set of pairs $u_1 \mapsto v_1$ arising from this decomposition defines a function $f_1$ and
similarly for $f_2$ and we then have $f = f_1 \circ f_2$ with $f_1 \in \semvalue{\Pi\deppair}\rho_1$
and $f_2 \in \semvalue{\Pi\deppair}\rho_1$ as desired.  Containment in the reverse direction is
similarly straightforward.

We now show (11.5) for the function type.  Consider $\rho_1,\rho_2 \in \convalue{\Gamma}$ with $\rho_1 \preceq \rho_2$.
We must show.
$$\semvalue{\Pi\deppair}\rho_1 \preceq \semvalue{\Pi\deppair}\rho_2.$$ Define $\sigma_1^*$ and
$\tau_1^*[u_1]$ in terms of $\rho_1$ and $\sigma_2^*$ and $\tau^*_2[u_2]$ in terms of $\rho_2$ as
usual.  Consider $\eta$ such that $\intype{\rho_2}{\eta}$. We have that $\rho_1 \preceq \rho_2$ implies
$\rho_1@\eta = \rho_2$.  Let ${\mathcal T}$ denote $\mem(\tempvalue{\sigma}\eta)$ and ${\mathcal S}$ denote
$\mem(\subtempvalue{\tau[x]}\eta[x \leftarrow {\mathcal T}])$.  By (8.7) it now suffices to show that
$$(\semvalue{\Pi\deppair}\rho_1)@\setof({\mathcal T} \rightarrow {\mathcal S}) = \semvalue{\Pi\deppair}\rho_2.$$
We will show containment in both direction.

We first note that by (11.5) for $\sigma$ we have $\sigma_1^* \preceq \sigma_2^*$.  By (11.5) for
$\tau[x]$ we have that for $u_1 \in \sigma_1^*$ and $u_2 \in \sigma_2^*$ with $u_1 \preceq u_2$ we
have $\tau_1^*[u_1] \preceq \tau_2^*[u_2]$.  First consider $f_1 \in \semvalue{\Pi\deppair}\rho_1$.
For $(u_1 \mapsto v_1) \in f_1$ we have $u_1@{\mathcal T} \in \sigma_2^*$.  We also have
$v_1 \in \tau^*[u_1] \preceq \tau_2^*[u_1@{\mathcal T}]$.  This implies
$v_1@{\mathcal S} \in \tau_2^*[u_1@{\mathcal T}]$.  This implies that $f@({\mathcal T} \rightarrow {\mathcal S})$ is defined.  Since the abstraction
of a value is a value we have that $f@({\mathcal T} \rightarrow {\mathcal S})$ is functional and we have
$f@({\mathcal T} \rightarrow {\mathcal S}) \in \semvalue{\Pi\deppair}\rho_2$.

For the reverse direction consider $f_2 \in \semvalue{\Pi\deppair}\rho_2$.  We must show that $f_2 =
f_1@({\mathcal T} \rightarrow {\mathcal S})$ for some $f_1 \in \semvalue{\Pi\deppair}\rho_1$.  To construct
the function $f_1$ we must select some $v_1 \in \tau^*_1[u_1]$ for each $u_1 \in \sigma^*_1$.  For
each such $u_1$ we have $u_1@{\mathcal T} \in \sigma^*_2$ and hence $f_2(u_1@{\mathcal T})$ is defined and
is in $\tau^*_2[u_1@{\mathcal T}]$.  But we have $\tau^*_1[u_1] \preceq \tau^*[u_1@{\mathcal T}]$ which by
lemma~\ref{lem:setdown} implies that there exists a unique $v_1 \in \tau^*_1[u_1]$ with $v_1@{\mathcal
  S} = f_2(u_1@{\mathcal T})$.  So we can take $f_1$ to map $u_1$ to this $v_1$ and we then get that
$f_1@({\mathcal T} \rightarrow {\mathcal S}) = f_2$.

$\bullet$ $\lambda\;\intype{x}{\sigma}\;e[x]$. We first show (11.3) for the lambda expression.
Consider $\rho \in \convalue{\Gamma}$ and let $\sigma^*$ and $e^*[u]$ be defined as usual for
$\rho$.  Let $\eta$ be a structure template such that $\intype{\rho}{\eta}$.  By the semantics of
lambda expressions (clause (9) of figure~\ref{fig:semantics}) we have that $\sigma^*$ is a set.
This implies that we have
$\intype{\sigma^*}{\setof({\mathcal T})}$ for some template ${\mathcal T}$.  Let ${\mathcal S}$ be the template
$\mem(\subtempvalue{e[x]}\eta[x \leftarrow {\mathcal T}])$.  By (11.6) for $e[x]$ we have that for
every $(u \mapsto v) \in \semvalue{\lambda\;\intype{x}{\sigma}\;e[x]}\rho$ we have $\intype{u}{\mathcal
  T}$ and $\intype{v}{\mathcal S}$.  This implies
$\intype{(\semvalue{\lambda\;\intype{x}{\sigma}\;e[x]}\rho)}{({\mathcal T} \rightarrow {\mathcal S})}$
which implies that the function is a value.

Property (11.6) follows from $\intype{(\semvalue{\lambda\;\intype{x}{\sigma}\;e[x]}\rho)}{({\mathcal T}
  \rightarrow {\mathcal S})}$ shown above.

We now prove (11.4).  Consider $\rho_1,\rho_2 \in \convalue{\Gamma}$ with $\rho_1 \circ \rho_2$
defined.  We must show
\begin{eqnarray*}
  & & \semvalue{\lambda\;\intype{x}{\sigma}\;e[x]}(\rho_1 \circ \rho_2) \\
  & = & (\semvalue{\lambda\;\intype{x}{\sigma}\;e[x]}\rho_1) \circ (\semvalue{\lambda\;\intype{x}{\sigma}\;e[x]}\rho_1).
\end{eqnarray*}
Let $\sigma_1^*$ and $e^*_1[u_1]$ be defined in terms of $\rho_1$ and $\sigma^*_2$ and $e^*_2[u_2]$ be defined in terms of $\rho_2$ as usual.  By
(11.4) for $\sigma$ we have $\semvalue{\sigma}(\rho_1 \circ \rho_2) = \sigma_1^* \circ \sigma_2^*$.
By the bijectivity of $\sigma_1^*$ and $\sigma_2^*$ every element $u$ of $\sigma_1^* \circ
\sigma_2^*$ has a unique factoring $u = u_1 \circ u_2$ with $u_1 \in \sigma_1^*$ and $u_2 \in
\sigma_2^*$.  By (11.4) for $e[x]$ we then have that
$\semvalue{\lambda\;\intype{x}{\sigma}\;e[x]}(\rho_1 \circ \rho_2)$ is the set of pairs of the form
$u_1 \circ u_2 \mapsto e^*_1[u_1] \circ e^*[u_2]$.  But this is the same as
$(\semvalue{\lambda\;\intype{x}{\sigma}\;e[x]}\rho_1) \circ (\semvalue{\lambda\;\intype{x}{\sigma}\;e[x]}\rho_1)$.

To show (11.5) consider $\rho_1,\rho_2 \in \convalue{\Gamma}$ with $\rho_1 \preceq \rho_2$.  Let
$\sigma_1^*$ and $e_1^*[u_1]$ be defined in terms of $\rho_1$ and let $\sigma_2^*$ and $e_2^*[u_2]$
be defined in terms of $\rho_2$ as usual.  Let $f_1^*$ abbreviate
$\semvalue{\lambda\;\intype{x}{\sigma}\;e[x]}\rho_1$ and $f_2^*$ be defined similarly in terms of
$\rho_2$.  As shown above we have $\intype{f_2^*}{({\mathcal T} \rightarrow {\mathcal S})}$ for some
templates ${\mathcal T}$ and ${\mathcal S}$. By (8.7) it suffices to show $f_1^*@({\mathcal T} \rightarrow {\mathcal S}) =
f_2^*$.  By (11.5) for $\sigma$ we have $\sigma^*_1@\setof({\mathcal T}) = \sigma_2^*$.  This implies
that for $u_1 \in \sigma_1^*$ we have that $u_1@{\mathcal T} \in \sigma_2^*$ which implies that
$\intype{e_2^*[u_1@{\mathcal T}]}{{\mathcal S}}$.  But we have $u_1 \preceq u_1@{\mathcal T}$ and by (11.5)
for $e[x]$ we have $e_1^*[u_1] \preceq e_2^*[u_1@{\mathcal T}]$.  We then have $e_1^*[u_1]@{\mathcal S}
= e_2^*[u_1@{\mathcal T}]$.  We now have that $f_1^*@({\mathcal T} \rightarrow {\mathcal S})$ equals the
set of pairs of the form $u_1@{\mathcal T} \rightarrow e^*_2[u_1@{\mathcal T}]$ which is the same as the
set of pairs in $f_2^*$.

\section{Soundness for Figure~\ref{fig:IsoRules}}

We first show that $\carrier(\sigma,\gamma,f,(\lambda\;\intype{\alpha}{\sett}\;\tau[\alpha]))$
satisfies evaluation invariants (11.3) through (11.6).  For (11.6) we must expand the definition of
template evaluation with
\begin{eqnarray*}
  & & \tempvalue{\carrier(\sigma,\gamma,f,(\lambda\;\intype{\alpha}{\sett}\;\tau[\alpha]))}\eta \\
  & = & \mem(\tempvalue{\tau[\sigma]}\eta) \rightarrow \mem(\tempvalue{\tau[\gamma]}\eta).
\end{eqnarray*}
To show (11.4)
we consider $\rho_1,\rho_2 \in \convalue{\Gamma}$ with $\rho_1 \circ \rho_2$ defined.  Define
$\sigma_1^*$, $\gamma_1^*$, $f_1^*$ and $\tau_1^*[s]$ as usual in terms of $\rho_1$ and
define $\sigma_2^*$, $\gamma_2^*$, $f_2^*$ and $\tau_2^*[s]$ similarly in terms of $\rho_2$.  We then have
the following.
\begin{eqnarray*}
  & & \semvalue{\carrier(\sigma,\gamma,f,(\lambda\;\intype{\alpha}{\sett}\;\tau[\alpha]))}(\rho_1 \circ \rho_2) \\
  & = & C(\tau_1^*[\sigma_1^*] \circ  \tau_2^*[\sigma_2^*],\;W^\inv,\;\tau_1^*[\gamma_1^*] \circ \tau_2^*[\gamma_2^*]) \\ \\
  W & = &  \subvalue{\Gamma;\;\intype{\alpha\;}{\;\sett}}{\tau[\alpha]}(\rho_1 \circ \rho_2)[\alpha \leftarrow Y(f_1^* \circ f_2^*)]
\end{eqnarray*}

We now observe the following where $u_1$ ranges over elements of $\sigma_1^*$ and $u_2$ ranges over
elements of $\sigma_2^*$.
\begin{eqnarray*}
  & & Y(f_1^* \circ f_2^*) \\
  & = & \left\{\begin{array}{l} \pointt\left(\begin{array}{l} \mathrm{Lindex}((f_1^*\circ f_2^*)(u_1 \circ  u_2)@\pointt), \\
                                                              \mathrm{Rindex}((u_1 \circ u_2)@\pointt)\end{array}\right), \\
                                u_1\circ u_2\;\mbox{defined}\end{array}\right\} \\
  & = & \left\{\begin{array}{l} \pointt(\mathrm{Lindex}(f_1(u_1)@\pointt),\;\mathrm{Rindex}(u_2@\pointt)\;), \\
                                        u_1\circ   u_2\;\mbox{defined} \end{array}\right\} \\
  & = & Y(f_1)\circ \sigma_2^*@\setof(\pointt)
\end{eqnarray*}
We then have
\begin{eqnarray*}
  & & \semvalue{\carrier(\sigma,\gamma,f,(\lambda\;\intype{\alpha}{\sett}\;\tau[\alpha]))}(\rho_1  \circ \rho_2) \\ \\
  & = & C\left(\begin{array}{l} \tau_1^*[\sigma_1^*] \circ  \tau_2^*[\sigma_2^*], \\
                                \tau_2^*[\sigma_2^*@\setof(\pointt)]^\inv \circ \tau_1^*[Y(f_1^*)]^\inv, \\
                                \tau_1^*[\gamma_1^*] \circ \tau_2^*[\gamma_2^*]\end{array}\right)
\end{eqnarray*}
Let ${\mathcal T}$ be a template such that $\intype{\tau_1^*[Y(f_1^*)]}{\mathcal T}$.  We now have that
$$C(\tau_1^*[\sigma_1^*] \circ \tau_2^*[\sigma_2^*],\;W^\inv,\;\tau_1^*[\gamma_1^*] \circ
\tau_2^*[\gamma_2^*])$$ is the set of pairs $u_1 \circ u_2 \mapsto v_1 \circ v_2$ such that
$$(u_1 \circ u_2)@{\mathcal T} \circ u_2@{\mathcal T}^\inv \circ w^\inv \circ (v_1 \circ v_2)@{\mathcal T}$$ is
defined with $w \in \tau_1^*[Y(f_1)^*]$.  This is the same as the set of pairs $u_1 \circ u_2
\mapsto v_1 \circ v_2$ such that $u_1 \mapsto v_1$ is a pair of
$$C(\tau_1^*[\sigma_1^*],\tau_1^*[Y(f_1^*)]^\inv, \tau_1[\gamma_1^*]).$$
But it is also possible to show that $Y(f_1^* \circ
f_2^*) = \gamma_1^*@\setof(\pointt) \circ Y(f_2)$.  By a similar argument we then have that
$$C(\tau_1^*[\sigma_1^*] \circ \tau_2^*[\sigma_2^*],\;W^\inv,\;\tau_1^*[\gamma_1^*] \circ
\tau_2^*[\gamma_2^*])$$ is the set of pairs $u_1 \circ u_2 \mapsto v_1 \circ v_2$ such that $u_2
\mapsto v_2$ is a pair of $C(\tau_2^*[\sigma_2^*], \tau_2^*[Y(f_2^*)]^\inv, \tau_2[\gamma_2^*])$.
So we now have
\begin{eqnarray*}
  & & C(\tau_1^*[\sigma_1^*] \circ \tau_2^*[\sigma_2^*],\;W^\inv,\;\tau_1^*[\gamma_1^*] \circ \tau_2^*[\gamma_2^*]) \\
  & = & C(\tau_1^*[\sigma_1^*], \tau_1^*[Y(f_1^*)]^\inv, \tau_1[\gamma_1^*]) \\
  & & \circ C(\tau_2^*[\sigma_2^*], \tau_2^*[Y(f_2^*)]^\inv, \tau_2[\gamma_2^*])
\end{eqnarray*}
as desired.

Next we consider (11.5). Consider $\rho_1,\rho_2 \in \convalue{\Gamma}$ with $\rho_1 \preceq
\rho_2$.  Define $\sigma_1^*$, $\gamma_1^*$, $f_1^*$ and $\tau_1^*[s]$ as usual in terms of $\rho_1$
and $\sigma_2^*$, $\gamma_2^*$, $f_2^*$ and $\tau_2^*[s]$ similarly in terms of $\rho_2$.  We first
note that $C(X,Y,Z) = C(X,Y@\setof(\pointt),Z)$.  To see this consider ${\mathcal T}$ with
$\intype{Y}{\setof({\mathcal T})}$.  We have that $X@\setof({\mathcal T}) \circ Y \circ Z@\setof({\mathcal T})$
is defined and for $u \in X@\setof({\mathcal T})$ we have that there exists a unique $w \in Y$ and $v \in Z@\setof({\mathcal T})$
with $u@{\mathcal T} \circ y \circ v@{\mathcal T}$ defined.  But by (8.11) and (8.12) this is defined if
and only if $x@\pointt \circ y@\pointt \circ z@\pointt$ is defined and hence
$C(X,Y@\setof(\pointt),Z)$ is the same function from $X$ to $Z$ as $C(X,Y,Z)$.  Given this
observation it suffices to show that for $X_1 \preceq X_2$ and $Z_1 \preceq Z_2$ and $Y$ a set of
points we have $C(X_1,Y,Z_1) \preceq C(X_2,Y,Z_2)$.

We establish $C(X_1,Y,Z_1) \preceq C(X_2,Y,Z_2)$ for a point set $Y$ using (8.7) by showing that
for $\intype{C(X_2,Y,Z_2)}{({\mathcal T} \rightarrow {\mathcal S})}$ we have
$$C(X_1,Y,Z_1)@({\mathcal T} \rightarrow {\mathcal S}) = C(X_2,Y,Z_2).$$
By lemma~\ref{lem:setdown} it now suffices to show
that for $(u \mapsto v) \in C(X_1,Y,Z_1)$ we have $(u@{\mathcal T} \mapsto v@{\mathcal S}) \in C(X_2,Y,Y_2)$.
But (8.11) and (8.12) imply that for $y \in Y$ we have that
$u@\pointt \circ y \circ v@\pointt$ is defined if and only if $u@{\mathcal T}@\pointt \circ y \circ v@{\mathcal S}@\pointt$ is defined
and the result follows.

Next we consider the rule

{\small
\centerline{\unnamed {\ant{\Gamma \rdash \intype{\sigma,\gamma}{\mathbf{Set}},
    \intype{f}{\mathbf{Bijection}[\sigma,\gamma]}} \ant{\Gamma;\;\intype{\alpha}{\sett} \rdash
    \intype{\tau[\alpha]}{\sett}}} {\ant{\Gamma \rdash \forall
    \;\intype{x}{\tau[\sigma]}\;\;(\sigma,x)
    \;=_{\Sigma_{\intype{\alpha\;}{\;\sett}}\;\tau[\alpha]}\;
    (\gamma,\mathbf{Carrier}(\sigma,\gamma,f,(\lambda\;\intype{\alpha}{\sett}\;\tau[\alpha]))(x))}}
}}

\medskip
Consider $\rho \in \convalue{\Gamma}$.  Define $\sigma^*$, $\gamma^*$, $f^*$ and $\tau^*[s]$ as
usual in terms of $\rho$.  We will also write $g^*$ for
$$\semvalue{\mathbf{Carrier}(\sigma,\gamma,f,(\lambda\;\intype{\alpha}{\sett}\;\tau[\alpha]))}\rho.$$
We then have
$$g^* = C(\tau^*[\sigma^*],\;\tau^*[Y(f^*)]^\inv, \tau^*[\gamma^*]).$$ Now consider $u \in
\tau^*[\sigma^*]$.  We must show
$$(\sigma^*,u)
=_{\semvalue{\Sigma_{\intype{\alpha\;}{\;\sett}}\;\tau[\alpha]}\rho}\;\;(\gamma^*,g^*(u)).$$ Let
$\setof(\pointt) \times {\mathcal T}$ be the interface for
$\semvalue{\Sigma_{\intype{\alpha\;}{\;\sett}}\;\tau[\alpha]}\rho$.  By the definition of $g^*$ we
have that there exists $y \in \tau^*[Y(f^*)]$ with $u@{\mathcal T} \circ y^\inv \circ g^*(u)@{\mathcal T}$
defined.  We also have that $\sigma@\setof(\pointt) \circ Y(f^*)^\inv \circ
\gamma^*@\setof(\pointt)$ is defined.  This implies that
$$\begin{array}{l} (\sigma^*,u)@(\setof(\pointt) \times {\mathcal T}) \circ (Y(f^*),y)^\inv \\
    \circ \;(\gamma^*,g^*(u))@(\setof(\pointt) \times {\mathcal T})
\end{array}$$
is defined.  We also have $(Y(f^*),y) \in
\semvalue{\Sigma_{\intype{\alpha\;}{\;\sett}}\;\tau[\alpha]}\rho$ which proves the result.

Next we consider the rule

~ \hfill
\unnamed {\ant{\Gamma \rdash \intype{\sigma,\gamma}{\mathbf{Set}},
    \intype{f}{\mathbf{Bijection}[\sigma,\gamma]}}} {\ant{\Gamma \rdash
    \mathbf{Carrier}(\sigma,\gamma,f,(\lambda \intype{\alpha}{\sett}\;\alpha)) \doteq f}} \hfill ~

\medskip
Consider $\rho$ in $\convalue{\Gamma}$ and let $\sigma^*$, $\gamma^*$, $f^*$ and $\tau^*[s]$ be
defined as usual for $\rho$.  We have
$$\mathbf{Carrier}(\sigma,\gamma,f,(\lambda \intype{\alpha}{\sett}\;\alpha)) =
C(\sigma^*,\;Y(f^*)^\inv,\;\gamma^*).$$ This is defined to be the function $g$ from $\sigma^*$ to
$\gamma^*$ such that for all $u \in \sigma^*$ there exists $y \in Y(f^*)$ such that $u@\pointt \circ
y^\inv \circ g(u)@\pointt$ is defined.  But the definition of $Y(f^*)$ yields that this property
holds for $f^*$ which proves the result.

Next we consider

~ \hfill
\unnamed {\ant{\Gamma \rdash \intype{\sigma,\gamma}{\mathbf{Set}},
    \intype{f}{\mathbf{Bijection}[\sigma,\gamma]}} \ant{\Gamma \rdash \intype{w}{\sett}}}
       {\ant{\Gamma \rdash \mathbf{Carrier}(\sigma,\gamma,f,(\lambda \intype{\alpha}{\sett}\;w))
           \doteq (\lambda \intype{x}{w}\;x)}} \hfill ~

\medskip
Consider $\rho$ in $\convalue{\Gamma}$ and let $\sigma^*$, $\gamma^*$, $f^*$ and $w^*$ be defined as
usual for $\rho$.  We have
$$\mathbf{Carrier}(\sigma,\gamma,f,(\lambda \intype{\alpha}{\sett}\;w)) =
C(w^*,\;(w^*)^\inv,\;w^*).$$ This is defined to be the function $g$ from $w^*$ to $w^*$ such that
for all $u \in w^*$ there exists $y \in w^*$ such that $u \circ y^\inv \circ g(u)$ is defined.  But
this implies that $g$ is the identity function on $w^*$.

\medskip
Next we consider

~ \hfill
\unnamed {\ant{\Gamma
    \rdash\mathbf{Carrier}(\sigma,\gamma,f,(\lambda\;\intype{\alpha}{\sett}\;\tau[\alpha])) \doteq
    g} \ant{\Gamma
    \rdash\mathbf{Carrier}(\sigma,\gamma,f,(\lambda\;\intype{\alpha}{\sett}\;\kappa[\alpha])) \doteq
    h}} {\ant{\Gamma \rdash
    \begin{array}{l}
      \mathbf{Carrier}(\sigma,\gamma,f,(\lambda\intype{\alpha}{\sett}\;\tau[\alpha] \times
      \kappa[\alpha])) \\ \doteq \lambda\intype{x}{(\tau[\sigma] \times
        \kappa[\sigma])}\;(g(\pi_1(x)),h(\pi_2(x)))
    \end{array}
    }} \hfill ~

\medskip
Consider $\rho$ in $\convalue{\Gamma}$ and let $\sigma^*$, $\gamma^*$, $f^*$, $\tau^*[s]$ and
$\kappa^*[s]$ be defined as usual for $\rho$.  We will also use the following abbreviations.
\begin{eqnarray*}
g^* & = & \semvalue{\mathbf{Carrier}(\sigma,\gamma,f,(\lambda
  \intype{\alpha}{\sett}\;\tau[\alpha]))}\rho \\ h^* & = &
\semvalue{\mathbf{Carrier}(\sigma,\gamma,f,(\lambda \intype{\alpha}{\sett}\;\kappa[\alpha]))}\rho
\end{eqnarray*}
We have
\begin{eqnarray*}
  & & \mathbf{Carrier}(\sigma,\gamma,f,(\lambda \intype{\alpha}{\sett}\;\tau[\alpha] \times \kappa[\alpha])) \\
  & = & C\left(\begin{array}{l} \tau^*[\sigma^*] \times \kappa^*[\sigma^*], \\
    (\tau^*[Y(f^*)] \times \kappa^*[Y(f^*)])^\inv, \\
    \tau^*[\gamma^*] \times \kappa^*[\gamma^*]\end{array}\right).
\end{eqnarray*}
Let ${\mathcal T}_1$ be a template satisfying $\intype{\tau^*[Y(f^*)]}{\setof({\mathcal T}_1)}$ and let
${\mathcal T}_2$ be a template satisfying $\intype{\kappa^*[Y(f^*)]}{\setof({\mathcal T}_2)}$.  We must show
that for $(u_1,u_2) \in \tau^*[\sigma^*] \times \kappa^*[\sigma^*]$ there exists $(y_1,y_2) \in
\tau^*[Y(f^*)] \times \kappa^*[Y(f^*)]$ such that
$$(u_1,u_2)@({\mathcal T}_1 \times {\mathcal T}_2)\;\; \circ \;\; (y_1,y_2)^\inv \;\; \circ \;\;
(g^*(u_1),h^*(u_2))@({\mathcal T}_1 \times {\mathcal T}_2)$$ is defined.  But this follows from the
definitions of $g^*$ and $h^*$.

Next we consider

~ \hfill
\unnamed {\ant{\Gamma
    \rdash\mathbf{Carrier}(\sigma,\gamma,f,(\lambda\;\intype{\alpha}{\sett}\;\tau[\alpha])) \doteq
    g} \ant{\Gamma
    \rdash\mathbf{Carrier}(\sigma,\gamma,f,(\lambda\;\intype{\alpha}{\sett}\;\kappa[\alpha])) \doteq
    h} \ant{\Gamma \rdash\intype{k}{\tau[\sigma] \rightarrow \kappa[\sigma]}} \ant{\Gamma \rdash
    \intype{a}{\tau[\sigma]}}} {\ant{\Gamma \rdash
    \mathbf{Carrier}(\sigma,\gamma,f,(\lambda\intype{\alpha}{\sett}\;\tau[\alpha] \rightarrow
    \kappa[\alpha]))(k)(g(a)) \doteq h(k(a))}} \hfill ~

\medskip
Consider $\rho$ in $\convalue{\Gamma}$ and let $\sigma^*$, $\gamma^*$, $f^*$, $\tau^*[s]$,
$\kappa^*[s]$, $g^*$ and $h^*$ be defined for $\rho$ as in the proof of the previous rule.  Let
$G^*$ be the unique functional from $\tau^*[\sigma^*] \rightarrow \kappa^*[\sigma^*]$ to
$\tau^*[\gamma^*] \rightarrow \kappa^*[\gamma^*]$ satisfying $G^*(k)(g^*(u)) = h^*(k(u))$.  We have
\begin{eqnarray*}
  & & \mathbf{Carrier}(\sigma,\gamma,f,(\lambda \intype{\alpha}{\sett}\;\tau[\alpha] \rightarrow \kappa[\alpha])) \\
  & = & C\left(\begin{array}{l} \tau^*[\sigma^*] \rightarrow \kappa^*[\sigma^*], \\
    (\tau^*[Y(f^*)] \rightarrow \kappa^*[Y(f^*)])^\inv, \\
    \tau^*[\gamma^*] \rightarrow \kappa^*[\gamma^*]\end{array}\right).
\end{eqnarray*}
Let ${\mathcal T}_1$ be a template satisfying $\intype{\tau^*[Y(f^*)]}{\setof({\mathcal T}_1)}$ and let
${\mathcal T}_2$ be a template satisfying $\intype{\kappa^*[Y(f^*)]}{\setof({\mathcal T}_2)}$.  We must show
that for $k \in \tau^*[\sigma^*] \rightarrow \kappa^*[\sigma^*]$ there exists $\tilde{k} \in
\tau^*[Y(f^*)] \rightarrow \kappa^*[Y(f^*)]$ such that
$$k@({\mathcal T}_1 \rightarrow {\mathcal T}_2) \;\; \circ \;\; \tilde{k}^\inv \;\; \circ \;\; G^*(k)@({\mathcal
  T}_1 \rightarrow {\mathcal T}_2)$$ is defined.  For two sets $s$ and $w$ we find it clearer here
to write $(a \mapsto b) \in (s \times \tau)$ as an alternative notation for $(a,b) \in (s \times \tau)$.
The above requirement is equivalent to the statement that for every input-output pair $u \mapsto v$ of $k$ there exists an input-output pair $(x \mapsto y) \in (\tau^*[Y(f^*] \times \kappa^*[Y(f^*)])$
such that
$$u \mapsto v@({\mathcal T}_1 \times {\mathcal T}_2) \;\; \circ \;\; (x \mapsto y)^\inv \;\; \circ \;\; (g^*(u) \mapsto h^*(v))@({\mathcal T}_1 \times {\mathcal T}_2)$$
is defined. But, as in the case of pairs, this follows from the definitions of $g^*$ and $h^*$.

Finally we consider

\centerline{
\unnamed
{\ant{\Gamma \vdash \intype{a,b}{S_{\intype{x\;}{\;\sigma}} \Phi[x]}}
  \ant{\Gamma \vdash a =_\sigma b}}
{\ant{\Gamma \vdash a =_{\left(S_{\intype{x\;}{\;\sigma}}\;\Phi[x]\right)}\;\;b}}
}

\medskip
Consider $\rho \in \convalue{\Gamma}$.  Let $a^*$, $b^*$, $\sigma^*$ and $\Phi^*[u]$ be defined as usual in terms of $\rho$.
We must show that there exists $z \in \sigma^*$ such that $\Phi^*[z]$ is true and such that $a^* \circ z^\inv \circ b^*$ is defined.
But the second antecedent implies that there exists $z \in \sigma^*$ such that $a^* \circ z^\inv \circ b^*$ is defined.
But we then have that $a^* \circ z^\inv \circ z$ is also defined and hence $a^* =_{\sigma^*} z$ and by the substitution of isomorphics we have $\Phi^*[z]$.


\begin{thebibliography}{1}

\bibitem{Brouwer}
L.E.J. Brouwer.
\newblock On the significance of the principle of excluded middle in
  mathematics, especially in function theory.
\newblock In J.~van Heijenoort, editor, {\em A Source Book in Mathematical
  Logic, 1879-1931}. Harvard University Press, 1977.

\bibitem{GRPD}
Martin Hofmann and Thomas Streicher.
\newblock The groupoid model refutes uniqueness of identity proofs.
\newblock In {\em Logic in Computer Science, 1994. LICS'94. Proceedings.,
  Symposium on}, pages 208--212. IEEE, 1994.

\bibitem{HOTT}
HoTT-Authors.
\newblock Homotopy type theory, univalent foundations of mathematics.
\newblock
  http://hottheory.files.wordpress.com/2013/03/hott-online-611-ga1a258c.pdf,
  2013.

\bibitem{SimpSets}
Chris Kapulkin, Peter~LeFanu Lumsdaine, and Vladimir Voevodsky.
\newblock The simpicial model of univalent foundations.
\newblock {\em CoRR}, abs/1211.2851, 2012.

\bibitem{MLTT1971}
Per Martin-L{\"o}f.
\newblock A theory of types, 1971.

\bibitem{Rota}
G.C. Rota.
\newblock {\em Indiscrete Thoughts}.
\newblock Birkhäuser Boston, Inc., 1997.

\end{thebibliography}
\end{document}